%% file: main.tex
\def \arxivversion {1}

\if \arxivversion 1
    \def \numsides {1}
    \def \showchanges {0}
    \documentclass{article}
    \usepackage[margin=0.680in]{geometry}
    \usepackage[utf8]{inputenc}
\else
    \def \numsides {2}
    \def \showchanges {1}
    \if \numsides 2
        \documentclass[draftcls,web]{ieeecolor}
    \else
        \documentclass[12pt,draftcls,onecolumn,web]{ieeecolor}
    \fi
    \usepackage{preamble/generic}
\fi

\input{preamble/utils}

\def\BibTeX{{\rm B\kern-.05em{\sc i\kern-.025em b}\kern-.08em
    T\kern-.1667em\lower.7ex\hbox{E}\kern-.125emX}}
\begin{document}

\if \arxivversion 1
    \title{An Error-Based Safety Buffer for Safe Adaptive Control (Extended Version)}
    \author{Peter A. Fisher, Johannes Autenrieb, and Anuradha M. Annaswamy}
    \date{}
\else
    \title{An Error-Based Safety Buffer for Safe Adaptive Control}
    \author{Peter A. Fisher, \IEEEmembership{Member, IEEE}, Johannes Autenrieb, \IEEEmembership{Member, IEEE}, and Anuradha M. Annaswamy, \IEEEmembership{Fellow, IEEE}
    \thanks{Manuscript received Jan, 2026. This work is supported by the Boeing Strategic University Initiative and by the Air Force Research Laboratory.}
    \thanks{DISTRIBUTION STATEMENT A. Approved for public release. Distribution is unlimited. AFRL-20XX-XXXX; Cleared XX/XX/XXXX.}
    \thanks{Peter A. Fisher is with the Mechanical Engineering Department at the Massachusetts Institute of Technology, Cambridge, MA 02139 USA (e-mail: pafisher@mit.edu).}
    \thanks{Johannes Autenrieb is with the Institute of Flight Systems at the German Aerospace Center and the Technical University of Braunschweig, Braunschweig, Germany (e-mail: 	johannes.autenrieb@dlr.de).}
    \thanks{Anuradha M. Annaswamy is with the Mechanical Engineering Department at the Massachusetts Institute of Technology, Cambridge, MA 02139 USA (e-mail: aanna@mit.edu).}}
\fi





\maketitle



\begin{abstract}
\input{Body/Abstract}
\end{abstract}

\if \arxivversion 0
\fi


\if \arxivversion 1
    \renewcommand{\thesection}{\Roman{section}}
    \renewcommand{\thesubsection}{\Roman{section}.\Alph{subsection}}
\fi

\section{Introduction} \label{sec:introduction}
\input{Body/Introduction_2}
\section{Preliminaries and Statement of the Problem} \label{sec:preliminaries_problem_statement}
\input{Body/Preliminaries}
\input{Body/Problem_Formulation}
\section{Safe Adaptive Control for Uncertain Unforced Dynamics} \label{sec:no_input_uncertainty}
\input{Body/No_Input_Uncertainty}
\section{Safe Adaptive Control with an Uncertain Input Matrix} \label{sec:input_uncertainty}
\input{Body/Input_Uncertainty}
\section{Simulations and Discussion} \label{sec:simulations}
\input{Body/Simulation.tex}
\section{Conclusions and Future Work} \label{sec:conclusions}
\input{Body/Conclusions.tex}

\if \arxivversion 1
\section*{Acknowledgment}
This work was supported by the Boeing Strategic University Initiative and by the Air Force Research Laboratory.
\fi


\if \arxivversion 1
    \appendix
    \renewcommand{\thesection}{Appendix \Roman{section}}
    \renewcommand{\thesubsection}{\Roman{section}.\Alph{subsection}}
\else
    \appendices
\fi

\input{Appendix/No_Input_Uncertainty}
\input{Appendix/Input_Uncertainty}
\if \arxivversion 1
    \input{Appendix/SMID}
    \input{Appendix/Simulation_Details}
\fi


\bibliographystyle{IEEEtran}
\bibliography{references/references2}


\if \arxivversion 0
\noindent\textbf{Peter Fisher} recently completed his PhD in the Active-Adaptive Control Laboratory at the Massachusetts Institute of Technology (MIT), Cambridge, MA, USA, and is now a Senior Scientist at STR (str.us). His research interests include safety-critical and constrained adaptive control, control barrier functions, adaptive optimal control, and discrete-time adaptive control and parameter learning.

\noindent\textbf{Johannes Autenrieb} is a Research Scientist with the Institute of Flight Systems, Department of Flight Dynamics and Simulation, at the German Aerospace Center (DLR), Braunschweig, Germany. Prior to completing his Ph.D. at the Technical University of Braunschweig, he was a Visiting Scholar at MIT in 2022-2023, where he worked on control barrier functions for adaptive and safety-critical flight control systems, and a Visiting Ph.D. Student at the University of Oxford, U.K., in 2024, focusing on robust online optimization algorithms for uncertain dynamical systems. His research interests include nonlinear guidance and control of aerospace systems, adaptive control theory, safety-critical control of autonomous systems, as well as multi-agent systems.

\noindent\textbf{Anuradha Annaswamy} is Founder and Director of the Active-Adaptive Control Laboratory in the Department of Mechanical Engineering at MIT. Her research interests span adaptive control theory and its applications to aerospace, automotive, propulsion, and energy systems as well as cyber physical systems such as Smart Grids, Smart Cities, and Smart Infrastructures. She has received best paper awards (Axelby, 1986; CSM, 2010), as well as Distinguished Member and Distinguished Lecturer awards from the IEEE Control Systems Society (CSS), best paper award from the IFAC journal Annual Reviews in Control for 2021-23, and a Presidential Young Investigator award from NSF, 1991-97. She is a Fellow of IEEE and International Federation of Automatic Control. She is the recipient of the Distinguished Alumni award from Indian Institute of Science. She received the IEEE Control Systems Technology Award from CSS in 2024.
\fi

\end{document}

%% file: preamble/utils.tex
\usepackage{amsmath}
\usepackage{amsfonts}
\usepackage{amssymb}
\if \arxivversion 1
    \usepackage{amsthm}
    \usepackage{mathabx}
\fi
\usepackage{mathtools}
\usepackage{graphicx}
\usepackage{float}
\usepackage{tabularx}
\usepackage[usenames,dvipsnames]{xcolor}
\usepackage{transparent}
\usepackage{algorithm}
\usepackage{algorithmic}
\usepackage{subcaption}
\usepackage{cite}
\usepackage{hyperref}
\hypersetup{hidelinks=true}
\usepackage{textcomp}

\renewcommand{\vec}{\mathbf}

\renewcommand{\bar}{\overline}
\renewcommand{\hat}{\widehat}
\renewcommand{\tilde}{\widetilde}

\DeclareMathOperator*{\argminbelow}{arg\,min}

\newcommand{\diag}{\mathrm{diag}}

\newcommand{\erf}{\mathrm{erf}}
\newcommand{\ubar}{\underline}
\newcommand{\bbR}{\mathbb{R}}

\newcommand{\bbP}{\mathbb{P}}

\newcommand{\calN}{\mathcal{N}}

\newcommand{\calK}{\mathcal{K}}
\newcommand{\calL}{\mathcal{L}}
\newcommand{\calC}{\mathcal{C}}

\if \arxivversion 0
    \newcommand{\qed}{$\ \Box$}
\fi

\newcommand{\indenti}[1]{\transparent{0}{#1}\transparent{1}}

\if \showchanges 1
    \definecolor{DarkRed}{RGB}{220,0,0}
    \newcommand{\majorchange}{\textcolor{DarkRed}}
    \newcommand{\minorchange}{\textcolor{blue}}
\else
    \newcommand{\majorchange}{}
    \newcommand{\minorchange}{}
\fi

\newtheorem{theorem}{Theorem}

\newtheorem{lemma}{Lemma}

\newtheorem{assumption}{Assumption}

\newtheorem{assumptionalt}{Assumption}[assumption]
\newenvironment{assumptionp}[1]{
  \renewcommand\theassumptionalt{#1}
  \assumptionalt
}{\endassumptionalt}
\newtheorem{definition}{Definition}
\newtheorem{remark}{Remark}

\newtheorem{problem}{Problem}
\newtheorem{property}{Property}

\newcommand{\remove}[1]{\iffalse {#1} \fi}

\allowdisplaybreaks

%% file: Body/Abstract.tex

We consider the problem of adaptive control of a class of feedback linearizable plants with matched parametric uncertainties whose states are accessible, subject to state constraints, which often arise due to safety considerations. In this paper, we combine adaptation and control barrier functions into a real-time control architecture that guarantees stability, ensures control performance, and remains safe even with parametric uncertainties. Two problems are considered, differing in the nature of the parametric uncertainties. In both cases, the control barrier function is assumed to have an arbitrary relative degree. In addition to guaranteeing stability, it is proved that both the control objective and the safety objective are met
\minorchange{with a buffer window that asymptotically decays to an arbitrarily small value.}
No excitation conditions are imposed on the command signal. Simulation results demonstrate the non-conservatism of all of the theoretical developments.

%% file: Body/Introduction_2.tex
Real time control of dynamic systems with parametric uncertainties so as to ensure high performance and safety is a challenging task.  The field of adaptive control has focused on providing real-time inputs for dynamic systems through parameter learning and control design using a stability framework \cite{Narendra05,Ioannou1996,Sastry_1989,Slotine1991,Krstic1995,Ast13}. While these works have accommodated parametric uncertainties, with extensions to magnitude and rate constraints on the input, consideration of constraints on the state, often necessitated by safety considerations, has been absent for the most part. In settings where the plant is fully known, Control Barrier Functions (CBFs) \cite{ames2019CBF,nguyen2016walkingCBFs,gunter2022automatedvehicles,alan2023automatedvehicles,pmlr-v229-zhang23h} have gained prominence as a minimally-invasive tool for enforcing state constraints that ensure safety. This paper proposes two new approaches for simultaneously realizing both a control objective and a safety objective in the presence of parametric uncertainties. The class of dynamic systems we focus on is feedback-linearizable and time-invariant, with states accessible, and our approaches combine an adaptive control architecture with a new CBF-based controller.

The main challenges in realizing the two objectives in a combined manner are to ensure that state constraints are satisfied during transients in the closed-loop adaptive system, and that the control objectives are fully met in steady-state. The inherent nonlinearity in the transients is the main difficulty in the first challenge; stringent conditions needed to reduce the effect of the parametric uncertainty cause the second challenge. Unlike the current literature, which address these two challenges only with compromises and conservatism, our approaches provably ensure that both of these challenges are fully overcome and constitute a clear advance of the state of the art in safe and stable adaptive control.

The state of the art in addressing problems where both parametric uncertainties and state constraints are present can be found in\minorchange{\cite{gurriet2018,xu2015CBFRobustness,kolathaya2019ISSf,alan2023PBF,ohnishi2019BarrierCertifiedAdaptiveRL,cheng2019SafeRL,fisac2019GeneralSafetyFramework,fan2020BayesianLearningAdaptive,pmlr-v120-taylor20a,wang2025OutputPositiveAdaptiveControl,alan2023DisturbanceObserverACBF,das2025RobustCBFUncertaintyEstimation,isaly2024AdaptiveSafetyRISE,sun2024DOBSCC,taylor2020aCBFs,lopez2021RaCBFs,lopez2023UnmatchedCBFs,wang2024AdaptiveCBFs,autenrieb2023EBR,solanocastellanos2025SafeFormationControl}}. Robust control designs have been proposed in \cite{gurriet2018,xu2015CBFRobustness,kolathaya2019ISSf,alan2023PBF} which offset the state constraints with buffer windows according to worst-case bounds on modeling uncertainties. The papers\minorchange{\cite{ohnishi2019BarrierCertifiedAdaptiveRL,cheng2019SafeRL,fisac2019GeneralSafetyFramework,fan2020BayesianLearningAdaptive,pmlr-v120-taylor20a,wang2025OutputPositiveAdaptiveControl}} have sought to learn from data to reduce conservatism. Disturbance observer-based approaches were proposed in \cite{alan2023DisturbanceObserverACBF,das2025RobustCBFUncertaintyEstimation,isaly2024AdaptiveSafetyRISE,sun2024DOBSCC} which treat the total effect of the parametric uncertainties on the dynamics as a state-dependent disturbance. However, robust methods can be very conservative depending on the size of the uncertainty, and data- and disturbance observer-based methods require either significant offline training, persistent excitation online, or high gains to avoid excessive conservatism.

Additionally, \cite{taylor2020aCBFs,lopez2021RaCBFs,lopez2023UnmatchedCBFs,wang2024AdaptiveCBFs} have explored solutions which take an online, adaptive approach to safety. In \cite{taylor2020aCBFs}, adaptive Control Barrier Functions (aCBFs) are proposed based on Control Lyapunov functions. Robust Adaptive Control Barrier Functions (RaCBFs) are proposed in \cite{lopez2021RaCBFs,lopez2023UnmatchedCBFs}, where the aCBF approach is altered to remove jitter and enable online reduction in conservatism through data-driven methods. The results of \cite{wang2024AdaptiveCBFs} extend aCBFs to include input certainties. The advantage of these methods over others mentioned in \cite{gurriet2018,xu2015CBFRobustness,kolathaya2019ISSf,alan2023PBF,ohnishi2019BarrierCertifiedAdaptiveRL,cheng2019SafeRL,fisac2019GeneralSafetyFramework,fan2020BayesianLearningAdaptive,pmlr-v120-taylor20a,wang2025OutputPositiveAdaptiveControl,alan2023DisturbanceObserverACBF,das2025RobustCBFUncertaintyEstimation,isaly2024AdaptiveSafetyRISE,sun2024DOBSCC} is that they can be less conservative, require no offline training, and avoid the need for high gains as in disturbance observers. Conservatism still remains which scales with the size of the uncertainty, however, \minorchange{and this conservatism can be removed only through online parameter learning, which requires persistent excitation in general.}

\minorchange{The adaptive controllers that we propose in this work apply conservatism only as needed during the transient period of adaptation, and ultimately remove conservatism without requirements on excitation, offline training time, or high gains.
These controllers build on our past work in \cite{autenrieb2023EBR,solanocastellanos2025SafeFormationControl}, both of which constructed a CBF-based modification of a reference input
which guaranteed stability and safety
assuming that the trajectories were within the safe set while adaptive control transients were present.}

\majorchange{In this paper, we introduce two new approaches
for safe adaptive control design, which
remove the assumption of safety during the transient period of adaptation which was required in our earlier work \cite{autenrieb2023EBR,solanocastellanos2025SafeFormationControl}, guarantee stability of all trajectories that start in the safe set, and ensure that the safety and control objectives are met. Our first proposed controller, denoted as an Error-Based Safety Buffer (EBSB), offsets the state constraint with a buffer that is proportional to the adaptive control tracking error. This buffer decays to a pre-determined small value as the tracking error goes to zero, which is proven to occur without the need for any form of excitation or high gains,
but with an assumption that the safe set is large relative to a ball whose diameter scales with the initial parametric uncertainty. Our second proposed controller, denoted as an Error-Based Safety Filter (EBSF), accommodates uncertainties in the input matrix as well as the drift dynamics. Additionally, while the aforementioned requirement on the safe set geometry is unavoidable for EBSB, it may require tight bounds on the parametric uncertainty for challenging safe set geometries. Thus, we design EBSF to avoid this geometric requirement.}

The contributions of this paper are as follows. Two approaches, EBSB and EBSF, are proposed for ensuring simultaneous realization of safety and control objectives for linear time-invariant plants with parametric uncertainties whose states are accessible are proposed, with EBSF applicable when input uncertainties are present. For both approaches, stability and safety are guaranteed, \minorchange{and buffer windows are guaranteed to asymptotically become arbitrarily small without requiring high gains or persistent excitation.} Furthermore, our designs accommodate CBFs with relative degree $r \geq 1$.


In Section \ref{sec:preliminaries_problem_statement}, we provide preliminaries and two problem statements. Section \ref{sec:no_input_uncertainty} presents EBSB, which addresses Problem \ref{prob:unforced_dynamics}, and Section \ref{sec:input_uncertainty} presents EBSF, which addresses Problem \ref{prob:input_matrix}.
Section \ref{sec:simulations} demonstrates the efficacy of EBSB and EBSF in simulation.
\if \arxivversion 1
    Proofs of all lemmas and theorems can be found in the Appendix, as well as additional supporting material.
\else
    Proofs of all lemmas and theorems can be found in the Appendix, and additional supporting material can be found in the arXiv version of our paper \cite{fisher2026EBSBarxiv}.
\fi

%% file: Body/Preliminaries.tex

\subsection{Definitions and Notation}

A few definitions are proposed in this section,
pertaining to dynamic systems of the form
\begin{equation} \label{eqn:prelim_plant}
    \dot{\vec{x}} = \vec{f}(\vec{x}) + G(\vec{x})\vec{u}
\end{equation}
where $\vec{x}(t) \in \bbR^n$, $\vec{u}(t) \in \bbR^m$, and $\vec{f}(\vec{x})$ and $G(\vec{x})$ are locally Lipschitz.
Consider a closed set $S \subset \bbR^n$ described by a smooth function $h(\vec{x})$ as follows:
\begin{subequations}
    \begin{align}
        S &= \{\vec{x} \in \bbR^n : h(\vec{x}) \geq 0\}, \label{eqn:safe_set_1} \\
        \partial S &= \{\vec{x} \in \bbR^n : h(\vec{x}) = 0\}, \label{eqn:safe_set_2} \\
        \mathrm{int}(S) &= \{\vec{x} \in \bbR^n : h(\vec{x}) > 0\} \label{eqn:safe_set_3}
    \end{align}
\end{subequations}
where $\partial S$ and $\mathrm{int}(S)$ denote the boundary and interior of $S$ respectively. When an objective of the control design is to choose $\vec{u}$ such that $\vec{x}(t) \in S\ \forall t \geq 0$, $S$ is commonly referred to as a safe set \cite{ames2019CBF}.

The following definitions from existing literature will be used in this paper:
\begin{definition}[$\calK_{\infty,e}$ \cite{ames2019CBF}] \label{def:k_infty_e}
    A continuous function $\alpha : \bbR \to \bbR$ is an \textit{extended class-$\calK_\infty$} function (denoted $\alpha \in \calK_{\infty,e}$) if it is strictly increasing, $\alpha(0) = 0$, $\lim_{z \to \infty} \alpha(z) = \infty$, and $\lim_{z \to -\infty} \alpha(z) = -\infty$.
\end{definition}
\begin{definition}[Control Barrier Function (CBF) \cite{ames2019CBF}] \label{def:CBF}
    Consider a smooth function $h : \bbR^n \to \bbR$ describing a set $S$ as in \eqref{eqn:safe_set_1}-\eqref{eqn:safe_set_3}. Then, $h$ is a \textit{control barrier function} for the dynamics in \eqref{eqn:prelim_plant} on $S$ if there exists an $\alpha \in \calK_{\infty,e}$ such that, for every $\vec{x} \in S$, there exists a $\vec{u} \in \bbR^m$ satisfying
    \begin{equation} \label{eqn:cbf_def}
        \calL_{\vec{f}}h(\vec{x}) + \calL_Gh(\vec{x})\vec{u} \geq -\alpha(h(\vec{x})),
    \end{equation}
    where $\calL_{\vec{f}}h$ and $\calL_Gh$ represent the Lie derivatives of $h$ with respect to $\vec{f}$ and $G$ respectively.
\end{definition}
\begin{definition}[Forward-Invariance \cite{ames2019CBF}] \label{def:forward_invariance}
    A control policy $\vec{u} = \vec{k}(\vec{x}, t)$ renders the dynamics in \eqref{eqn:prelim_plant} \textit{forward-invariant in $S$} if, for any $\vec{x}(0) \in S$, the policy results in $\vec{x}(t) \in S\ \forall t \geq 0$.
\end{definition}
\begin{definition}[Safety \cite{ames2019CBF}] \label{def:safety}
    A control policy $\vec{u} = \vec{k}(\vec{x}, t)$ renders the dynamics in \eqref{eqn:prelim_plant} \textit{safe with respect to $S$} if the policy renders the dynamics forward-invariant in $S$.
\end{definition}
\begin{definition}[Relative Degree \cite{xiao2019HOCBF}] \label{def:relative_degree}
    Consider a function $h : \bbR^n \to \bbR$ describing a set $S$ as in \eqref{eqn:safe_set_1}-\eqref{eqn:safe_set_3}. Then, $h$ has \textit{relative degree $r \geq 1$} with respect to the dynamics in \eqref{eqn:prelim_plant} if $\calL_Gh(\vec{x}) = \calL_G\calL_f^1h(\vec{x}) = \cdots = \calL_G\calL_f^{r-2}h(\vec{x}) = 0\ \forall \vec{x} \in \bbR^n$, and $\calL_G\calL_f^{r-1}h(\vec{x}) \neq 0$ for all $\vec{x} \in \bbR^n$ except on a measure-zero subset of $\bbR^n$.
\end{definition}

In addition, the following definitions are useful:
\minorchange{\begin{definition}[$\delta$-Strengthened CBF] \label{def:delta_strengthened_cbf}
    Consider a smooth function $h : \bbR^n \to \bbR$ describing a set $S$ as in \eqref{eqn:safe_set_1}-\eqref{eqn:safe_set_3}. Then, $h$ is a \textit{$\delta$-strengthened CBF} for the dynamics in \eqref{eqn:prelim_plant} on $S$ for a constant $\delta \geq 0$ if there exists an $\alpha \in \calK_{\infty,e}$ such that, for every $\vec{x} \in S$, there exists a $\vec{u} \in \bbR^m$ satisfying
    \begin{equation} \label{eqn:delta_cbf_def}
        \calL_{\vec{f}}h(\vec{x}) + \calL_Gh(\vec{x})\vec{u} \geq -\alpha(h(\vec{x}) - \delta).
    \end{equation}
\end{definition}}
\vspace{0.5\baselineskip}
\begin{definition}[Tracking] \label{def:tracking}
    Given two time-varying signals $\vec{y}(t), \vec{z}(t) \in \bbR^d$, $\vec{y}$ is said to \textit{asymptotically track} $\vec{z}$ with error $M$ if there exists a finite $M \geq 0$ such that $\limsup_{t \to \infty} \|\vec{y}(t) - \vec{z}(t)\| = M$.
\end{definition}

In addition to standard notations, the $d \times d$ identity matrix is denoted as $I_d$, and a $p \times q$ zero matrix is denoted as $0_{p \times q}$. For a vector $\vec{a} \in \bbR^d$, $\diag(\vec{a}) \in \bbR^{d \times d}$ is the diagonal matrix with the elements of $\vec{a}$ on the diagonal.

\subsection{CBF-Based Control Synthesis}

\minorchange{Given the plant in \eqref{eqn:prelim_plant}, a safe set $S \subset \bbR^n$ described by a CBF $h(\vec{x})$, and a nominal control policy $\vec{u}_*(\vec{x}, t)$, a standard approach to ensuring safety with respect to $S$ is to choose $\vec{u}$ according to the following quadratic program \cite{ames2019CBF}:
\begin{subequations} \label{eqn:QP-CBF}
    \begin{gather}
        \vec{u}(\vec{x}, t) = \argminbelow_{\vec{u} \in \bbR^m} \|\vec{u} - \vec{u}_*(\vec{x}, t)\|^2 \label{eqn:QP-CBF_cost} \\
        \mathrm{s.t.} \nonumber \\
        \calL_\vec{f}h(\vec{x}) + \calL_Gh(\vec{x})\vec{u} \geq -\alpha(h(\vec{x})). \label{eqn:QP-CBF_const}
    \end{gather}
\end{subequations}
By forming the Lagrangian of \eqref{eqn:QP-CBF} and solving for stationarity, one may straightforwardly show that \eqref{eqn:QP-CBF} is equivalent to
\begin{equation} \label{eqn:prelim_safe_controller}
    \vec{u}(\vec{x}, t) = \vec{u}_*(\vec{x}, t) + \mu(\vec{x}, t)(\calL_Gh(\vec{x}))^\top,
\end{equation}
\vspace{-\baselineskip}
\begin{subequations} \label{eqn:prelim_safe_lagrange_QP}
    \begin{gather}
        \mu(\vec{x}, t) = \argminbelow_{\mu \in \bbR} \mu^2 \label{eqn:prelim_safe_lagrange_QP_cost} \\
        \mathrm{s.t.} \nonumber \\
        \calL_\vec{f}h(\vec{x}) + \calL_Gh(\vec{x})\vec{u}_*(\vec{x}, t) + \mu\|\calL_Gh(\vec{x})\|^2 \geq -\alpha(h(\vec{x})). \label{eqn:prelim_safe_lagrange_QP_const}
    \end{gather}
\end{subequations}}
\if \showchanges 1
\vspace{-\baselineskip}
\fi

%% file: Body/Problem_Formulation.tex
\subsection{Statement of the Problem}

The class of dynamic systems we consider in this paper is of the form
\begin{equation} \label{eqn:general_plant}
    \dot{\vec{x}} = A\vec{x} + B(\Lambda\vec{u} - F(\vec{x})\theta_*)
\end{equation}
where the state $\vec{x}(t) \in \bbR^n$ is available for measurement, $\vec{u}(t) \in \bbR^m$ is the input, $A \in \bbR^{n \times n}$, $B \in \bbR^{n \times m}$, and $F : \bbR^n \to \bbR^{m \times p}$ are known, $F(\vec{x})$ is bounded for bounded $\vec{x}$, and $(A, B)$ is controllable. The system in \eqref{eqn:general_plant} also includes constant parametric uncertainties $\theta_*$ and $\Lambda$, where $\Lambda = \diag(\lambda_*)$ and $\theta_* \in \bbR^p$ and $\lambda_* \in \bbR^m$ are unknown.
In addition, a closed set $S$ as in \eqref{eqn:safe_set_1}-\eqref{eqn:safe_set_3} and smooth function $h : \bbR^n \to \bbR$ with relative degree $r \geq 1$ with respect to \eqref{eqn:general_plant} are assumed to be given. The problem addressed in this paper is the design of $\vec{u}$ in \eqref{eqn:general_plant} so that a safety objective and control objective are met, which are defined as follows: the \textit{safety objective} is to ensure that $\vec{x}(t) \in S\ \forall t \geq 0$. For the control objective, we define a desired trajectory $\vec{x}_*(t) \in \bbR^n$ as the solution to the dynamical system
\begin{equation} \label{eqn:reference_trajectory}
    \dot{\vec{x}}_* = A_m\vec{x}_* + B\vec{r}_*
\end{equation}
for any bounded reference input $\vec{r}_*(t) \in \bbR^m$, where $A_m$ is Hurwitz and $A_m = A + BK$ for some $K \in \bbR^{m \times n}$.
Then, the \textit{control objective} is to ensure that $\|\vec{x}(t) - \vec{x}_*(t)\|$ is bounded and that $\vec{x}$ asymptotically tracks $\vec{x}_*$ with an error $M$.
The overall problem is therefore the design of $\vec{u}$ in \eqref{eqn:general_plant} so as to meet the safety and control objectives simultaneously.

Our particular focus is on the following two problems that will be addressed in Sections \ref{sec:no_input_uncertainty} and \ref{sec:input_uncertainty} respectively:
\begin{problem}[Uncertain Unforced Dynamics] \label{prob:unforced_dynamics}
    Design an adaptive controller that meets the safety and control objectives for the plant in \eqref{eqn:general_plant} with $\Lambda$ known and $\theta_*$ unknown.
\end{problem}
\begin{problem}[Uncertain Input Matrix] \label{prob:input_matrix}
    Design an adaptive controller that meets the safety and control objectives for the plant in \eqref{eqn:general_plant} with both $\Lambda$ and $\theta_*$ unknown.
\end{problem}

\if \showchanges 1
\vspace{-\baselineskip}
\fi
\minorchange{\begin{remark}
    Problem \ref{prob:input_matrix} is a superset of Problem \ref{prob:unforced_dynamics}. However, we consider Problem \ref{prob:unforced_dynamics} first, as it lends itself to the design of a safe adaptive controller which has certain desirable properties that are more difficult to achieve in Problem \ref{prob:input_matrix}. For more details, see Section \ref{subsec:simulations_discussion}.
\end{remark}
\begin{remark}
    The parametric uncertainties considered in this work are matched, as the parametric uncertainties in \eqref{eqn:general_plant} lie in the span of $B$.
    Extension to unmatched uncertainties will be a focus of our future work. Unmatched uncertainties add challenge in the relative degree $r > 1$ setting, as taking time derivatives of $h$ may cause the unknown parameter to appear before the control input. Thus, this extension may necessitate output feedback-based tools including backstepping.
\end{remark}
\begin{remark}
    The plant in \eqref{eqn:general_plant} is a special case of the general control-affine plant in \eqref{eqn:prelim_plant} with $\vec{f}(\vec{x}) = A\vec{x} - BF(\vec{x})\theta_*$ and $G(\vec{x}) = B\Lambda$. Thus, for a function $h : \bbR^n \to \bbR$, $\calL_\vec{f}h(\vec{x}) = \frac{\partial h}{\partial\vec{x}}(A\vec{x} - BF(\vec{x})\theta_*)$ and $\calL_Gh(\vec{x}) = \frac{\partial h}{\partial\vec{x}}B\Lambda$.
\end{remark}}

%% file: Body/No_Input_Uncertainty.tex
In this section, we address Problem \ref{prob:unforced_dynamics}. We assume without loss of generality that $\Lambda = I_m$, allowing \eqref{eqn:general_plant} to be written as
\begin{equation} \label{eqn:plant_no_input_uncertainty}
    \dot{\vec{x}} = A\vec{x} + B(\vec{u} - F(\vec{x})\theta_*).
\end{equation}
In order to meet both the control and safety objectives, we introduce a reference model of the form
\begin{equation} \label{eqn:reference_model}
    \dot{\vec{x}}_m = A_m\vec{x}_m + B\vec{r}_s
\end{equation}
where $A_m$ is defined in \eqref{eqn:reference_trajectory} and $\vec{r}_s(t)$ is a bounded reference signal to be determined. In Section \ref{subsec:EBSB_adaptive_control}, we design $\vec{u}$ as an adaptive controller so that $\vec{x}$ asymptotically tracks $\vec{x}_m$. In Section \ref{subsec:EBSB}, we show that $\vec{r}_s$ can be designed using $h$ such that it renders both the reference model in \eqref{eqn:reference_model} and the overall adaptive system safe with respect to $S$, and such that $\vec{x}_m$ asymptotically tracks $\vec{x}_*$.

We introduce the following assumption on the nature of the parametric uncertainty:
\begin{assumption} \label{asn:theta_in_set}
    $\theta_* \in \Theta$ for a known compact, convex set $\Theta \subset \bbR^p$ with nonzero $p$-dimensional volume.
\end{assumption}

\subsection{Adaptive Control Design} \label{subsec:EBSB_adaptive_control}

Recalling that $A_m = A + BK$, the adaptive controller is chosen as
\begin{equation} \label{eqn:input}
    \vec{u} = K\vec{x} + \vec{r}_s + F(\vec{x})\hat{\theta},
\end{equation}
\minorchange{where $\hat{\theta}(t)$ is an estimate of $\theta_*$. Combining \eqref{eqn:plant_no_input_uncertainty}-\eqref{eqn:input} and defining $\vec{e}_x(t) := \vec{x}(t) - \vec{x}_m(t)$ and $\tilde{\theta}(t) := \hat{\theta}(t) - \theta_*$}, we obtain a standard error model \cite{Narendra05}:
\begin{equation} \label{eqn:error_model}
    \dot{\vec{e}}_x = A_m\vec{e}_x + BF(\vec{x})\tilde{\theta}.
\end{equation}
Finally, the adaptive law for updating $\hat{\theta}$ is given by
\minorchange{\begin{equation} \label{eqn:theta_adaptive_law}
    \dot{\hat{\theta}} = -\gamma_\theta F(\vec{x})^\top B^\top P\vec{e}_x,
\end{equation}}
where $\gamma_\theta > 0$ is any adaptive gain, and $P$ is the symmetric positive-definite solution to the Lyapunov equation
\begin{equation} \label{eqn:lyapunov_eqn}
    A_m^\top P + PA_m = -Q
\end{equation}
for any symmetric positive-definite matrix $Q$. \minorchange{We also define the following quantity for subsequent analysis:
\begin{equation} \label{eqn:E_1}
    E_1 := \sqrt{\frac{1}{\lambda_{min}(P)}\left(\vec{e}_x(0)^\top P\vec{e}_x(0) + \frac{1}{\gamma_\theta}\|\tilde{\theta}(0)\|^2\right)}
\end{equation}}

We establish key properties of the adaptive control design in the following theorem, which follows from standard arguments in adaptive control \cite{Narendra05}:
\begin{theorem} \label{thm:lyapunov_no_input_uncertainty}
    The closed-loop adaptive system consisting of \eqref{eqn:plant_no_input_uncertainty}-\eqref{eqn:theta_adaptive_law}
    has the following properties:
    \begin{enumerate}
        \item[(a)] \minorchange{$\hat{\theta}(t)$ is bounded}, and
        \item[(b)] $\|\vec{e}_x(t)\| \leq E_1\ \forall t \geq 0$
    \end{enumerate}
    \minorchange{for $E_1$ in \eqref{eqn:E_1}.}
    Additionally, if $\vec{r}_s(t)$ is bounded, then
    \begin{enumerate}
        \item[(b)] $\vec{x}(t)$ and $\vec{x}_m(t)$ are bounded, and
        \item[(c)] $\lim_{t \to \infty} \|\vec{e}_x(t)\| = 0$.
    \end{enumerate}
\end{theorem}

\subsection{Error-Based Safety Buffer} \label{subsec:EBSB}

We now design $\vec{r}_s$ to ensure simultaneous safety of both the plant and reference model \minorchange{by, in an intuitive sense, keeping the reference model farther from $\partial S$ than the error between the plant and reference model. We refer to this approach as rendering the reference model safe with an error-based safety buffer (EBSB).}
As a first step, as $h$ has relative degree $r \geq 1$ with respect to \eqref{eqn:plant_no_input_uncertainty}, we construct a sequence of high-order CBFs \cite{xiao2019HOCBF}
as follows:
\begin{subequations}
\begin{align}
    h_1(\vec{x}) &= h(\vec{x}) \label{eqn:hocbf_base_case} \\
    h_i(\vec{x}) &= \frac{\partial h_{i-1}}{\partial\vec{x}}\Big|_{\vec{x}}A\vec{x} + \alpha_{i-1}h_{i-1}(\vec{x})\ \forall i \in [2, r] \label{eqn:hocbf_recursion}
\end{align}
\end{subequations}
for any constants $\alpha_1, \dots, \alpha_{r-1} > 0$. Define $S_i \subset \bbR^n$ as the set described by $h_i$ as in \eqref{eqn:safe_set_1}-\eqref{eqn:safe_set_3}, and note that $S_1 = S$.
Note also that $h_r$ has relative degree 1 with respect to both the plant in \eqref{eqn:plant_no_input_uncertainty} and the reference model in \eqref{eqn:reference_model}. Furthermore, we construct a smooth augmented CBF candidate $\bar{h} : \bbR^n \to \bbR$ and corresponding set $\bar{S} \subset \bbR^n$ described by $\bar{h}$ as in \eqref{eqn:safe_set_1}-\eqref{eqn:safe_set_3} with the following properties:
\minorchange{\begin{property} \label{ppty:hbar}
$\bar{h}$ and $\bar{S}$ must satisfy the following:
\begin{enumerate}
    \item[(a)] $\bar{h}$ has relative degree 1 with respect to \eqref{eqn:plant_no_input_uncertainty} and \eqref{eqn:reference_model};
    \item[(b)] $\bar{S} \subseteq S_r$;
    \item[(c)] $\exists$ a closed, bounded set $\bar{\calC} = \{\vec{x} \in \bbR^n : b_i(\vec{x}) \geq 0\ \forall i \in [1, N_\calC]\}$, where each $b_i : \bbR^n \to \bbR$ is smooth and satisfies either (i) $b_i$ has relative degree $r_i = 1$ w.r.t \eqref{eqn:general_plant} and $b_i(\vec{x}) \geq 0\ \forall \vec{x} \in \bar{S}$, or (ii) $b_i$ has relative degree $r_i > 1$ w.r.t \eqref{eqn:general_plant} and $\frac{\partial b_i}{\partial\vec{x}}A\vec{x} \geq 0\ \forall \vec{x} \in \bar{\calC}$ where $b_i(\vec{x}) = 0$.
\end{enumerate}
\end{property}}
\begin{remark}
    \minorchange{In order for a function to be a CBF and thus useful in control design, it must have relative degree 1 with respect to the dynamics. This is the motivation for the HOCBF design of $h_r$ in \eqref{eqn:hocbf_base_case}-\eqref{eqn:hocbf_recursion}: $h_r$ has relative degree 1 with respect to the plant and reference model, and is constructed such that $\vec{x}(0) \in S_1 \cap \cdots \cap S_r$ and $\vec{x}(t) \in S_r\ \forall t \geq 0$ implies $\vec{x}(t) \in S_1 \cap \cdots \cap S_r\ \forall t \geq 0$. Then, $\bar{h}$ is an auxiliary CBF candidate constructed from $h_r$ such that it preserves the relative degree 1 property and ensures that $\vec{x} \in \bar{S} \implies \vec{x} \in S_r$. Property \ref{ppty:hbar}(c) of $\bar{h}$ is an additional requirement which ensures that rendering $\bar{S}$ forward-invariant also keeps the state bounded. As our subsequent control designs are nonlinear, (c) is necessary to ensure boundedness of the overall adaptive system.} A function $\bar{h}$ that satisfies Properties \ref{ppty:hbar}(a)-(c) can always be constructed (see Section \ref{subsec:constructing_hbar} for an example).
\end{remark}

\minorchange{We define the quantities $e_h(t) := \bar{h}(\vec{x}(t)) - \bar{h}(\vec{x}_m(t))$ and $\vec{g}(\vec{x}) := \nabla\bar{h}(\vec{x})$. We will require the following assumptions:}
\majorchange{\begin{assumption} \label{asn:hbar_CBF}
    There exist constants $\alpha_r, \delta, d_1 > 0$ such that $\bar{h}$ is a $(\delta + d_1)$-strengthened CBF for the reference dynamics in \eqref{eqn:reference_model} on $\bar{S}$ with the $\calK_{\infty,e}$ function $\alpha(z) = \alpha_rz$.
\end{assumption}
\begin{assumption} \label{asn:gradients_unopposed_at_boundary}
    There exists a constant $E_* > E_1$, where $E_1$ is given in \eqref{eqn:E_1}, such that every pair of points $\vec{x}_1, \vec{x}_2 \in S_1 \cap \cdots \cap S_{r-1} \cap \bar{S}$ with $\bar{h}(\vec{x}_1), \bar{h}(\vec{x}_2) \leq \delta$ and $\vec{g}(\vec{x}_1)^\top BB^\top\vec{g}(\vec{x}_2) = -\|B^\top\vec{g}(\vec{x}_1)\|\|B^\top\vec{g}(\vec{x}_2)\|$ satisfies $\|\vec{x}_1 - \vec{x}_2\| \geq E_*$.
\end{assumption}}

\minorchange{Now, analogous to \eqref{eqn:prelim_safe_controller}, we propose the following
design of the reference input $\vec{r}_s$ in \eqref{eqn:safe_reference_EBSB}-\eqref{eqn:z_p_ebsb}:}
\majorchange{\begin{equation} \label{eqn:safe_reference_EBSB}
    \vec{r}_s = \vec{r}_* + \mu_mB^\top\vec{g}(\vec{x}_m) + \mu_pB^\top\vec{g}(\vec{x}),
\end{equation}
with $\mu_m(\vec{x}, \vec{x}_m, t)$ and $\mu_p(\vec{x}, \vec{x}_m, t)$ chosen according to
\if \numsides 2
\begin{footnotesize}
\begin{subequations} \label{eqn:safe_lagrange_EBSB}
    \begin{gather}
        \mu_m(\vec{x}, \vec{x}_m, t), \mu_p(\vec{x}, \vec{x}_m, t) = \argminbelow_{\mu_m, \mu_p, \nu_m, \nu_p \in \bbR} J(\mu_m, \mu_p, \nu_m, \nu_p) \label{eqn:safe_lagrange_EBSB_cost} \\
        \mathrm{s.t.} \nonumber \\
        \rho_m + \max\{\bar{h}(\vec{x}_m) - \delta, 0\}\nu_m \geq \alpha_r\delta - z_m(\vec{x}_m, t), \label{eqn:safe_lagrange_EBSB_const1} \\
        \begin{aligned} \label{eqn:safe_lagrange_EBSB_const2}
            \bar{h}(\vec{x})(\rho_m + \max\{\bar{h}(\vec{x}), 0\}\nu_p) - e_h\rho_p
            \geq &-\bar{h}(\vec{x})z_m(\vec{x}_m, t) \\
            &+ \max_{\theta \in \Theta} e_hz_p(\vec{x}, t, \theta), 
        \end{aligned}
    \end{gather}
\end{subequations}
\end{footnotesize}
\else
\begin{subequations} \label{eqn:safe_lagrange_EBSB}
    \begin{gather}
        \mu_m(\vec{x}, \vec{x}_m, t), \mu_p(\vec{x}, \vec{x}_m, t) = \argminbelow_{\mu_m, \mu_p, \nu_m, \nu_p \in \bbR} J(\mu_m, \mu_p, \nu_m, \nu_p) \label{eqn:safe_lagrange_EBSB_cost} \\
        \mathrm{s.t.} \nonumber \\
        \rho_m + \max\{\bar{h}(\vec{x}_m) - \delta, 0\}\nu_m \geq \alpha_r\delta - z_m(\vec{x}_m, t), \label{eqn:safe_lagrange_EBSB_const1} \\
        \bar{h}(\vec{x})(\rho_m + \max\{\bar{h}(\vec{x}), 0\}\nu_p) - e_h\rho_p \geq -\bar{h}(\vec{x})z_m(\vec{x}_m, t) + \max_{\theta \in \Theta} e_hz_p(\vec{x}, t, \theta), \label{eqn:safe_lagrange_EBSB_const2}
    \end{gather}
\end{subequations}
\fi
where the dependencies of $\rho_m$ and $\rho_p$ on $\vec{x}$, $\vec{x}_m$, $\mu_m$, and $\mu_p$ are suppressed for ease of exposition, and
\if \numsides 2
\begin{footnotesize}
\fi
\begin{align}
    J&(\mu_m, \mu_p, \nu_m, \nu_p) = \mu_m^2 + \mu_p^2 + C_\nu(\nu_m^2 + (\nu_p - \nu_m)^2), \label{eqn:J} \\
    \rho_m&(\mu_m, \mu_p, \vec{x}, \vec{x}_m) = \mu_m\|B^\top\vec{g}(\vec{x}_m)\|^2 + \mu_p\vec{g}(\vec{x}_m)^\top BB^\top\vec{g}(\vec{x}), \label{eqn:rho_m} \\
    \rho_p&(\mu_m, \mu_p, \vec{x}, \vec{x}_m) = \mu_m\vec{g}(\vec{x}_m)^\top BB^\top\vec{g}(\vec{x}) + \mu_p\|B^\top\vec{g}(\vec{x})\|^2, \label{eqn:rho_p} \\
    z_m&(\vec{x}, t) = \alpha_r\bar{h}(\vec{x}) + \vec{g}(\vec{x})^\top(A_m\vec{x} + B\vec{r}_*(t)), \label{eqn:z_m_ebsb} \\
    z_p&(\vec{x}, t, \theta) = \alpha_r\bar{h}(\vec{x}) + \vec{g}(\vec{x})^\top(A_m\vec{x} + B(\vec{r}_*(t) + F(\vec{x})(\hat{\theta}(t) - \theta))) \label{eqn:z_p_ebsb}
\end{align}
\if \numsides 2
\end{footnotesize}
\fi
for any $C_\nu > 0$.} \minorchange{The constraint in \eqref{eqn:safe_lagrange_EBSB_const2} implicitly constructs an EBSB for the reference model, as will be made clear in Section \ref{subsec:EBSB_intuition}.}

We now present our main result for Problem \ref{prob:unforced_dynamics}. First, we define the following set:
\begin{equation} \label{eqn:initial_condition_set}
    \calC_0 = S_1 \cap \cdots \cap S_{r-1} \cap \bar{S} \cap \bar{\calC},
\end{equation}
where $\bar{\calC}$ is defined in the construction of $\bar{h}$. Then, our main result for Problem 1 is summarized in Theorem 2:
\begin{theorem} \label{thm:safety_objective_no_input_uncertainty}
    \majorchange{Let $\vec{x}_m(0), \vec{x}(0) \in \calC_0$ and $\bar{h}(\vec{x}_m(0)) \geq \delta$.} Then, for the closed-loop adaptive system consisting of \eqref{eqn:plant_no_input_uncertainty}-\eqref{eqn:z_p_ebsb} under \majorchange{Assumptions \ref{asn:theta_in_set}-\ref{asn:gradients_unopposed_at_boundary}}, the following results hold:
    \begin{enumerate}
        \item[(a)] \majorchange{the optimization problem in \eqref{eqn:safe_lagrange_EBSB} has a unique solution for all $\vec{x}, \vec{x}_m \in \calC_0$ such that $\bar{h}(\vec{x}_m) \geq \delta$,}
        \item[(b)] \majorchange{$\bar{h}(\vec{x}_m(t)) \geq \delta$ and $\bar{h}(\vec{x}(t)) \geq 0\ \forall t \geq 0$,}
        \item[(c)] $h(\vec{x}(t)) \geq 0\ \forall t \geq 0$,
        \item[(d)] all results of Theorem \ref{thm:lyapunov_no_input_uncertainty} hold,
        \item[(e)] \majorchange{there exists a finite time $T \geq 0$ such that the constraint in \eqref{eqn:safe_lagrange_EBSB_const2} is inactive for all $t \geq T$,} and
        \item[(f)] the safety and control objectives are met with $M$ independent of the parametric uncertainty.
    \end{enumerate}
\end{theorem}
\begin{remark}
    The initial condition requirement in Theorem \ref{thm:safety_objective_no_input_uncertainty} makes results (b) and (c) of Theorem \ref{thm:lyapunov_no_input_uncertainty} into domain of attraction type of results, as they depend on the size of the initial conditions. This is necessary because of our problem statement that requires both the control and safety objectives to be met simultaneously.
\end{remark}
\begin{remark}
    Our approach may be extended from a single state constraint to multiple in a straightforward manner using the softmin CBF approach in \cite{rabiee2024SoftminCBFs}.
\end{remark}
\majorchange{\begin{remark} \label{rmk:NLP}
    Assumption \ref{asn:gradients_unopposed_at_boundary} may be verified offline for a given $\delta > 0$ by solving the nonlinear program (NLP) $E_* = \min_{\vec{x}_1, \vec{x}_2 \in S_1 \cap \cdots \cap S_{r-1} \cap \bar{S}} \|\vec{x}_1 - \vec{x}_2\|$ s.t. $\vec{g}(\vec{x}_1)^\top BB^\top\vec{g}(\vec{x}_2) = -\|B^\top\vec{g}(\vec{x}_1)\|\|B^\top\vec{g}(\vec{x}_2)\|$, $\bar{h}(\vec{x}_1) \leq \delta$, $\bar{h}(\vec{x}_2) \leq \delta$.
    The value of $E_*$ obtained via this NLP provides an a priori bound on the largest permissible $E_1$.
\end{remark}
\begin{remark}
    Assumption \ref{asn:gradients_unopposed_at_boundary} is a geometric requirement on $\bar{S}$ for EBSB to achieve its goal of maintaining safety of both the plant and the reference model with respect to $\bar{S}$. In geometric terms, this assumption states that
    any two reachable states $\vec{x}_1, \vec{x}_2$ near the boundary of $\bar{S}$ with $B^\top\vec{g}(\vec{x}_1)$ and $B^\top\vec{g}(\vec{x}_2)$ being directly opposed must have $\|\vec{x}_1 - \vec{x}_2\| > E_1$.
    Thus, this assumption may not be satisfied if $\bar{S}$ contains a region which is narrower than a ball of diameter $E_1$, implying that the safe set geometry is too narrow for the worst-case a priori parametric uncertainty. If this assumption is not satisfied, then either (i) the narrow region must be removed from $\bar{S}$, e.g. via safe set intersections using $\mathrm{softmin}$ (see Section \ref{subsec:constructing_hbar}), or (ii) the state constraints must be relaxed to widen the narrow region of $\bar{S}$.
\end{remark}}
\if \showchanges 1
\vspace{-\baselineskip}
\fi
\minorchange{\begin{remark} \label{rmk:ebsb_conservatism_scales_with_parameter_set}
    As suggested by the maximization over $\Theta$ in \eqref{eqn:safe_lagrange_EBSB}, looser bounds on the uncertainty in Assumption \ref{asn:theta_in_set} will tend to result in choices of $\mu_m$ and $\mu_p$ that keep $\vec{x}$ and $\vec{x}_m$ farther from $\partial\bar{S}$, only during the transient period of adaptation. As the component of \eqref{eqn:safe_lagrange_EBSB} which requires a maximization over $\Theta$ is scaled by $e_h$, the conservatism resulting from the maximization vanishes in steady-state as $\vec{e}_x$ goes to zero.
\end{remark}
\begin{remark}
    Our approach of modifying the reference input $\vec{r}_s$ rather than the control input $\vec{u}$ to ensure safety takes inspiration from the field of reference governors (see e.g. \cite{garone2017ReferenceGovernors}). However, we apply a CBF-based approach to constructing $\vec{r}_s$ rather than a reference governor-based approach in this work, as the former is better suited to our problem statement.
\end{remark}}

\subsection{Intuition Behind EBSB and Assumption \ref{asn:gradients_unopposed_at_boundary}} \label{subsec:EBSB_intuition}

\majorchange{We now explain the reasoning behind $\vec{r}_s$ in \eqref{eqn:safe_reference_EBSB}-\eqref{eqn:z_p_ebsb}. First, we note that
$z_m(\vec{x}_m, t) + \rho_m = \frac{d}{dt}\bar{h}(\vec{x}_m) + \alpha_r\bar{h}(\vec{x}_m)$. The constraints in \eqref{eqn:safe_lagrange_EBSB_const1}-\eqref{eqn:safe_lagrange_EBSB_const2} can then be rewritten as
\begin{align}
    \frac{d}{dt}\bar{h}(\vec{x}_m) + \alpha_r\bar{h}(\vec{x}_m) &\geq \alpha_r\delta - \max\{\bar{h}(\vec{x}_m) - \delta, 0\}\nu_m \tag{\ref{eqn:safe_lagrange_EBSB_const1}'} \label{eqn:safe_lagrange_EBSB_const1_delta} \\
    \frac{d}{dt}\bar{h}(\vec{x}_m) + \alpha_r\bar{h}(\vec{x}_m) &\geq \alpha_r\Delta_{ebsb} - \max\{\bar{h}(\vec{x}), 0\}\nu_p \tag{\ref{eqn:safe_lagrange_EBSB_const2}'} \label{eqn:safe_lagrange_EBSB_const2_Delta_ebsb}
\end{align}
where
\begin{equation} \label{eqn:Delta_ebsb}
    \Delta_{ebsb} = \frac{e_h\rho_p + \max_{\theta \in \Theta} e_hz_p(\vec{x}, t, \theta)}{\alpha_r\bar{h}(\vec{x})}
\end{equation}
whenever $\bar{h}(\vec{x}) > 0$. It is immediate from \eqref{eqn:safe_lagrange_EBSB_const1_delta} that the constraint in \eqref{eqn:safe_lagrange_EBSB_const1} enforces $\bar{h}(\vec{x}_m(t)) \geq \delta\ \forall t \geq 0$.
Relative to \eqref{eqn:safe_lagrange_EBSB_const1_delta}, \eqref{eqn:safe_lagrange_EBSB_const2_Delta_ebsb} imposes additional conservatism in the form of a time-varying buffer $\Delta_{ebsb}(t)$, which ensures that
$\bar{h}(\vec{x}_m(t)) \geq |e_h(t)|$ at every time $t$ such that $\bar{h}(\vec{x}(t)) \leq \bar{h}(\vec{x}_m(t))$. That is, $\vec{r}_s$ constrains the reference model with this extra time-varying buffer to ensure that the closed-loop adaptive system is safe during the transient period when it is contracting towards the reference model.
It is important to note that $\vec{e}_x(t)$ and therefore $e_h(t)$ tends to zero asymptotically, which follows from Theorem \ref{thm:safety_objective_no_input_uncertainty}(d). This implies that
there exists a finite time $T$ beyond which $\Delta_{ebsb}(t) \leq \delta\ \forall t \geq 0$. Once this happens, the constraint in \eqref{eqn:safe_lagrange_EBSB_const2} becomes inactive, and
as \eqref{eqn:safe_reference_EBSB} and \eqref{eqn:safe_lagrange_EBSB_cost}-\eqref{eqn:safe_lagrange_EBSB_const1} are independent of the parametric uncertainty, EBSB is non-conservative for all $t \geq T$.
The reader is referred to Appendix \ref{app:safety_objective_no_input_uncertainty} for the proof of Theorem \ref{thm:safety_objective_no_input_uncertainty}.}


\majorchange{While $\mu_m$ and $\mu_p$ are the primary decision variables in \eqref{eqn:safe_lagrange_EBSB} that produce the safe reference input $\vec{r}_s$ in \eqref{eqn:safe_reference_EBSB}, \eqref{eqn:safe_lagrange_EBSB_const1} and \eqref{eqn:safe_lagrange_EBSB_const2} may be understood as constraints on the transformed decision variables $\rho_m$ and $\rho_p$, with the transformation given in \eqref{eqn:rho_m}-\eqref{eqn:rho_p}. This transformation need not always be non-singular; in particular, the slack variables $\nu_m$ and $\nu_p$ allow \eqref{eqn:safe_lagrange_EBSB} to be feasible whenever $\bar{h}(\vec{x}_m) > \delta$ and $\bar{h}(\vec{x}) > 0$ regardless of the invertibility of \eqref{eqn:rho_m}-\eqref{eqn:rho_p}. However, in order to ensure that the slack variables do not interfere with safety, they must vanish from \eqref{eqn:safe_lagrange_EBSB} as $\bar{h}(\vec{x}_m) \to \delta$ or $\bar{h}(\vec{x}) \to 0$. Thus, whenever $\bar{h}(\vec{x}_m) = \delta$, $\bar{h}(\vec{x}) = 0$, and both constraints in \eqref{eqn:safe_lagrange_EBSB_const1}-\eqref{eqn:safe_lagrange_EBSB_const2} are active, we require
Assumption \ref{asn:gradients_unopposed_at_boundary} to ensure that the transformation in \eqref{eqn:rho_m}-\eqref{eqn:rho_p} is non-singular.
Furthermore, consider any two safe states $\vec{x}$ and $\vec{x}_m$ where $\bar{h}(\vec{x}_m) = \delta$ and $\bar{h}(\vec{x}) = 0$. From Theorem \ref{thm:lyapunov_no_input_uncertainty}, we must have $\|\vec{x} - \vec{x}_m\| \leq E_1$. Therefore, Assumption \ref{asn:gradients_unopposed_at_boundary} ensures non-singularity by stating that the transformation may only be non-singular for pairs of safe states near $\partial\bar{S}$ separated by a distance greater than $E_1$.}

\minorchange{The choice of $\delta$ represents a permissible tracking error in the event that $\vec{x}_*(t)$ lies in $\partial\bar{S}$; that is, it is allowable to have $\bar{h}(\vec{x}(t)) \geq \delta$ while $\bar{h}(\vec{x}_*(t)) = 0$.
While our analysis does not place a positive lower bound on $\delta$, it is lower-bounded in practice by the sampling period: in a sampled-data system, $\delta$ must be sufficiently small that $\bar{h}(\vec{x})$ cannot decrease by more than $\delta$ within a single sampling period.}

\subsection{Constructing the Augmented CBF Candidate} \label{subsec:constructing_hbar}

Here, we provide an example of how to construct a function $\bar{h}$
\minorchange{with the properties in Property \ref{ppty:hbar}.}
Choose any desired element-wise upper and lower bounds on each state: without loss of generality, choose constants $\bar{x}_i, \ubar{x}_i \in \bbR, i \in [1, n]$ with $\bar{x}_i > \ubar{x}_i\ \forall i$. For each state, define corresponding CBF candidates and HOCBFs
\begin{subequations}
    \begin{align}
        \bar{b}_{i,1}(\vec{x}) &= c_i(\bar{x}_i - x_i), \\
        \ubar{b}_{i,1}(\vec{x}) &= c_i(x_i - \ubar{x}_i), \\
        \bar{b}_{i,j}(\vec{x}) &= \frac{\partial\bar{b}_{i,(j-1)}}{\partial\vec{x}}A\vec{x} + \alpha_{i,(j-1)}\bar{b}_{i,(j-1)}(\vec{x})\ \forall i \in [2, r_i], \\
        \ubar{b}_{i,j}(\vec{x}) &= \frac{\partial\ubar{b}_{i,(j-1)}}{\partial\vec{x}}A\vec{x} + \alpha_{i,(j-1)}\ubar{b}_{i,(j-1)}(\vec{x})\ \forall i \in [2, r_i]
    \end{align}
\end{subequations}
for any constants $\alpha_{i,j} > 0$, where $r_i$ is the relative degree of $\bar{b}_{i,1}$ and $\ubar{b}_{i,1}$ with respect to the plant in \eqref{eqn:plant_no_input_uncertainty}.
Then, the augmented CBF candidate is given by
\if \numsides 2
\begin{equation} \label{eqn:hbar}
    \begin{aligned}
        \bar{h}(\vec{x}) = \mathrm{softmin}_\kappa(&h_r(\vec{x}), \bar{b}_{1,r_1}(\vec{x}), \dots, \bar{b}_{n,r_n}(\vec{x}), \\
        &\ubar{b}_{1,r_1}(\vec{x}), \dots, \ubar{b}_{n,r_n}(\vec{x})),
    \end{aligned}
\end{equation}
\else
\begin{equation} \label{eqn:hbar}
    \bar{h}(\vec{x}) = \mathrm{softmin}_\kappa(h_r(\vec{x}), \bar{b}_{1,r_1}(\vec{x}), \dots, \bar{b}_{n,r_n}(\vec{x}), \ubar{b}_{1,r_1}(\vec{x}), \dots, \ubar{b}_{n,r_n}(\vec{x})),
\end{equation}
\fi
where $\mathrm{softmin}$ is defined as in \cite{rabiee2024SoftminCBFs} for any choice of $\kappa > 0$ such that $\bar{S}$ is non-empty.
Furthermore, define
\begin{equation} \label{eqn:Cbar}
    \bar{\calC} = \{\vec{x} \in \bbR^n : \bar{b}_{i,j}(\vec{x}), \ubar{b}_{i,j}(\vec{x}) \geq 0\ \forall i, j\},
\end{equation}
which is bounded by construction. Using Proposition 2 in \cite{rabiee2024SoftminCBFs},
\minorchange{it is straightforward to show that $\bar{h}$ in \eqref{eqn:hbar} and $\bar{\calC}$ in \eqref{eqn:Cbar} satisfy Properties \ref{ppty:hbar}(a)-(c).}

%% file: Body/Input_Uncertainty.tex
In this section, we address Problem \ref{prob:input_matrix}. Accounting for the input uncertainty $\Lambda$ in \eqref{eqn:general_plant} adds further difficulty in meeting the safety objective.
As will be shown in Section \ref{subsec:EBSF}, this added difficulty will be overcome using an Error-Based Safety Filter (EBSF), \minorchange{which eliminates Assumption \ref{asn:gradients_unopposed_at_boundary} and accommodates the uncertainty $\Lambda$,} rather than an EBSB as in Section \ref{sec:no_input_uncertainty}.

In addition to Assumption \ref{asn:theta_in_set}, the parametric uncertainty is assumed to satisfy the following assumption:
\begin{assumption} \label{asn:lambda_in_set}
    $\Lambda = \diag(\lambda_*)$, where $\lambda_* \in L := \{\lambda \in \bbR^m : \ubar{\lambda}_i \leq \lambda_i \leq \bar{\lambda}_i\ \forall i\}$ for known constants $\bar{\lambda}_i > \ubar{\lambda}_i > 0$.
\end{assumption}

\subsection{Adaptive Control Design}

\minorchange{First, define $\vartheta_* := \lambda_*^{-1}$ where the inverse is taken element-wise, and define the set $L^{-1} := \{\vartheta \in \bbR^m : \frac{1}{\bar{\lambda}_i} \leq \vartheta_i \leq \frac{1}{\ubar{\lambda}_i}\ \forall i\}$. Evidently, $\vartheta_* \in L^{-1}$.} Then, the adaptive control input is designed in a straightforward manner as
\minorchange{\begin{subequations}
\begin{gather}
    \vec{u} = \diag(\hat{\vartheta})\phi(\vec{x}, \vec{r}_s, \hat{\theta}), \label{eqn:input_input_uncertainty_1} \\
    \phi(\vec{x}, \vec{r}_s, \hat{\theta}) = (K\vec{x} + \vec{r}_s + F(\vec{x})\hat{\theta}), \label{eqn:input_input_uncertainty_2}
\end{gather}
\end{subequations}
where $\hat{\theta}(t)$ is an estimate of $\theta_*$ and $\hat{\vartheta}(t)$ is an estimate of $\vartheta_*$.} Combining \eqref{eqn:general_plant}, \eqref{eqn:reference_model}, and \eqref{eqn:input_input_uncertainty_1}-\eqref{eqn:input_input_uncertainty_2}, and defining \minorchange{$\vec{e}_x(t) := \vec{x}(t) - \vec{x}_m(t)$, $\tilde{\theta}(t) := \hat{\theta}(t) - \theta_*$ and $\tilde{\vartheta}(t) := \hat{\vartheta}(t) - \vartheta_*$,} we obtain a standard error model:
\minorchange{\begin{equation} \label{eqn:error_model_input_uncertainty}
    \dot{\vec{e}}_x = A_m\vec{e}_x + B(F(\vec{x})\tilde{\theta} + \Lambda\diag(\tilde{\vartheta})\phi).
\end{equation}}
Finally, the adaptive laws for updating $\hat{\theta}$ and $\hat{\lambda}$ are given by \minorchange{\eqref{eqn:theta_adaptive_law} and
\begin{subequations}
    \begin{align}
        \bar{\vartheta} &= -\gamma_\vartheta\diag(\phi)B^\top P\vec{e}_x, \label{eqn:vartheta_adaptive_law_1} \\
        \dot{\hat{\vartheta}}_i &= \begin{cases} 0, & (\hat{\vartheta}_i \geq \frac{1}{\ubar{\lambda}_i} \land \bar{\vartheta}_i \geq 0) \lor (\hat{\vartheta}_i \leq \frac{1}{\bar{\lambda}}_i \land \bar{\vartheta}_i \leq 0) \\ \bar{\vartheta}_i, & \mathrm{otherwise} \end{cases}\ \forall i \label{eqn:vartheta_adaptive_law_2}
    \end{align}
\end{subequations}}
where $\gamma_\vartheta > 0$ is any adaptive gain, and $P$ is the solution to \eqref{eqn:lyapunov_eqn} for any symmetric positive-definite matrix $Q$.

We establish key properties of the adaptive control design in the following theorem, which follows from standard arguments in adaptive control \cite{Narendra05}:
\begin{theorem} \label{thm:lyapunov_input_uncertainty}
    The closed-loop adaptive system consisting of \eqref{eqn:general_plant}, \eqref{eqn:reference_model}, \eqref{eqn:theta_adaptive_law}, and \eqref{eqn:input_input_uncertainty_1}-\eqref{eqn:vartheta_adaptive_law_2} under Assumption \ref{asn:lambda_in_set} with \minorchange{$\hat{\vartheta}(0) \in L^{-1}$} has the following properties:
    \begin{enumerate}
        \item[(a)] \minorchange{$\hat{\vartheta}(t) \in L^{-1}\ \forall t \geq 0$,}
        \item[(b)] \minorchange{$\hat{\theta}(t)$ is bounded,} and
        \item[(c)] \minorchange{$\vec{e}_x(t)$ is bounded.}
    \end{enumerate}
    Additionally, if $\vec{r}_s(t)$ is bounded, then
    \begin{enumerate}
        \item[(d)] $\vec{x}(t)$ and $\vec{x}_m(t)$ are bounded, and
        \item[(e)] $\lim_{t \to \infty} \|\vec{e}_x(t)\| = 0$.
    \end{enumerate}
\end{theorem}

\subsection{Error-Based Safety Filter} \label{subsec:EBSF}

As in Section \ref{subsec:EBSB}, we now design $\vec{r}_s$ to simultaneously accomplish the control and safety objectives. As $h$ has relative degree $r \geq 1$ with respect to \eqref{eqn:general_plant}, we define a sequence of high-order CBFs \cite{xiao2019HOCBF} $h_1, \dots, h_r : \bbR^n \to \bbR$ as in \eqref{eqn:hocbf_base_case}-\eqref{eqn:hocbf_recursion} for any constants $\alpha_1, \dots, \alpha_{r-1} > 0$. Define $S_i \subset \bbR^n$ as the set described by $h_i$ as in \eqref{eqn:safe_set_1}-\eqref{eqn:safe_set_3}, and note that $S_1 = S$. Note also that $h_r$ has relative degree 1 with respect to the plant in \eqref{eqn:plant_no_input_uncertainty}.
Additionally, we construct an augmented CBF candidate $\bar{h} : \bbR^n \to \bbR$ and corresponding set $\bar{S} \subset \bbR^n$ described by $\bar{h}$ as in \eqref{eqn:safe_set_1}-\eqref{eqn:safe_set_3} with
\minorchange{Properties \ref{ppty:hbar}(a)-(c) as in Section \ref{subsec:EBSB}.}
Then, using the constants $\alpha_r, \delta > 0$ from Assumption \ref{asn:hbar_CBF} and defining $\vec{g}(\vec{x}) := \nabla\bar{h}(\vec{x})$ for convenience,
we propose the following choice of $\vec{r}_s$, which we refer to as an Error-Based Safety Filter (EBSF):
\majorchange{\begin{equation} \label{eqn:safe_reference_EBSF}
    \vec{r}_s = \vec{r}_* + \mu B^\top\vec{g}(\vec{x}),
\end{equation}
where $\mu(\vec{x}, \vec{x}_m, t)$ is chosen according to
\if \numsides 2
\begin{subequations} \label{eqn:safe_lagrange_EBSF}
    \begin{gather}
        \mu(\vec{x}, \vec{x}_m, t) = \argminbelow_{\mu \in \bbR} \mu^2 \label{eqn:safe_lagrange_EBSF_cost} \\
        \mathrm{s.t.} \nonumber \\
        \begin{gathered} \label{eqn:safe_lagrange_EBSF_const}
            \mu\min_{\lambda \in L} \vec{g}(\vec{x})^\top B((1 - \beta)I_m + \beta\diag(\lambda)\diag(\hat{\vartheta}))B^\top\vec{g}(\vec{x}) \\
            \begin{aligned}
                &\geq \alpha_r\delta - (1 - \beta)z_m(\vec{x}, t) - \beta\min_{\theta \in \Xi, \lambda \in L} z_p(\vec{x}, t, \theta, \lambda),
            \end{aligned}
        \end{gathered}
    \end{gather}
\end{subequations}
\else
\begin{subequations} \label{eqn:safe_lagrange_EBSF}
    \begin{gather}
        \mu(\vec{x}, \vec{x}_m, t) = \argminbelow_{\mu \in \bbR} \mu^2 \label{eqn:safe_lagrange_EBSF_cost} \\
        \mathrm{s.t.} \nonumber \\
        \mu\min_{\lambda \in L} \vec{g}(\vec{x})^\top B((1 - \beta)I_m + \beta\diag(\lambda)\diag(\hat{\vartheta}))B^\top\vec{g}(\vec{x}) \geq \alpha_r\delta - (1 - \beta)z_m(\vec{x}, t) - \beta\min_{\theta \in \Xi, \lambda \in L} z_p(\vec{x}, t, \theta, \lambda), \label{eqn:safe_lagrange_EBSF_const}
    \end{gather}
\end{subequations}
\fi
$\Xi \subset \bbR^p$ is the set $\Xi := \Theta \cup \{\hat{\theta}\}$, $z_m$ is defined in \eqref{eqn:z_m_ebsb}, and $z_p$ is given by
\if \numsides 2
\begin{footnotesize}
\begin{equation} \label{eqn:z_p_ebsf}
    \begin{aligned}
        z_p(\vec{x}, t, \theta, \lambda) &= \alpha_r\bar{h}(\vec{x}) + \vec{g}(\vec{x})^\top(A\vec{x} - BF(\vec{x})\theta) \\
        &\indenti{=} + \vec{g}(\vec{x})^\top B\diag(\lambda)\diag(\hat{\vartheta}(t))(K\vec{x} + \vec{r}_*(t) + F(\vec{x})\hat{\theta}(t)).
    \end{aligned}
\end{equation}
\end{footnotesize}
\else
\begin{equation} \label{eqn:z_p_ebsf}
    z_p(\vec{x}, t, \theta, \lambda) = \alpha_r\bar{h}(\vec{x}) + \vec{g}(\vec{x})^\top(A\vec{x} - BF(\vec{x})\theta) + \vec{g}(\vec{x})^\top B\diag(\lambda)\diag(\hat{\vartheta}(t))(K\vec{x} + \vec{r}_*(t) + F(\vec{x})\hat{\theta}(t)).
\end{equation}
\fi
$\beta : \bbR^n \times \bbR_{\geq 0} \to [0, 1]$ is an error-based interpolation parameter given by
\begin{equation} \label{eqn:ebsf_interpolation}
    \beta(\vec{x}, \omega) = \mathrm{sat}\left(\frac{\max\{\omega, \frac{2\delta}{3}\} - \bar{h}(\vec{x})}{\max\{\omega, \frac{2\delta}{3}\} - \frac{\delta}{3}}\right),
\end{equation}
where $\mathrm{sat}$ is the saturation function given by
\begin{equation} \label{eqn:sat}
    \mathrm{sat}(z) = \begin{cases} 0, & z \leq 0 \\ z, & 0 < z < 1 \\ 1, & z \geq 1 \end{cases}
\end{equation}
and $\omega(t)$ is any bounded continuous time-varying signal with the following properties:
\begin{enumerate}
    \item[(a)] $\omega(t) \geq 0\ \forall t \geq 0$, and
    \item[(b)] $\|\vec{e}_x\| \to 0 \implies \omega \to 0$.
\end{enumerate}}

We now present our main result for Problem \ref{prob:input_matrix}:
\begin{theorem} \label{thm:safety_objective_input_uncertainty}
    For $\calC_0$ defined in \eqref{eqn:initial_condition_set}, let \minorchange{$\vec{x}(0) \in \calC_0$ and $\hat{\vartheta}(0) \in L^{-1}$}. Then, for the closed-loop adaptive system consisting of \eqref{eqn:general_plant}, \eqref{eqn:reference_model}, \eqref{eqn:theta_adaptive_law}, \eqref{eqn:hocbf_base_case}-\eqref{eqn:hocbf_recursion}, \eqref{eqn:z_m_ebsb} and \eqref{eqn:input_input_uncertainty_1}-\eqref{eqn:sat} under \majorchange{Assumptions \ref{asn:theta_in_set}, \ref{asn:hbar_CBF}, and \ref{asn:lambda_in_set}}, the following results hold:
    \begin{enumerate}
        \item[(a)] \majorchange{the optimization problem in \eqref{eqn:safe_lagrange_EBSF} has a unique solution for all $\vec{x} \in \bar{S} \cap \bar{\calC}$,}
        \item[(b)] \majorchange{$\bar{h}(\vec{x}(t)) \geq 0\ \forall t \geq 0$,}
        \item[(c)] $h(\vec{x}(t)) \geq 0\ \forall t \geq 0$,
        \item[(d)] all results of Theorem \ref{thm:lyapunov_input_uncertainty} hold,
        \item[(e)] there exists a finite time $T \geq 0$ such that $\beta(\vec{x}(t), \omega(t)) = 0\ \forall t \geq T$, and
        \item[(f)] the safety and control objectives are met with $M$ independent of the parametric uncertainty.
    \end{enumerate}
\end{theorem}
\if \showchanges 1
\vspace{-\baselineskip}
\fi
\minorchange{\begin{remark} \label{rmk:ebsf_conservatism_scales_with_parameter_set}
    Similarly to Remark \ref{rmk:ebsb_conservatism_scales_with_parameter_set} regarding EBSB, looser bounds on the uncertain parameters in Assumptions \ref{asn:theta_in_set} and \ref{asn:lambda_in_set} will tend to result in more conservative choices of $\mu$ during the transient period of adaptation. However, because all components of EBSF which require minimizations over $\Theta$ and $L$ are scaled by $\beta$, and $\beta \to 0$ as $\vec{e}_x \to 0$ as proven in Theorem \ref{thm:safety_objective_input_uncertainty}, the conservatism resulting from the minimizations vanishes as $\vec{e}_x$ goes to zero.
\end{remark}}
\begin{remark}
    It should be noted that the use of the word Filter is not in reference to the frequency response of \eqref{eqn:safe_reference_EBSF}-\eqref{eqn:safe_lagrange_EBSF}. Rather, we refer to \eqref{eqn:safe_reference_EBSF}-\eqref{eqn:safe_lagrange_EBSF} as an Error-Based Safety Filter as it tends to filter out unsafe reference inputs, in parity with the use of the term Safety Filter elsewhere in the CBF literature (e.g. \cite{ames2019CBF}).
\end{remark}

\subsection{Intuition Behind EBSF} \label{subsec:ebcg_intuition}

\majorchange{In Problem \ref{prob:input_matrix}, which includes input uncertainties, we apply EBSF and not EBSB. In Problem \ref{prob:unforced_dynamics}, our EBSB approach designs $\vec{r}_s$ to render both the plant and reference model safe with respect to $\bar{S}$ simultaneously. Such an $\vec{r}_s$ exists by virtue of Assumption \ref{asn:gradients_unopposed_at_boundary}. In contrast, in Problem \ref{prob:input_matrix}, the uncertainty $\Lambda$ prevents such an $\vec{r}_s$ from existing in general. As a result, in EBSF, we only seek to render the plant safe without regard to the safety of the reference model. This also has an added advantage of totally avoiding the need for Assumption \ref{asn:gradients_unopposed_at_boundary}.}

\majorchange{Instead of choosing $\vec{r}_s$ to enforce an error-based safety buffer, EBSF chooses $\vec{r}_s$ according to \eqref{eqn:safe_reference_EBSF} - as one would choose $\vec{r}_s$ if all parameters were known - and considers two possibilities for $\mu$.} The first possibility assumes that the adaptive controller in \eqref{eqn:input_input_uncertainty_1}-\eqref{eqn:input_input_uncertainty_2} has accomplished its task and that the plant is exhibiting the desired closed-loop behavior - i.e. that $\dot{\vec{x}} = A_m\vec{x} + B\vec{r}_s$ - and chooses an ideal $\mu$ which will ensure $\frac{d}{dt}\bar{h}(\vec{x}) \geq -\alpha_r(\bar{h}(\vec{x}) - \delta)$ for the desired closed-loop plant dynamics. This first possibility is captured in $z_m$ in \eqref{eqn:z_m_ebsb}. The second possibility for $\mu$ makes no assumptions on the completeness of the adaptation, but rather uses prior knowledge of the parameter sets $\Theta$ and $L$ to choose a conservative $\mu$ which will ensure that the true plant in \eqref{eqn:general_plant} satisfies $\frac{d}{dt}\bar{h}(\vec{x}) \geq -\alpha_r(\bar{h}(\vec{x}) - \delta)$ for any $\theta_* \in \Theta$ and $\lambda_* \in L$. This second possibility is captured in $z_p$ in \eqref{eqn:z_p_ebsf}. Given these two possible choices of $\mu$, then, the EBSF approach is to apply the ideal $\mu$ when $\vec{x}$ is far away from $\partial\bar{S}$, to apply the worst-case $\mu$ when $\vec{x}$ is close to $\partial\bar{S}$, and to linearly interpolate between the two choices of $\mu$ at intermediate distances from $\partial\bar{S}$. Concretely, at every time $t$, we establish an interpolation window around $\partial\bar{S}$ consisting of \majorchange{$\bar{h}(\vec{x}) \in (\frac{\delta}{3}, \max\{\omega(t), \frac{2\delta}{3}\})$}. For $\bar{h}(\vec{x}) \leq \frac{\delta}{3}$, we apply the worst-case choice of $\mu$; for $\bar{h}(\vec{x}) \geq \max\{\omega(t), \frac{2\delta}{3}\}$, we apply the ideal choice of $\mu$; and for $\vec{x}$ in the interpolation window, we linearly interpolate between the two choices of $\mu$. \majorchange{$\omega(t)$ is a hand-tuned time-varying signal which controls how the width of the interpolation window varies as $\|\vec{e}_x\|$ varies}, and it ensures that in the limit as $\|\vec{e}_x\|$ goes to zero and the plant follows the reference model dynamics exactly, the width of the interpolation window shrinks to a value sufficiently small that the plant will never enter the window.

The results in Theorem \ref{thm:safety_objective_input_uncertainty} are independent of the choice of \majorchange{$\omega$ in \eqref{eqn:ebsf_interpolation} as long as the design requirements in (a) and (b) are satisfied}. In fact, it turns out that one can remove the error dependence altogether by choosing $\omega = 0$, and Theorem \ref{thm:safety_objective_input_uncertainty} will still hold. However, for good tracking of $\vec{x}_*$, $\delta$ typically must be very small. Thus, an interpolation window given only by $\bar{h}(\vec{x}) \in (\frac{\delta}{3}, \frac{2\delta}{3})$ would be very thin. If the plant were to enter such a thin interpolation window, $\beta$ could change from 0 to 1 rapidly, resulting in a rapid change in $\mu$ that would tend to push the plant back across the window and cause $\beta$ to rapidly change from 1 to 0. Therefore, while $\mu(t)$ is never discontinuous with time, a thin interpolation window can result in a chattering behavior where $\mu$ changes rapidly between the two aforementioned possibilities. The EBSF solution to this problem is to vary the width of the interpolation window as $\|\vec{e}_x\|$ varies. Intuitively, while $\|\vec{e}_x\|$ is large, the plant is less predictable and more inclined to enter the window, and thus we need the window to be wide to avoid the aforementioned chatter. Conversely, in the limit as $\|\vec{e}_x\|$ goes to zero, the plant becomes predictable, and we can shrink the window without fear of the plant entering. Thus, $\omega(t)$ acts as as an error-based smoothing of $\mu(t)$, helping to reduce chatter during the transient period of adaptation. Principled ways to design $\omega$ or otherwise avoid chatter are an important matter for future work. The simulation results in Section \ref{sec:simulations} will further explore this chattering behavior, and will propose a way of largely mitigating it in practice.

%% file: Body/Simulation.tex
To demonstrate the efficacy of our proposed approaches, we simulate a simple academic system consisting of a point mass $m$ moving in the $x$-$y$ plane, tethered to the origin by an ideal spring and damper with coefficients $k$ and $b$. The control inputs are forces $F_x$ and $F_y$. The control objective is to cause $\vec{p} = [x, y]^\top$ to track a moving reference point $\vec{p}_*(t) = [x_*(t), y_*(t)]^\top$, while the safety objective is to avoid a circular pillar centered at $(x_{\rm pillar}, y_{\rm pillar})$ with radius $r_{\rm pillar}$. In Problem \ref{prob:unforced_dynamics}, only $k$ and $b$ are unknown, while in Problem \ref{prob:input_matrix}, $m$ is unknown as well. It is straightforward to show that the dynamics reduce to
\if \numsides 2
\begin{equation}
    \begin{gathered} \label{eqn:simulation_plant}
        \dot{\vec{x}} = A\vec{x} + B(\lambda_*\vec{u} - F(\vec{x})\theta_*), \\
        \vec{x} = \left[x, y, \dot{x}, \dot{y}\right]^\top, \quad \vec{u} = \left[F_x, F_y\right]^\top, \quad \lambda_* = \frac{1}{m}, \\
        \theta_* = \left[\frac{k}{m}, \frac{b}{m}\right]^\top, \quad
        A = \begin{bmatrix} 0_{2 \times 2} & I_2 \\ 0_{2 \times 2} & 0_{2 \times 2} \end{bmatrix}, \\
        B = \begin{bmatrix} 0_{2 \times 2} \\ I_2 \end{bmatrix}, \quad
        F(\vec{x}) = \begin{bmatrix} x & \frac{x(x\dot{x} + y\dot{y})}{x^2 + y^2} \\ y & \frac{y(x\dot{x} + y\dot{y})}{x^2 + y^2} \end{bmatrix}.
    \end{gathered}
\end{equation}
\else
\begin{equation}
    \begin{gathered} \label{eqn:simulation_plant}
        \dot{\vec{x}} = A\vec{x} + B(\lambda_*\vec{u} - F(\vec{x})\theta_*), \\
        \vec{x} = \left[x, y, \dot{x}, \dot{y}\right]^\top, \quad \vec{u} = \left[F_x, F_y\right]^\top, \quad \lambda_* = \frac{1}{m}, \quad \theta_* = \left[\frac{k}{m}, \frac{b}{m}\right]^\top, \\
        A = \begin{bmatrix} 0_{2 \times 2} & I_2 \\ 0_{2 \times 2} & 0_{2 \times 2} \end{bmatrix}, \quad B = \begin{bmatrix} 0_{2 \times 2} \\ I_2 \end{bmatrix}, \quad
        F(\vec{x}) = \begin{bmatrix} x & \frac{x(x\dot{x} + y\dot{y})}{x^2 + y^2} \\ y & \frac{y(x\dot{x} + y\dot{y})}{x^2 + y^2} \end{bmatrix}.
    \end{gathered}
\end{equation}
\fi

We use the following true parameter values and initial parameter estimates: \minorchange{$m = 1.1\mathrm{kg}$, $k = 1.2\mathrm{N/m}$, $b = 0.9\mathrm{N/(m/s)}$, $\hat{k}(0) = 0.5\mathrm{N/m}$, and $\hat{b}(0) = 2\mathrm{N/(m/s)}$}. In the adaptive control design, we assume that $k \in [\ubar{k}, \bar{k}]$ and $b \in [\ubar{b}, \bar{b}]$ where \minorchange{$\ubar{k} = 0\mathrm{N/m}$, $\bar{k} = 3\mathrm{N/m}$, $\ubar{b} = 0\mathrm{N/(m/s)}$, and $\bar{b} = 3\mathrm{N/(m/s)}$}. Additionally, in Problem 2, we use the initial parameter estimate $\hat{m}(0) = 2\mathrm{kg}$, and we assume that $m \in [\ubar{m}, \bar{m}]$ with $\ubar{m} = 0.5\mathrm{kg}$, $\bar{m} = 2.5\mathrm{kg}$. Initial estimates $\hat{\theta}(0)$ and $\hat{\vartheta}(0)$ are derived from $\hat{k}(0)$, $\hat{b}(0)$, and $\hat{m}(0)$, and element-wise bounds on $\theta_*$ and $\lambda_*$ are derived from $\ubar{k}$, $\bar{k}$, $\ubar{b}$, $\bar{b}$, $\ubar{m}$, and $\bar{m}$.
We implement each controller with a zero-order hold at 100 Hz.
\minorchange{The initial parameter estimate is chosen such that, once the point mass reaches the end of the desired trajectory, it finds itself pulled toward the pillar more quickly than initially expected.} Furthermore, the desired trajectory $\vec{p}_*$ is chosen to drive the plant through the obstacle and to stop inside it, necessitating deviation from the desired trajectory to ensure safety.
\if \arxivversion 1
See Appendix \ref{app:simulation_nominal_control} for the ideal control solution if all parameters are known.
\else
See Appendix IV.A in \cite{fisher2026EBSBarxiv} for the ideal control solution if all parameters are known.
\fi

\subsection{Augmented CBF Design} \label{subsec:simulations_composite_cbf}

In our simulations, we use the 1st-order CBF candidate
\minorchange{\begin{equation} \label{eqn:simulations_h_1}
    h_1(\vec{x}) = \sqrt{(x - x_{\rm pillar})^2 + (y - y_{\rm pillar})^2} - r_{\rm pillar}
\end{equation}}
to represent the safe set, and construct $h_2$ as in \eqref{eqn:hocbf_recursion} with $\alpha_1 = 1$. Then,
by choosing arbitrary maximum and minimum values for each state as $\bar{x} = \bar{y} = -\ubar{x} = -\ubar{y} = 10$m, $\bar{\dot{x}} = \bar{\dot{y}} = -\ubar{\dot{x}} = -\ubar{\dot{y}} = 5$m/s and the $\mathrm{softmin}$ smoothing parameter $\kappa = 10$, we construct $\bar{h}$ as in Section \ref{subsec:constructing_hbar}.


\subsection{Problem \ref{prob:unforced_dynamics}: Simulation Results} \label{subsec:simulations_problem_1}

For Problem \ref{prob:unforced_dynamics} with $k$ and $b$ unknown and $m$ known, we simulate and compare four approaches with formal guarantees of safety: aCBF \cite{taylor2020aCBFs}, RaCBF \cite{lopez2021RaCBFs}, EBSB, and EBSF.
\if \arxivversion 1
See Appendices \ref{app:aCBF} and \ref{app:RaCBF} respectively for aCBF and RaCBF implementation details. For comparison, we also simulate the hypothetical ideal safe controller in Appendix \ref{app:simulation_nominal_control}.
\else
See Appendices IV.B and IV.C in \cite{fisher2026EBSBarxiv} for aCBF and RaCBF implementation details. For comparison, we also simulate the hypothetical ideal safe controller in Appendix IV.A in \cite{fisher2026EBSBarxiv}.
\fi
Set Membership Identification (SMID), proposed in \cite{lopez2021RaCBFs} to improve RaCBF performance, is not used here; we defer discussion of SMID to Section \ref{subsec:simulations_SMID}.
The control adaptive laws in all approaches use $\gamma_\theta = \gamma_\vartheta = 3$, and for the additional safety adaptive laws in the aCBF and RaCBF approaches, \minorchange{$\gamma_{\theta,s} = \gamma_{\lambda,s} = 20$} is used.
\minorchange{Solving the NLP in Remark \ref{rmk:NLP}, we obtain $E_* = 2.00$, while calculating $E_1$ in \eqref{eqn:E_1} yields $E_1 = 1.41$. Thus, Assumption \ref{asn:gradients_unopposed_at_boundary} is satisfied.}
\majorchange{For $\omega(t)$ in \eqref{eqn:ebsf_interpolation}, we propose
\begin{equation} \label{eqn:omega}
    \dot{\omega} = \begin{cases} \sigma_1(\sigma_3\|\vec{e}_x\| - \omega), & \sigma_3\|\vec{e}_x\| \geq \omega \\ \sigma_2(\sigma_3\|\vec{e}_x\| - \omega), & \sigma_3\|\vec{e}_x\| < \omega \end{cases},
\end{equation}
where $\sigma_1, \sigma_2, \sigma_3 > 0$ are hand-tuned to control how the EBSF interpolation window (see Section \ref{subsec:ebcg_intuition}) varies with $\|\vec{e}_x\|$. When $\sigma_1 \gg \sigma_2$, the interpolation window widens much more quickly in response to increases in $\|\vec{e}_x\|$ than it narrows in response to decreases in $\|\vec{e}_x\|$, creating a smoothened hysteresis-like effect that helps produce a smoother EBSF control input under the zero-order hold. For Problem \ref{prob:unforced_dynamics}, we choose $\sigma_1 = 100$, $\sigma_2 = 0.5$, and $\sigma_3 = 3$.}

The $(x, y)$ trajectories of all simulated approaches are compared in Fig. \ref{fig:all_trajectories_1}, and \minorchange{the control inputs $\vec{u}$ are compared in Fig. \ref{fig:all_inputs_1}}. Clearly, while both aCBF and RaCBF maintain safety, they maintain a distance from the pillar which depends on the choice of $\bar{k} - \ubar{k}$ and $\bar{b} - \ubar{b}$.
In contrast, EBSB and EBSF are able to stay closer to the pillar in transient, and as shown in the preceding theorems, they permit the plant to start approaching the pillar as $\|\vec{e}_x\|$ tends toward zero.

\if \numsides 2
\begin{figure}
    \centering
    \includegraphics[width=3.5in]{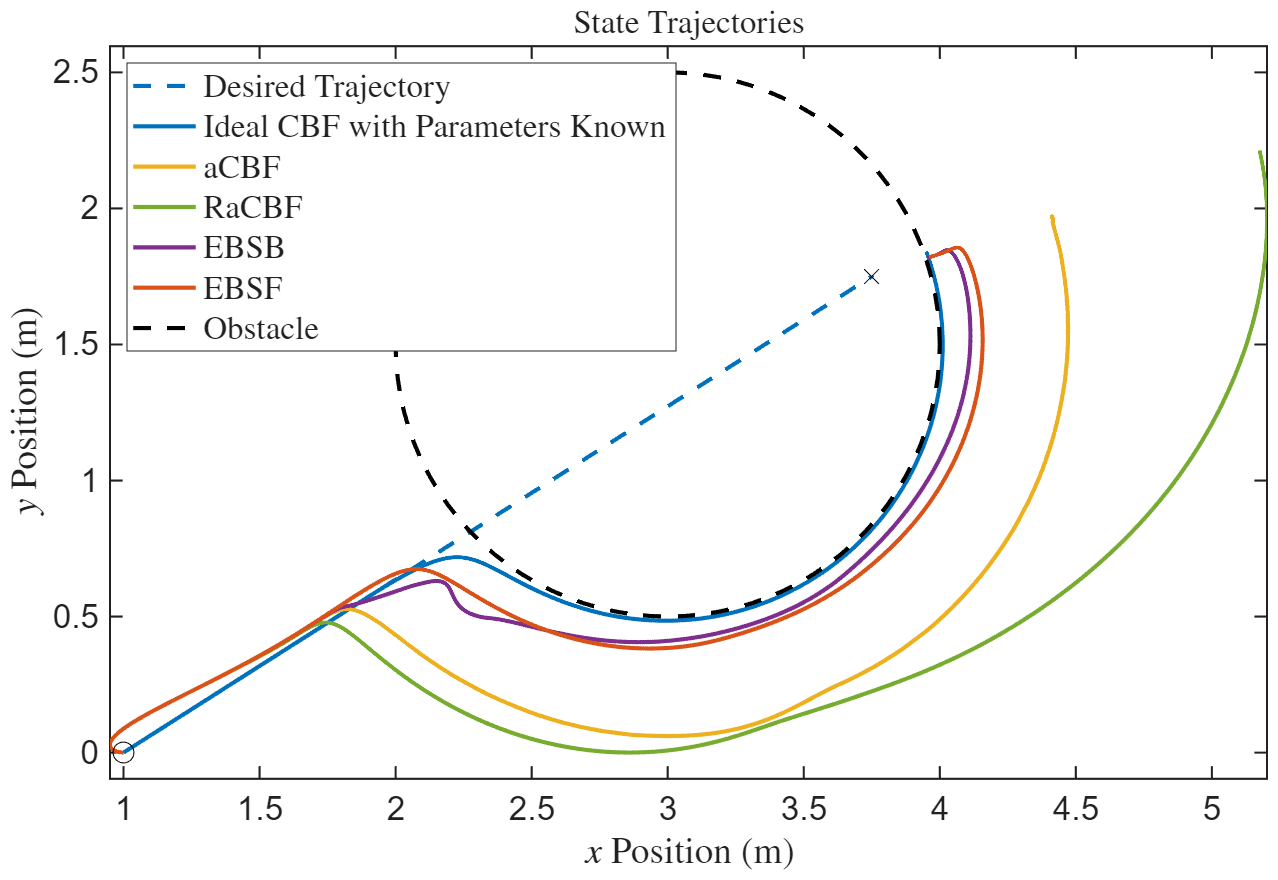}
    \caption{Trajectories resulting from aCBF \cite{taylor2020aCBFs}, RaCBF \cite{lopez2021RaCBFs}, EBSB, and EBSF for Problem \ref{prob:unforced_dynamics} compared to the ideal controller with all parameters known.}
    \label{fig:all_trajectories_1}
\end{figure}

\begin{figure}
    \centering
    \includegraphics[width=3.5in]{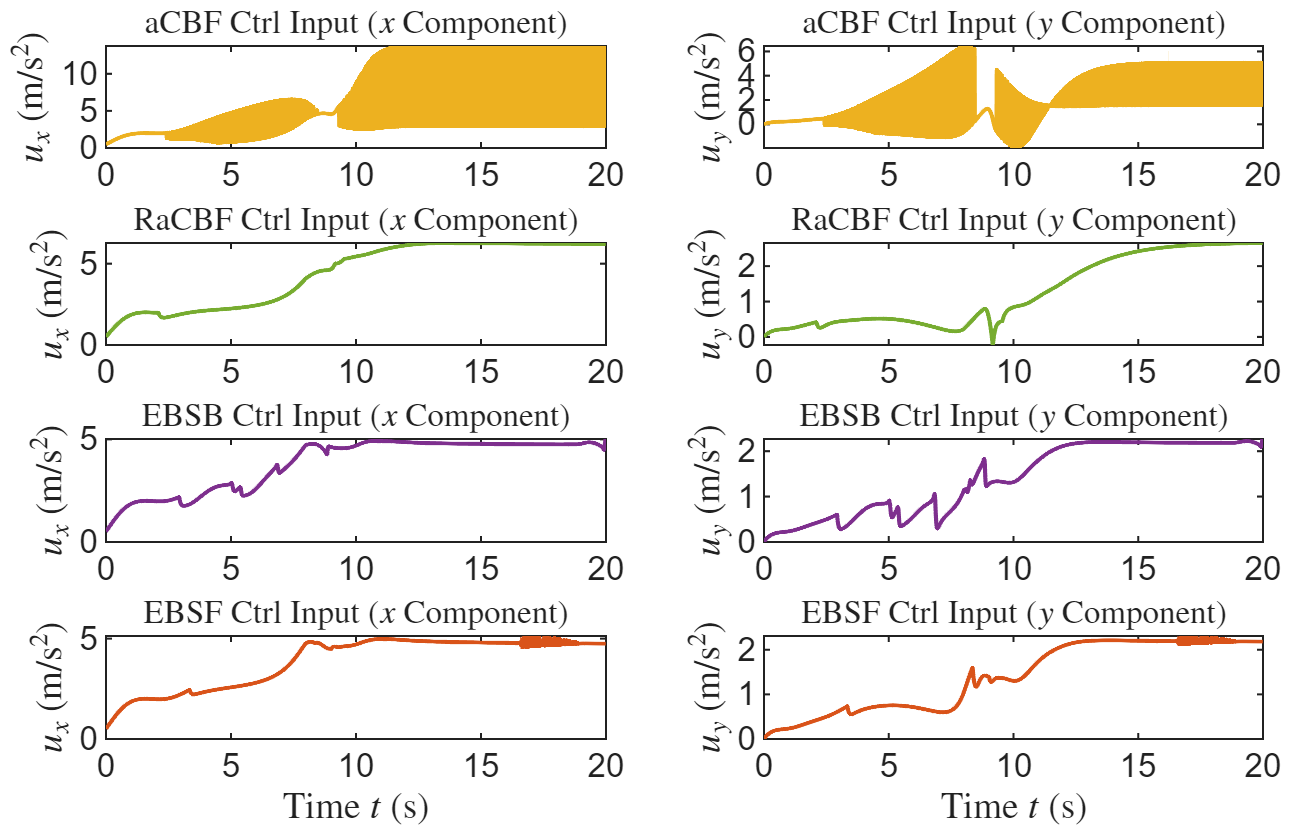}
    \caption{\minorchange{Control inputs $\vec{u}$} from aCBF \cite{taylor2020aCBFs}, RaCBF \cite{lopez2021RaCBFs}, EBSB, and EBSF for Problem \ref{prob:unforced_dynamics}.}
    \label{fig:all_inputs_1}
\end{figure}
\else
\begin{figure}
    \centering
    \begin{minipage}{3.5in}
        \centering
        \includegraphics[width=3.5in]{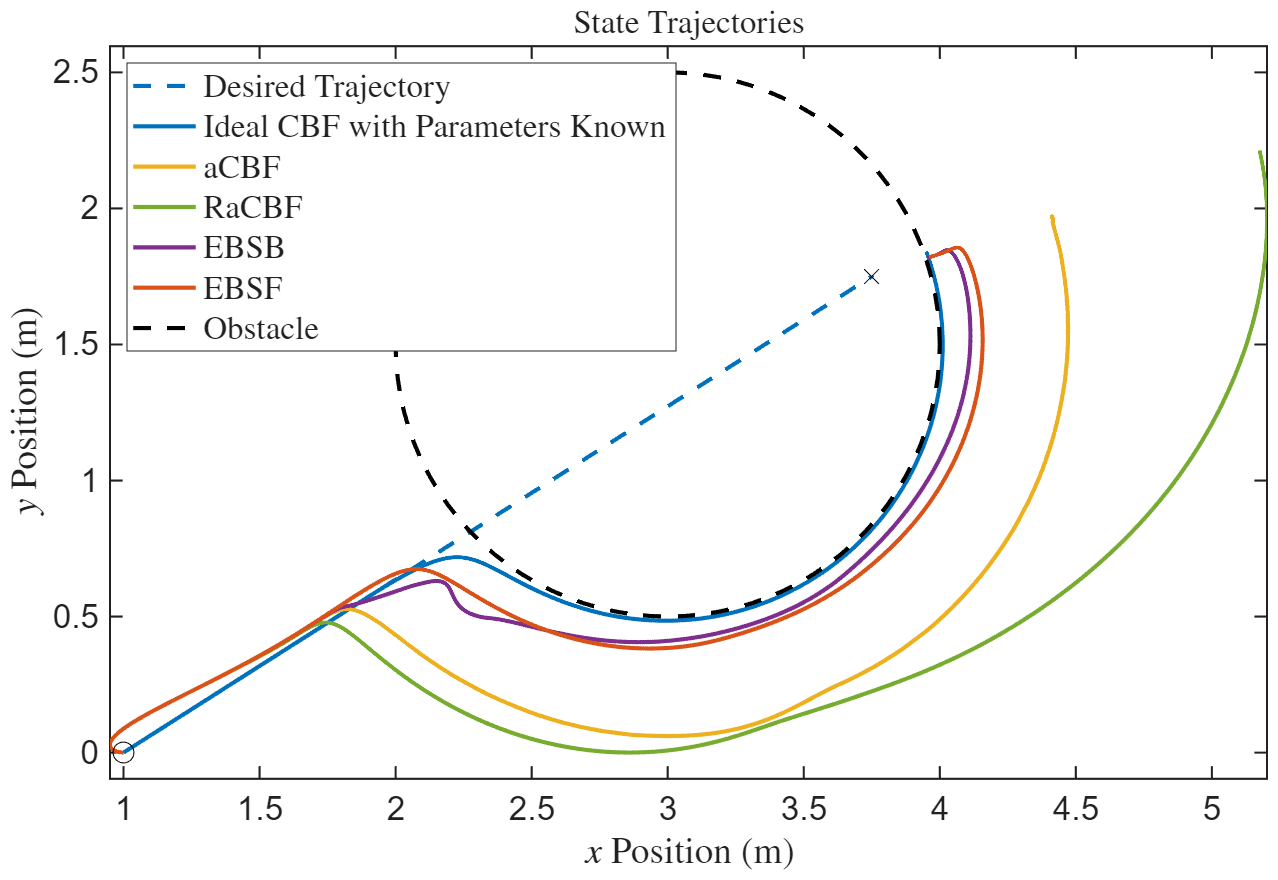}
        \caption{Trajectories resulting from aCBF \cite{taylor2020aCBFs}, RaCBF \cite{lopez2021RaCBFs}, EBSB, and EBSF for Problem \ref{prob:unforced_dynamics} compared to the ideal controller with all parameters known.}
        \label{fig:all_trajectories_1}
    \end{minipage}
    \hfill
    \begin{minipage}{3.5in}
        \centering
        \includegraphics[width=3.5in]{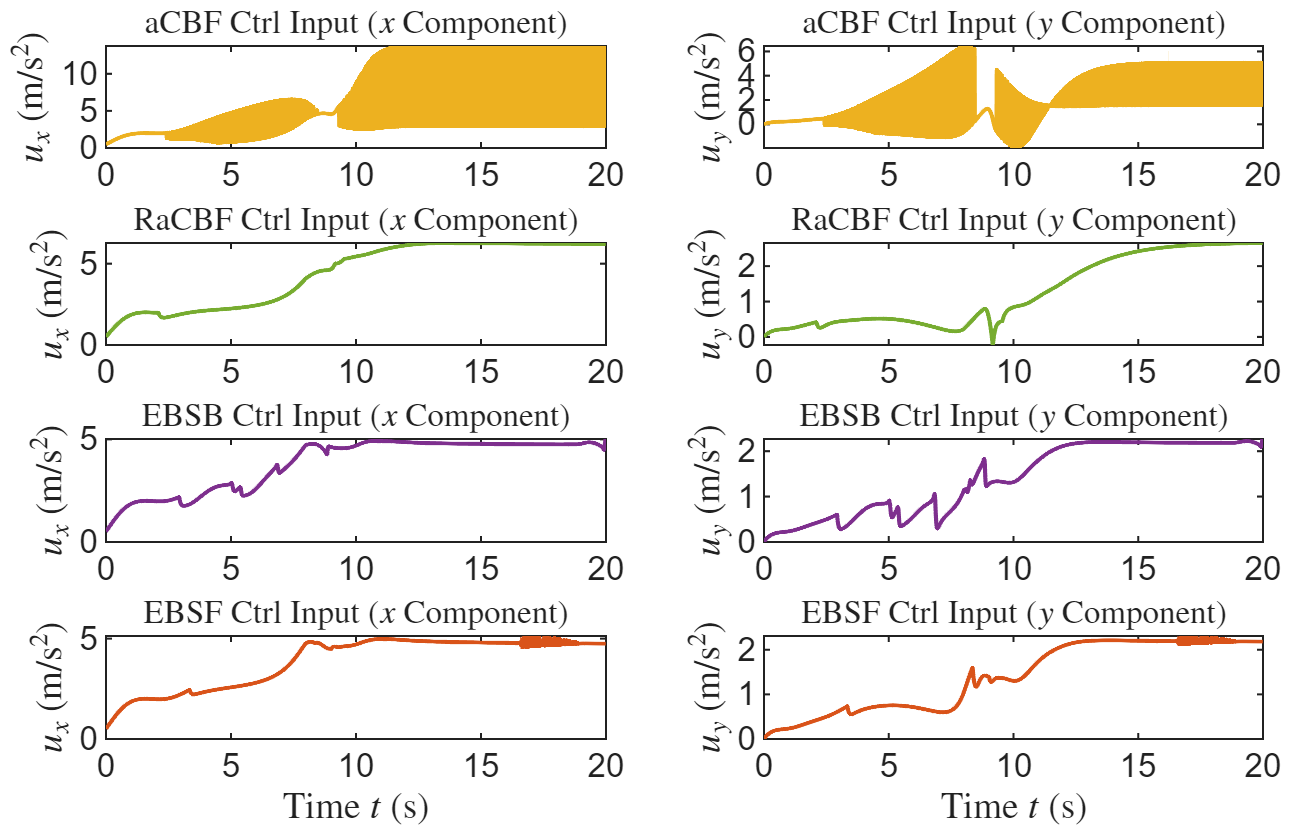}
        \caption{\minorchange{Control inputs $\vec{u}$} from aCBF \cite{taylor2020aCBFs}, RaCBF \cite{lopez2021RaCBFs}, EBSB, and EBSF for Problem \ref{prob:unforced_dynamics}.}
        \label{fig:all_inputs_1}
    \end{minipage}
\end{figure}
\fi

\subsection{Problem \ref{prob:input_matrix}: Simulation Results} \label{subsec:simulations_problem_2}

For Problem \ref{prob:input_matrix} with $k$, $b$, and $m$ unknown, we simulate and compare three approaches: aCBF \cite{taylor2020aCBFs}, RaCBF \cite{lopez2021RaCBFs}, and EBSF.
\if \arxivversion 1
For comparison, we also simulate the hypothetical ideal safe controller in Appendix \ref{app:simulation_nominal_control}.
\else
For comparison, we also simulate the hypothetical ideal safe controller in Appendix IV.A in \cite{fisher2026EBSBarxiv}.
\fi
The true parameters, initial estimates, and upper/lower bounds are the same as in Section \ref{subsec:simulations_problem_1}, with the addition of $\hat{m}(0) = 2\mathrm{kg}$ and the assumption that $m \in [\ubar{m}, \bar{m}]$ with $\ubar{m} = 0.5\mathrm{kg}$, $\bar{m} = 2.5\mathrm{kg}$. \majorchange{Furthermore, for Problem \ref{prob:input_matrix}, we choose $\sigma_2 = 0.3$.}

The $(x, y)$ trajectories of all simulated approaches are compared in Fig. \ref{fig:all_trajectories_2}, and \minorchange{the control inputs $\vec{u}$} are compared in Fig. \ref{fig:all_inputs_2}. Again, aCBF and RaCBF remain safe but are even more conservative than in Problem \ref{prob:unforced_dynamics} due to the increased uncertainty in the parameters, \minorchange{so much so that RaCBF begins to experience chatter from the zero-order hold.} In contrast, EBSF is able to be less conservative while maintaining safety even under input uncertainties.

\if \numsides 2
\begin{figure}
    \centering
    \includegraphics[width=3.5in]{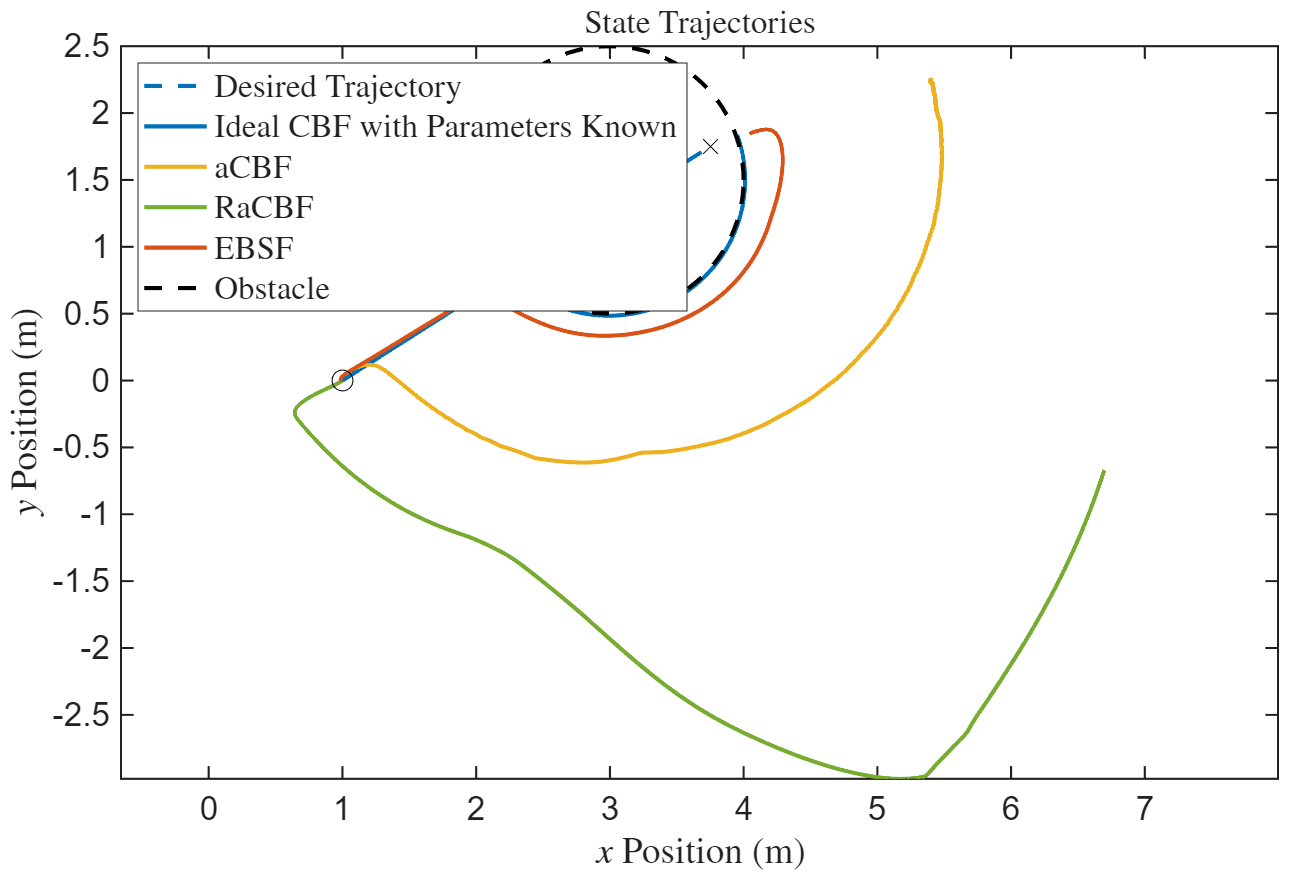}
    \caption{Trajectories resulting from aCBF \cite{taylor2020aCBFs}, RaCBF \cite{lopez2021RaCBFs}, and EBSF for Problem \ref{prob:input_matrix} compared to the ideal controller with all parameters known.}
    \label{fig:all_trajectories_2}
\end{figure}

\begin{figure}
    \centering
    \includegraphics[width=3.5in]{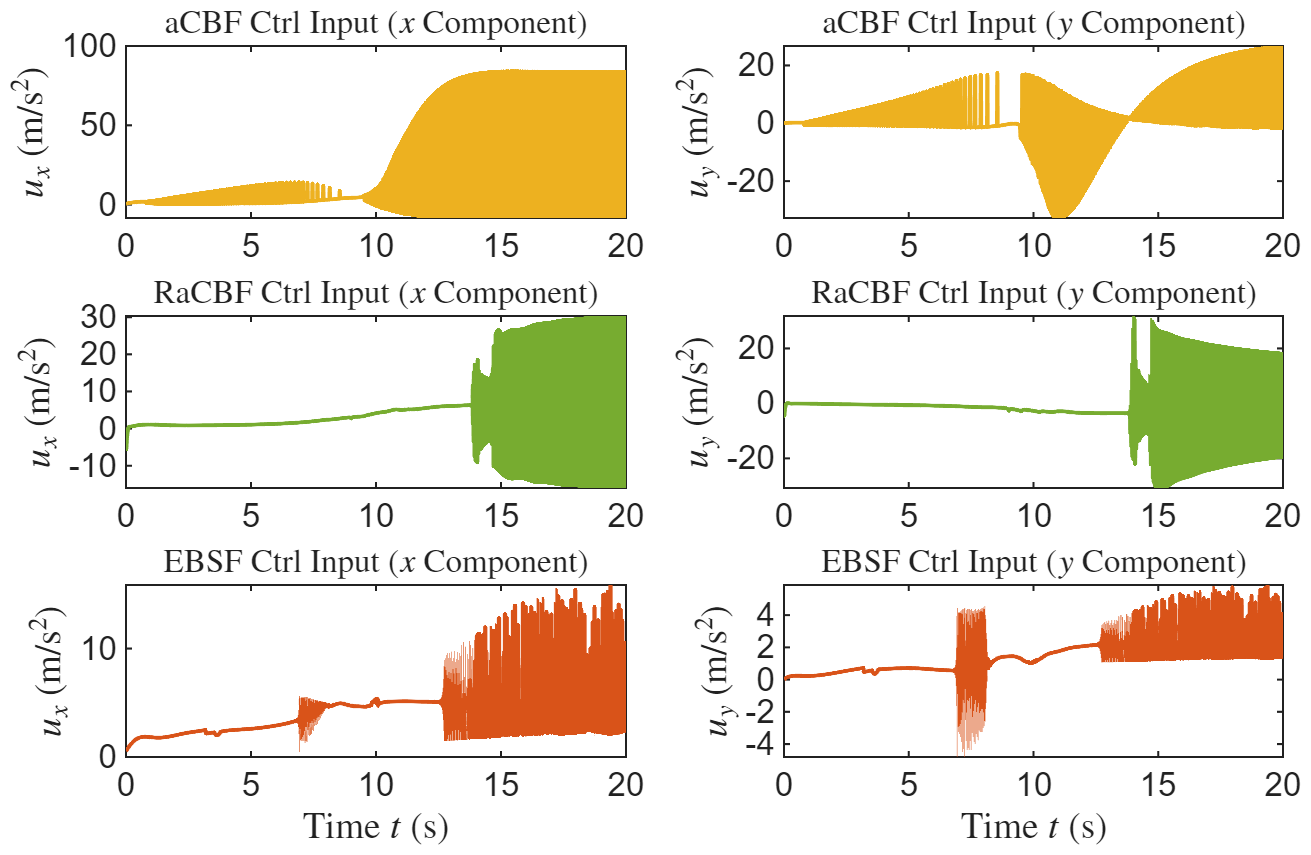}
    \caption{\minorchange{Control inputs $\vec{u}$} from aCBF \cite{taylor2020aCBFs}, RaCBF \cite{lopez2021RaCBFs}, and EBSF for Problem \ref{prob:input_matrix}.}
    \label{fig:all_inputs_2}
\end{figure}
\else
\begin{figure}
    \centering
    \begin{minipage}{3.5in}
        \centering
        \includegraphics[width=3.5in]{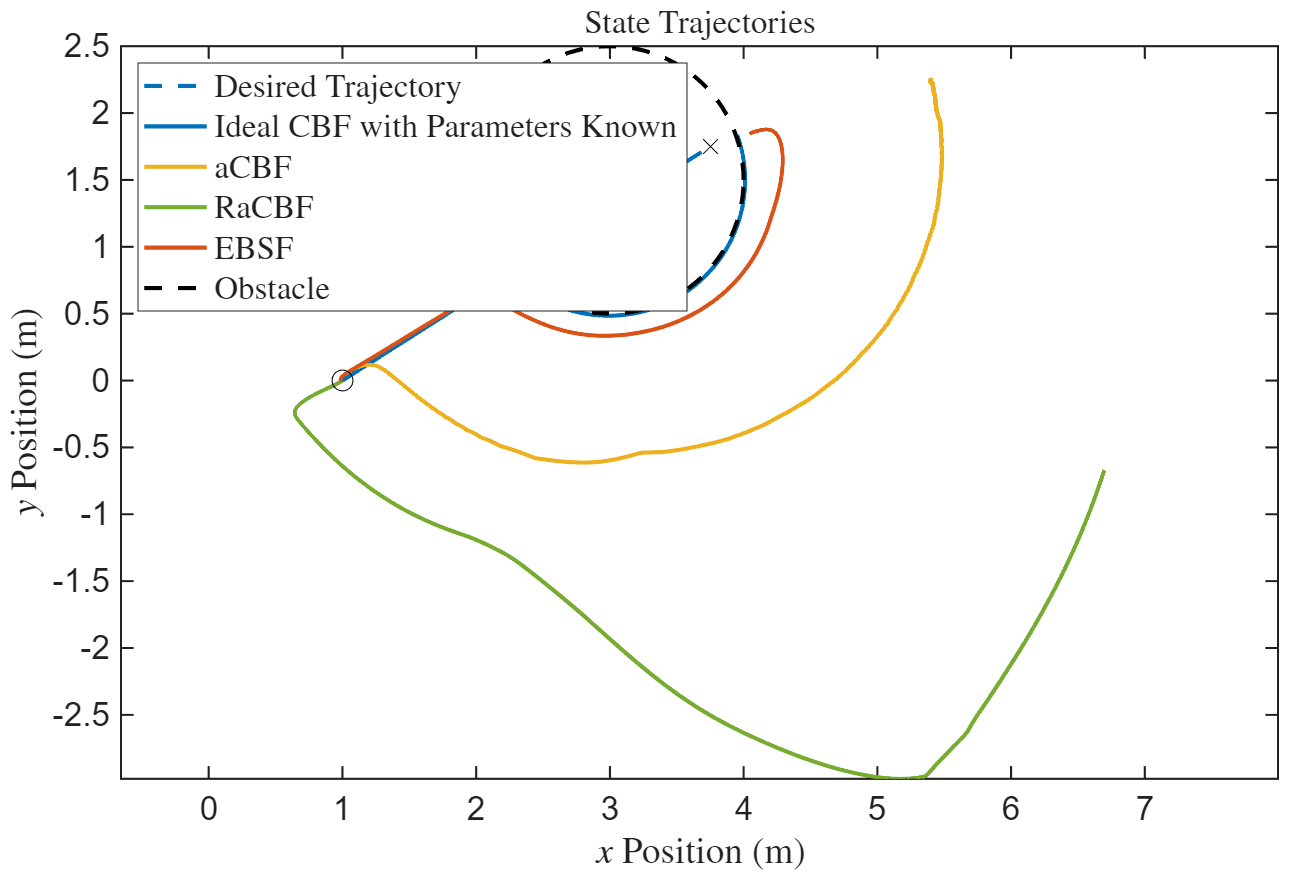}
        \caption{Trajectories resulting from aCBF \cite{taylor2020aCBFs}, RaCBF \cite{lopez2021RaCBFs}, and EBSF for Problem \ref{prob:input_matrix} compared to the ideal controller with all parameters known.}
        \label{fig:all_trajectories_2}
    \end{minipage}
    \hfill
    \begin{minipage}{3.5in}
        \centering
        \includegraphics[width=3.5in]{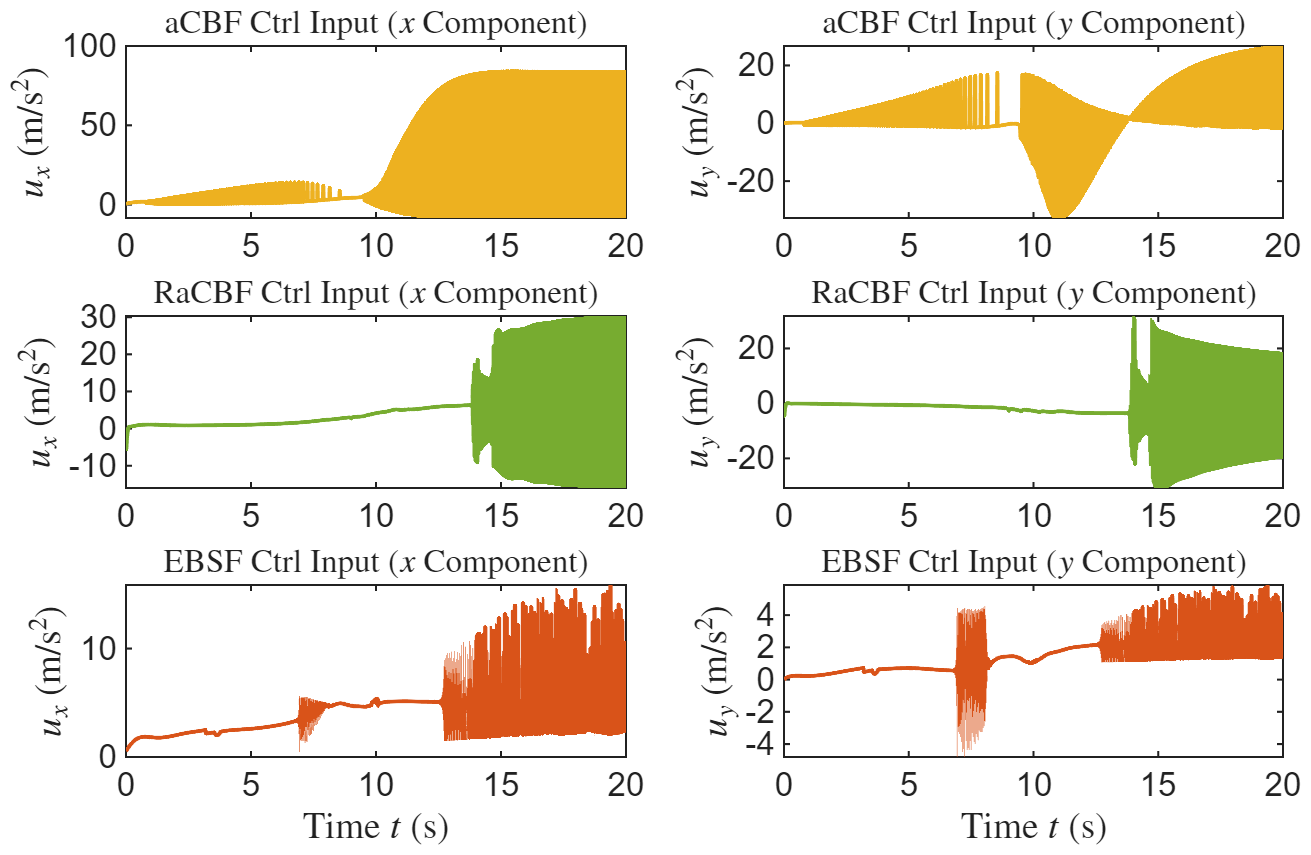}
        \caption{\minorchange{Control inputs $\vec{u}$} from aCBF \cite{taylor2020aCBFs}, RaCBF \cite{lopez2021RaCBFs}, and EBSF for Problem \ref{prob:input_matrix}.}
        \label{fig:all_inputs_2}
    \end{minipage}
\end{figure}
\fi







\subsection{Discussion: EBSB and EBSF} \label{subsec:simulations_discussion}

Figs. \ref{fig:all_trajectories_1}-\ref{fig:all_inputs_2} demonstrate the relative advantages and drawbacks of EBSB and EBSF.
\minorchange{Clearly, EBSF is applicable to a wider range of problems than EBSB. Within Problem \ref{prob:unforced_dynamics}, however,
the primary advantage of EBSB over EBSF is its avoidance of chatter in the control input, visible in EBSF in Figs. \ref{fig:all_inputs_1} and \ref{fig:all_inputs_2} as $\|\vec{e}_x\|$ goes to zero and the plant converges towards $\partial\bar{S}$. \majorchange{The smoothed hysteresis design of $\omega(t)$ in \eqref{eqn:omega} was chosen to heuristically mitigate chatter, and the parameters $\sigma_1, \sigma_2, \sigma_3$ were hand-tuned.} However, the theory presented in Section \ref{sec:input_uncertainty} provides no principled method of designing $\omega$, meaning that obtaining desirable performance with EBSF may require extensive design iteration and tuning. Such a principled method is an important topic for future work. In contrast, with EBSB, one may design and tune the adaptive controller in Section \ref{subsec:EBSB_adaptive_control} and a CBF-based safety filter for the reference model in \eqref{eqn:reference_model} independently using standard methods. Then, \eqref{eqn:safe_reference_EBSB}-\eqref{eqn:z_p_ebsb} integrate these two designs with no hand-tuning required for chatter-free, safe, and non-conservative control.}

\minorchange{The chatter exhibited by EBSF is an artifact of the zero-order hold used in our simulations. Under a zero-order hold, when $\omega$ becomes very small and adaptation is not yet complete, $\beta$ can go from zero to one or vice versa in the span of a single time step, causing the input produced by EBSF to change rapidly. When the time step is decreased, we find that the chatter disappears. However, digital control systems always have a finite time step, and thus further exploration of this phenomenon and how to mitigate it is essential for practical application.
In the following section, we will explore the use of a data-driven method to reduce chatter in practice.}

\if \numsides 2
\begin{figure}
    \centering
    \includegraphics[width=3.5in]{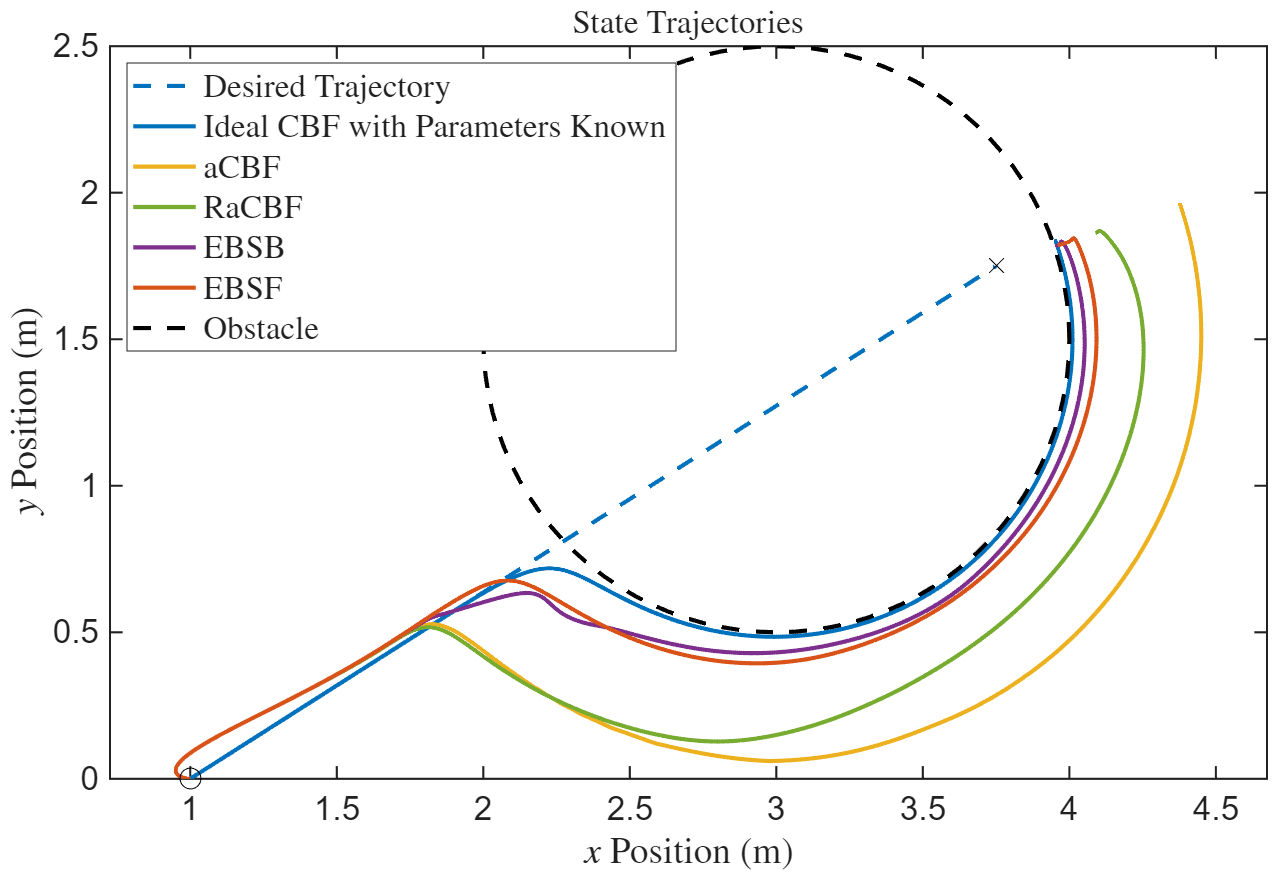}
    \caption{Trajectories resulting from aCBF \cite{taylor2020aCBFs}, RaCBF \cite{lopez2021RaCBFs}, EBSB, and EBSF with SMID for Problem \ref{prob:unforced_dynamics} compared to the ideal controller with all parameters known.}
    \label{fig:all_trajectories_1_smid}
\end{figure}

\begin{figure}
    \centering
    \includegraphics[width=3.5in]{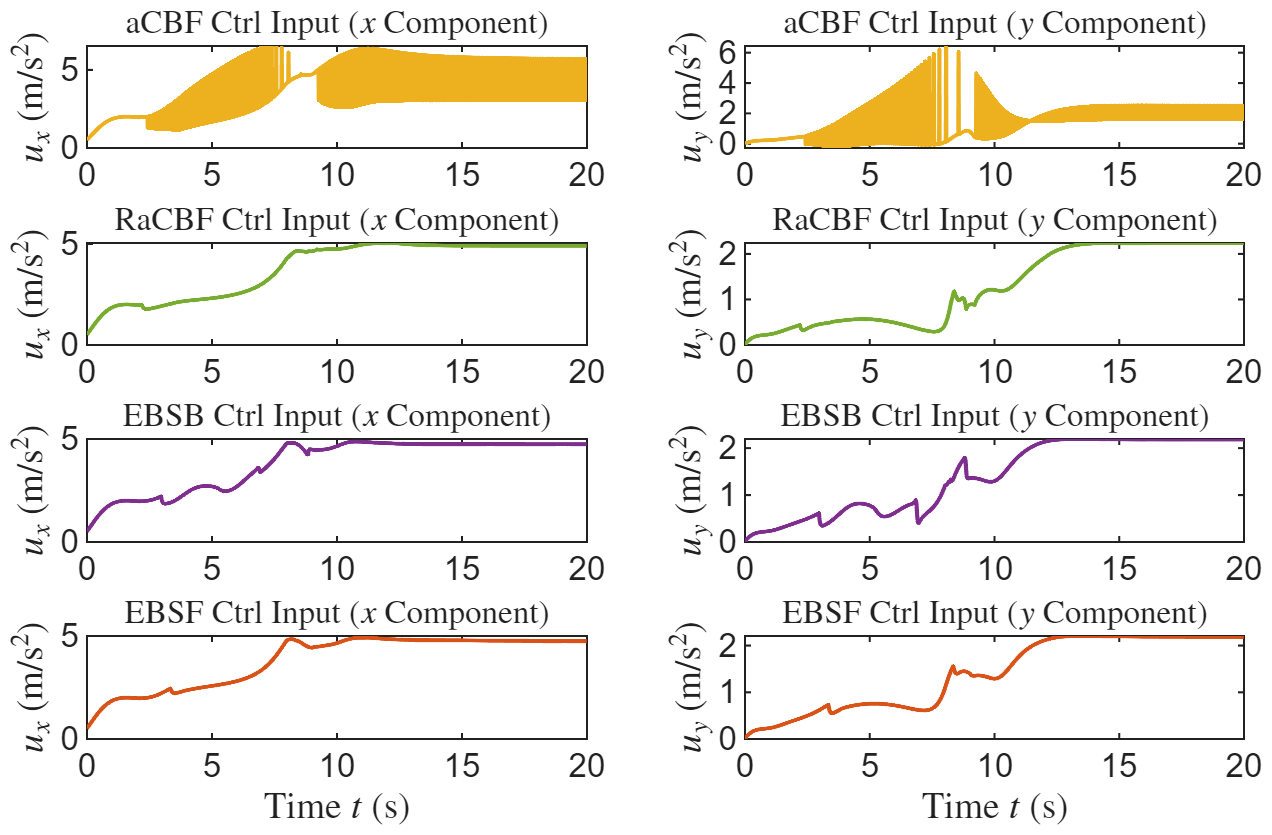}
    \caption{\minorchange{Control inputs $\vec{u}$} from aCBF \cite{taylor2020aCBFs}, RaCBF \cite{lopez2021RaCBFs}, EBSB, and EBSF with SMID for Problem \ref{prob:unforced_dynamics}.}
    \label{fig:all_inputs_1_smid}
\end{figure}
\else
\begin{figure}
    \centering
    \begin{minipage}{3.5in}
        \centering
        \includegraphics[width=3.5in]{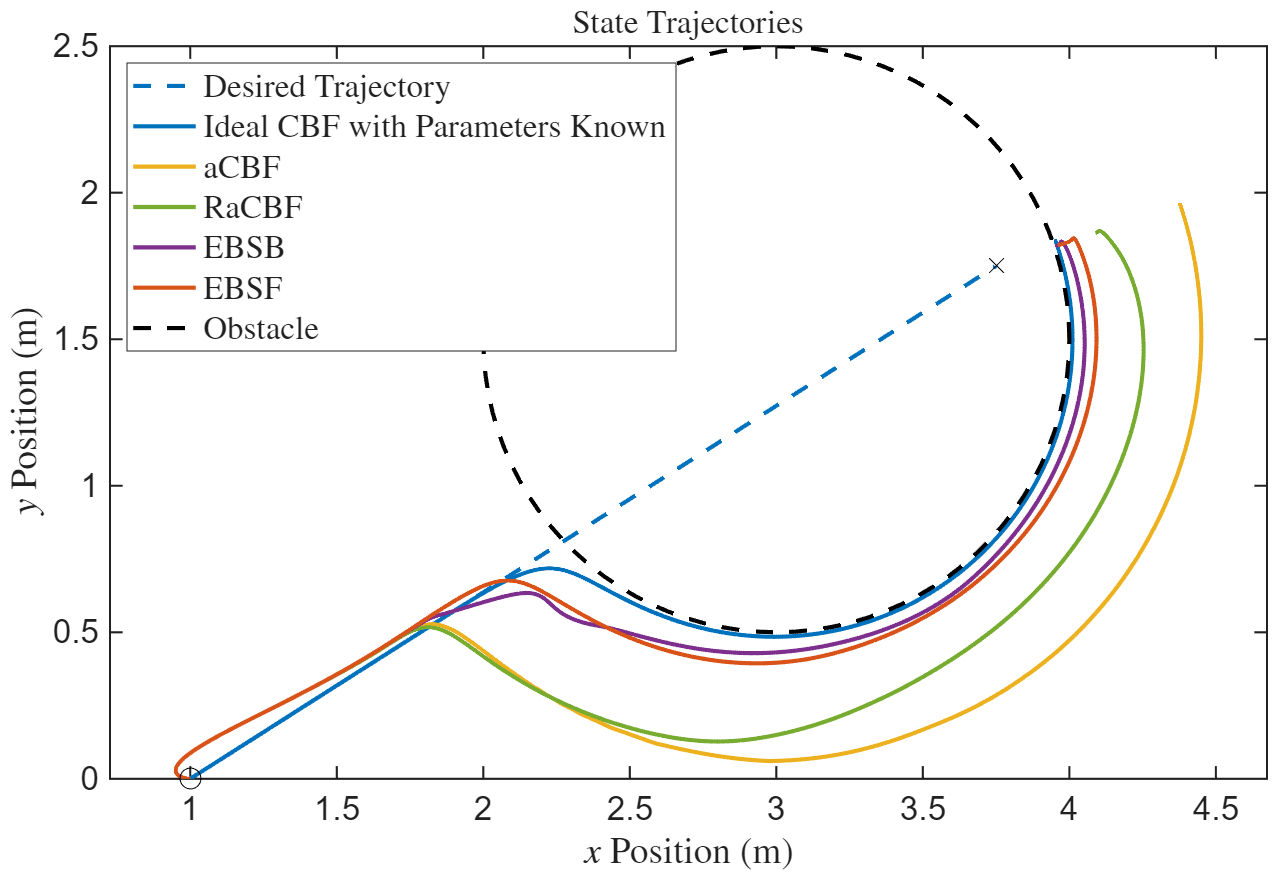}
        \caption{Trajectories resulting from aCBF \cite{taylor2020aCBFs}, RaCBF \cite{lopez2021RaCBFs}, EBSB, and EBSF with SMID for Problem \ref{prob:unforced_dynamics} compared to the ideal controller with all parameters known.}
        \label{fig:all_trajectories_1_smid}
    \end{minipage}
    \hfill
    \begin{minipage}{3.5in}
        \centering
        \includegraphics[width=3.5in]{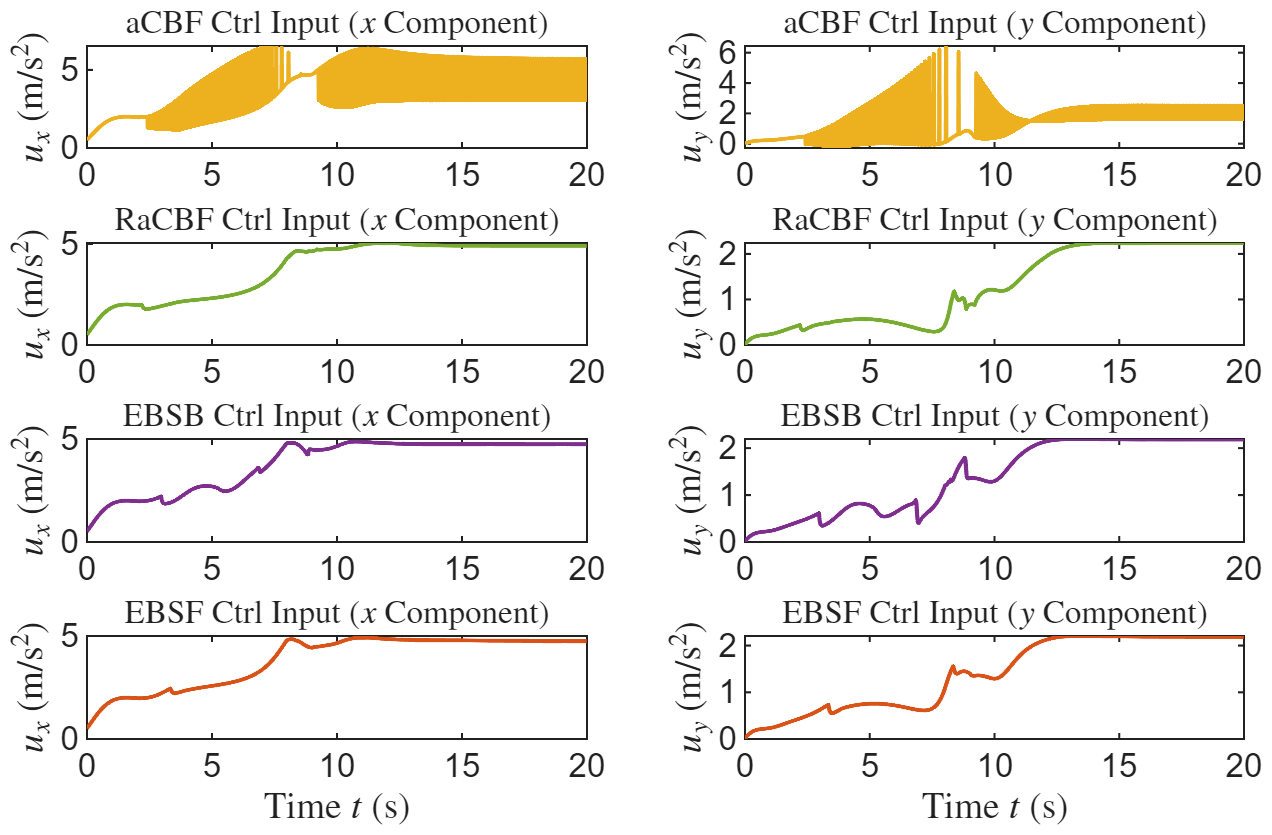}
        \caption{\minorchange{Control inputs $\vec{u}$} from aCBF \cite{taylor2020aCBFs}, RaCBF \cite{lopez2021RaCBFs}, EBSB, and EBSF with SMID for Problem \ref{prob:unforced_dynamics}.}
        \label{fig:all_inputs_1_smid}
    \end{minipage}
\end{figure}
\fi

\subsection{Performance Enhancement with Set Membership Identification} \label{subsec:simulations_SMID}

Set Membership Identification (SMID) was proposed in \cite{lopez2021RaCBFs} as a method of reducing the conservatism of RaCBF in practice. In the preceeding simulations, SMID was omitted. In this section, we add SMID to RaCBF following the approach in \cite{lopez2021RaCBFs},
\if \arxivversion 1
and we add SMID to EBSB and EBSF following the approaches in \ref{app:SMID}. Details of our SMID implementation can be found in Appendix \ref{app:SMID_implementation}.
\else
and we add SMID to EBSB and EBSF following the approaches in Appendix III in \cite{fisher2026EBSBarxiv}. Details of our SMID implementation can be found in Appendix IV.D in \cite{fisher2026EBSBarxiv}.
\fi

Figs. \ref{fig:all_trajectories_1_smid} and \ref{fig:all_inputs_1_smid} show the trajectories and \minorchange{control inputs} respectively of aCBF, RaCBF, EBSB, and EBSF augmented with SMID in the same simulation scenario as in Section \ref{subsec:simulations_problem_1}. SMID greatly reduces the conservatism of RaCBF, and also further reduces the conservatism of EBSB and EBSF, enabling EBSB and EBSF to still provide superior reference tracking. Furthermore, by narrowing the bounds on the uncertainty online, \minorchange{SMID eliminates the chatter in EBSF that was present in Fig. \ref{fig:all_inputs_1}.}

Figs. \ref{fig:all_trajectories_2_smid} and \ref{fig:all_inputs_2_smid} show the trajectories and \minorchange{control inputs} respectively of aCBF, RaCBF, and EBSF augmented with SMID in the same simulation scenario as in Section \ref{subsec:simulations_problem_2}. The same result is apparent: both RaCBF and EBSF are greatly reduced in conservatism, with EBSF still the least conservative approach.
\minorchange{Again, the use of SMID to narrow the bounds on the uncertainty in real time allows the chatter that was present in Fig. \ref{fig:all_inputs_2} to be eliminated.}

\if \numsides 2
\begin{figure}
    \centering
    \includegraphics[width=3.5in]{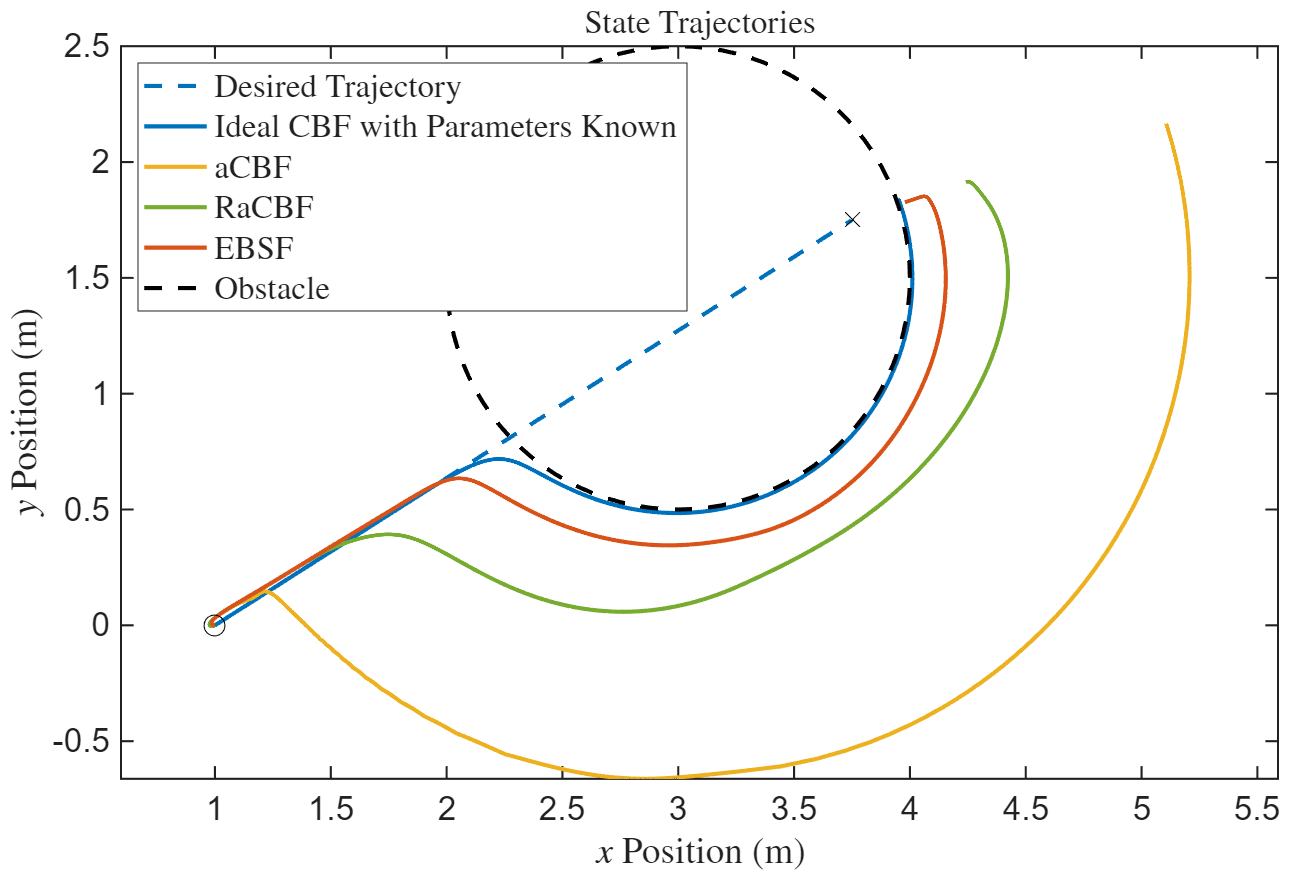}
    \caption{Trajectories resulting from aCBF \cite{taylor2020aCBFs}, RaCBF \cite{lopez2021RaCBFs}, and EBSF with SMID for Problem \ref{prob:input_matrix} compared to the ideal controller with all parameters known.}
    \label{fig:all_trajectories_2_smid}
\end{figure}

\begin{figure}
    \centering
    \includegraphics[width=3.5in]{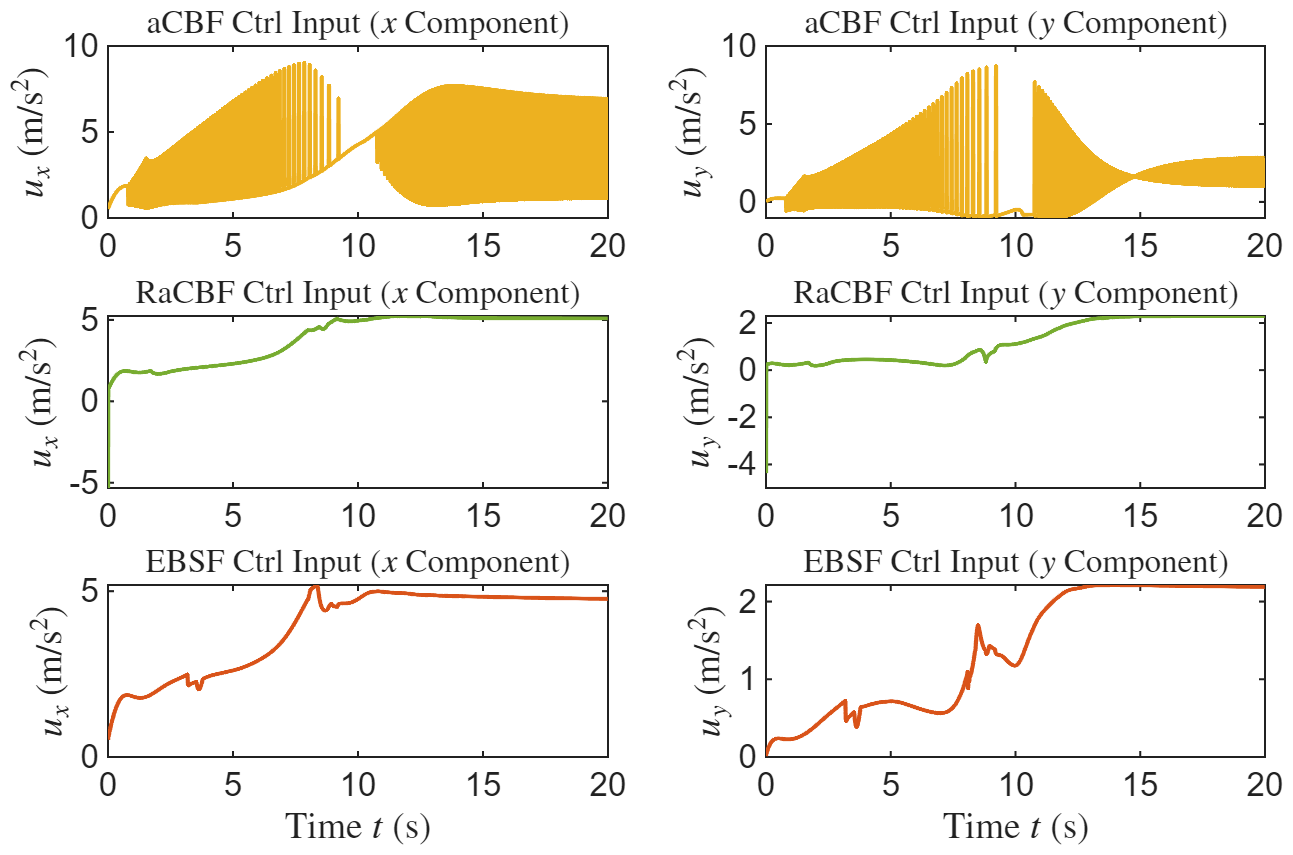}
    \caption{\minorchange{Control inputs $\vec{u}$} from aCBF \cite{taylor2020aCBFs}, RaCBF \cite{lopez2021RaCBFs}, and EBSF with SMID for Problem \ref{prob:input_matrix}.}
    \label{fig:all_inputs_2_smid}
\end{figure}
\else
\begin{figure}
    \centering
    \begin{minipage}{3.5in}
        \centering
        \includegraphics[width=3.5in]{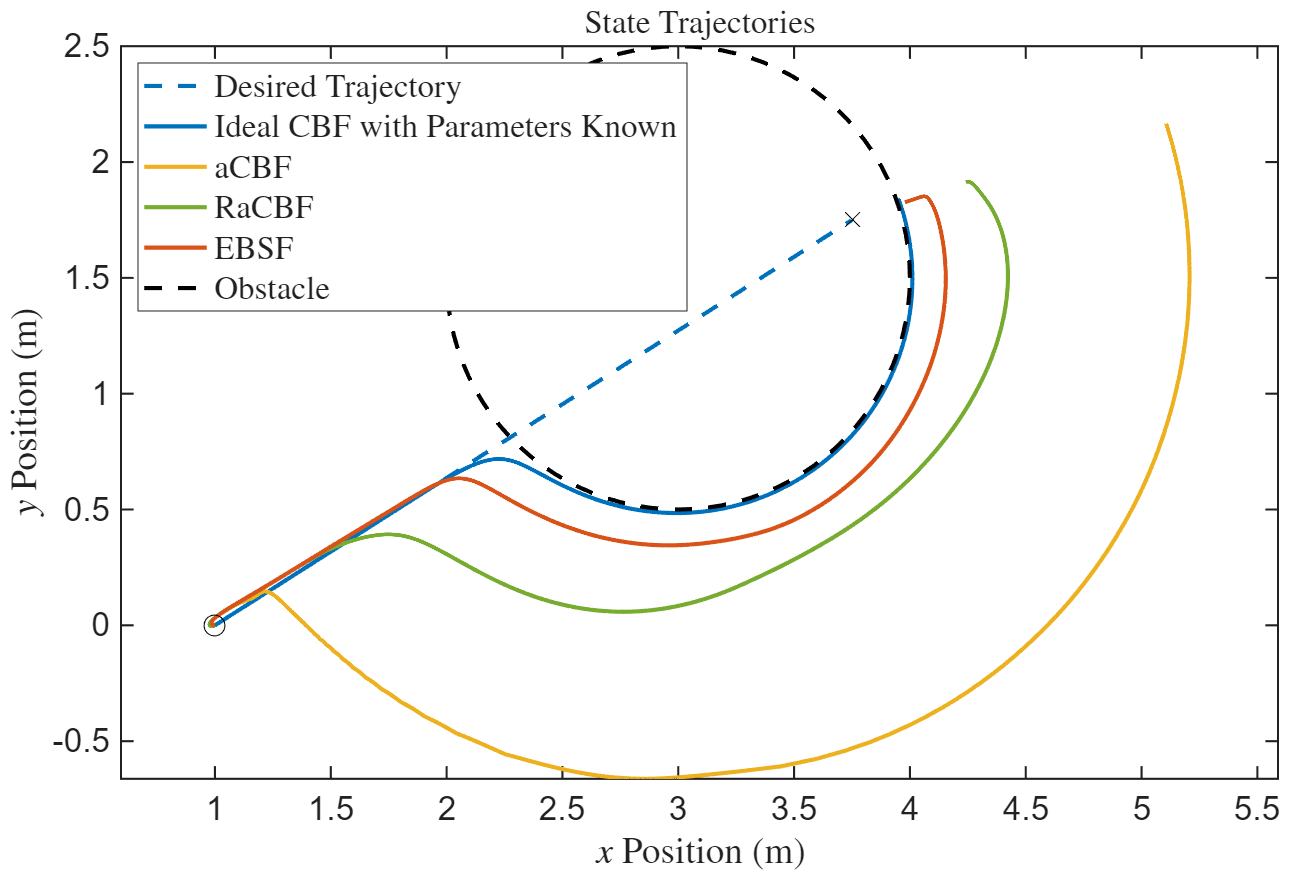}
        \caption{Trajectories resulting from aCBF \cite{taylor2020aCBFs}, RaCBF \cite{lopez2021RaCBFs}, and EBSF with SMID for Problem \ref{prob:input_matrix} compared to the ideal controller with all parameters known.}
        \label{fig:all_trajectories_2_smid}
    \end{minipage}
    \hfill
    \begin{minipage}{3.5in}
        \centering
        \includegraphics[width=3.5in]{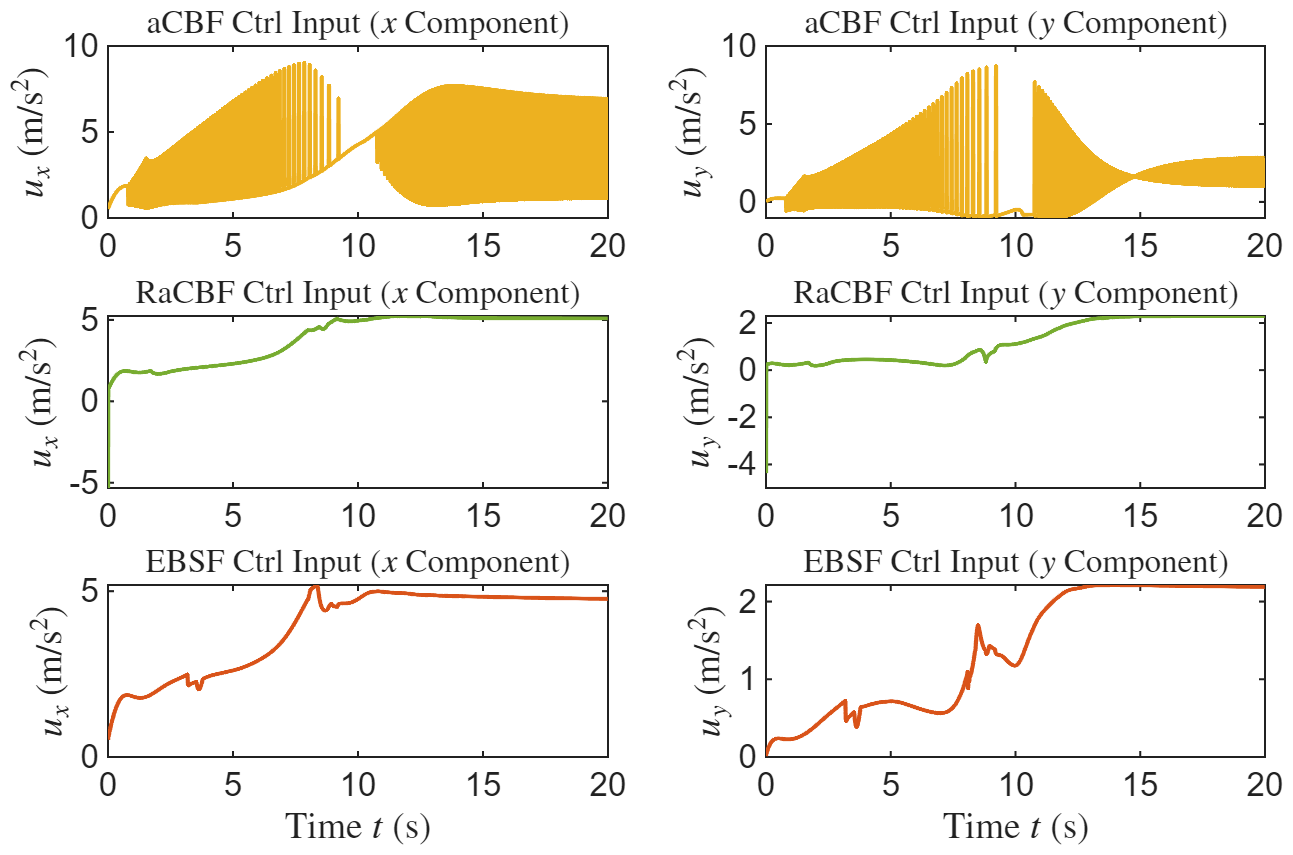}
        \caption{\minorchange{Control inputs $\vec{u}$} from aCBF \cite{taylor2020aCBFs}, RaCBF \cite{lopez2021RaCBFs}, and EBSF with SMID for Problem \ref{prob:input_matrix}.}
        \label{fig:all_inputs_2_smid}
    \end{minipage}
\end{figure}
\fi

%% file: Body/Conclusions.tex


In this paper, we consider safety and stability of adaptive control of a class of feedback linearizable plants with matched parametric uncertainties whose states are accessible, subject to a state constraint. The adaptive control architectures proposed are shown to guarantee stability, ensure control performance, and maintain safety even with the parametric uncertainties. Two problems are considered, where the first includes parametric uncertainties in the system Jacobian, while the second assumes additional uncertainties in the input matrix. In both cases, the control barrier function is assumed to have an arbitrary relative degree. No excitation conditions are imposed on the command signal. Simulation results demonstrate the non-conservatism of all of the theoretical developments.

Considerable work remains in the direction of ensuring simultaneous satisfaction of control performance and safety. How the proposed methods can be extended to plants with output feedback and general nonlinear systems that are not feedback linearizable are important directions for future work. Accommodation of input constraints due to magnitude and rate limits needs to be addressed as well.
\minorchange{More immediately, further study is needed to quantify the impact of $\omega(t)$ on EBSF and eliminate chatter.}

%% file: Appendix/No_Input_Uncertainty.tex
\section{}
\label{app:unforced_dynamics}

\subsection{Proof of Theorem \ref{thm:lyapunov_no_input_uncertainty}}
\label{app:lyapunov_no_input_uncertainty}


We propose the following Lyapunov function candidate, which is standard in adaptive control \cite{Narendra05}:
\begin{equation} \label{eqn:lyapunov_function}
    V = \vec{e}_x^\top P\vec{e}_x + \frac{\|\tilde{\theta}\|^2}{\gamma_\theta}.
\end{equation}
Since $\theta_*$ is assumed constant, we have $\dot{\tilde{\theta}} = \dot{\hat{\theta}}$.
Then, using \eqref{eqn:lyapunov_eqn} and \eqref{eqn:error_model}, one may show through straightforward algebra that
\begin{equation} \label{eqn:lyap_decay_rate_EBSB}
    \dot{V} = -\vec{e}_x^\top Q\vec{e}_x \leq 0
\end{equation}
and thus that $V(t) \leq V(0)\ \forall t \geq 0$. Claims (a) and (b) follow.

If $\vec{r}_s(t)$ is bounded, then from \eqref{eqn:reference_model} and Hurwitzness of $A_m$, we know that $\vec{x}_m(t)$ is bounded. It follows from boundedness of $\vec{e}_x(t)$ that $\vec{x}(t)$ is bounded, thus giving Claim (c). Finally, \eqref{eqn:lyap_decay_rate_EBSB} shows that $\vec{e}_x(t) \in L^2$, and \eqref{eqn:error_model} with boundedness of $\vec{e}_x(t)$, $\tilde{\theta}(t)$, $\tilde{\vartheta}(t)$, $\vec{x}(t)$, and $\vec{r}_s(t)$ shows that $\dot{\vec{e}}_x(t)$ is bounded. Claim (e) follows from Barbalat's lemma \cite{Narendra05}.
\hfill\qed

\subsection{Results of Assumption for Problem \ref{prob:unforced_dynamics}} \label{app:hbar_CBF}

Here, before proving Theorem \ref{thm:safety_objective_no_input_uncertainty}, we provide useful lemmas that result from Assumptions \ref{asn:hbar_CBF} and \ref{asn:gradients_unopposed_at_boundary}.

\begin{lemma} \label{lem:hbar_CBF_for_reference_model}
    Under Assumption \ref{asn:hbar_CBF}, there exists a constant $d_4 > 0$ such that $z_m(\vec{x}, t) \geq \alpha_r(\delta + \frac{d_1}{2})$ for every $t \geq 0$ and every $\vec{x} \in \bar{S} \cap \bar{\calC}$ where $\|B^\top\vec{g}(\vec{x})\| < d_4$.
\end{lemma}
\proof{
Assumption \ref{asn:hbar_CBF} together with Definition \ref{def:delta_strengthened_cbf} implies that, for all $\vec{x} \in \bar{S}$, there exists an $\vec{r} \in \bbR^m$ satisfying
\begin{equation*}
    \vec{g}(\vec{x})^\top(A_m\vec{x} + B\vec{r}) \geq -\alpha_r(\bar{h}(\vec{x}) - \delta - d_1).
\end{equation*}
It follows that, at any $\vec{x} \in \bar{S}$ where $\|B^\top\vec{g}(\vec{x})\| = 0$,
\begin{equation} \label{eqn:hbar_cbf_autonomous}
    \alpha_r\bar{h}(\vec{x}) + \vec{g}(\vec{x})^\top A_m\vec{x} \geq \alpha_r(\delta + d_1)
\end{equation}
holds. Finally, the claim follows from \eqref{eqn:z_m_ebsb}, boundedness of $\vec{r}_*(t)$ and $\bar{\calC}$, and smoothness of $\bar{h}(\vec{x})$ and $\vec{g}(\vec{x})$.
\hfill\qed
}

\if \numsides 2
\begin{lemma} \label{lem:hbar_CBF_for_plant}
    Under Assumptions \ref{asn:theta_in_set} and \ref{asn:hbar_CBF}, there exists a constant $d_5 > 0$ such that $z_m(\vec{x}, t) + \min_{\theta \in \Theta} \vec{g}(\vec{x})^\top BF(\vec{x})(\hat{\theta}(t) - \theta) \geq \alpha_r\delta$ for every $t \geq 0$ and every $\vec{x} \in \bar{S} \cap \bar{\calC}$ where $\|B^\top\vec{g}(\vec{x})\| < d_5$.
\end{lemma}
\else
\begin{lemma} \label{lem:hbar_CBF_for_plant}
    Under Assumptions \ref{asn:theta_in_set} and \ref{asn:hbar_CBF}, there exists a constant $d_5 > 0$ such that \\
    $z_m(\vec{x}, t) + \min_{\theta \in \Theta} \vec{g}(\vec{x})^\top BF(\vec{x})(\hat{\theta}(t) - \theta) \geq \alpha_r\delta$ for every $t \geq 0$ and every $\vec{x} \in \bar{S} \cap \bar{\calC}$ where $\|B^\top\vec{g}(\vec{x})\| < d_5$.
\end{lemma}
\fi
\proof{
Follows identically to the proof of Lemma \ref{lem:hbar_CBF_for_reference_model}, but additionally using \eqref{eqn:z_p_ebsb}, the boundedness of $F(\vec{x})$ for bounded $\vec{x}$, boundedness of $\hat{\theta}(t)$ from Theorem \ref{thm:lyapunov_no_input_uncertainty}(a), and boundedness of $\Theta$ from Assumption \ref{asn:theta_in_set}.
\hfill\qed
}

\begin{lemma} \label{lem:gradients_unopposed_at_boundary}
    Under Assumption \ref{asn:gradients_unopposed_at_boundary}, there exist constants $d_2 > \delta$, $d_3 > -1$ such that, for all $\vec{x}_1, \vec{x}_2 \in S_1 \cap \cdots \cap S_{r-1} \cap \bar{S}$ with $\bar{h}(\vec{x}_1), \bar{h}(\vec{x}_2) \leq d_2$ and $\|\vec{x}_1 - \vec{x}_2\| \leq E_1$, the inequality $\vec{g}(\vec{x}_1)^\top BB^\top\vec{g}(\vec{x}_2) \geq d_3\|B^\top\vec{g}(\vec{x}_1)\|\|B^\top\vec{g}(\vec{x}_2)\|$ holds.
\end{lemma}
\proof{
Immediate from the contrapositive of Assumption \ref{asn:gradients_unopposed_at_boundary} and smoothness of $\bar{h}$.
\hfill\qed
}

\subsection{Proof of Theorem \ref{thm:safety_objective_no_input_uncertainty}}
\label{app:safety_objective_no_input_uncertainty}

\noindent\textbf{Proof of Claim (a)}

To prove Claim (a), due to the strong convexity of the cost function in \eqref{eqn:J}, it is sufficient to show that a feasible solution to \eqref{eqn:safe_lagrange_EBSB} exists whenever $\bar{h}(\vec{x}) \geq 0$, $\bar{h}(\vec{x}_m) \geq \delta$, and $\vec{x}, \vec{x}_m \in \calC_0$. To this end, we will consider four overall cases:
\begin{enumerate}
    \item $\bar{h}(\vec{x}) \geq d_2$ and $\bar{h}(\vec{x}_m) \geq d_2$,
    \item $\bar{h}(\vec{x}) \geq d_2$ and $\bar{h}(\vec{x}_m) < d_2$, and
    \item $\bar{h}(\vec{x}) < d_2$ and $\bar{h}(\vec{x}_m) \geq d_2$, and
    \item $\bar{h}(\vec{x}) < d_2$ and $\bar{h}(\vec{x}_m) < d_2$,
\end{enumerate}
where $d_2 > \delta$ is the constant from Lemma \ref{lem:gradients_unopposed_at_boundary}. \\

\noindent\underline{Case 1) $\bar{h}(\vec{x}) \geq d_2$ and $\bar{h}(\vec{x}_m) \geq d_2$}

In this case, dropping the dependencies of $\rho_m$, $\rho_p$, and $z_m$ on $\mu_m$, $\mu_p$, $\vec{x}$, $\vec{x}_m$, and $t$ for simplicity, $\nu_m$ and $\nu_p$ satisfying
\begin{subequations}
    \begin{gather}
        (\bar{h}(\vec{x}_m) - \delta)\nu_m = \alpha_r\delta - z_m - \rho_m, \label{eqn:feasible_num_case_1} \\
        \bar{h}(\vec{x})^2\nu_p = -\bar{h}(\vec{x})z_m + \max_{\theta \in \Theta} e_hz_p(\vec{x}, t, \theta) - \bar{h}(\vec{x})\rho_m + e_h\rho_p \label{eqn:feasible_nup_case_1}
    \end{gather}
\end{subequations}
satisfy \eqref{eqn:safe_lagrange_EBSB_const1} and \eqref{eqn:safe_lagrange_EBSB_const2} for any $\mu_m$ and $\mu_p$. Thus, $\mu_m = \mu_p = 0$ and $\nu_m$, $\nu_p$ as in \eqref{eqn:feasible_num_case_1}-\eqref{eqn:feasible_nup_case_1} is a feasible solution to \eqref{eqn:safe_lagrange_EBSB}. \\

\noindent\underline{Case 2) $\bar{h}(\vec{x}) \geq d_2$ and $\bar{h}(\vec{x}_m) < d_2$}

In this case, $\nu_p$ as in \eqref{eqn:feasible_nup_case_1} satisfies \eqref{eqn:safe_lagrange_EBSB_const2} for any $\mu_m$ and $\mu_p$. Then, to show that $\mu_m$ and $\mu_p$ exist satisfying \eqref{eqn:safe_lagrange_EBSB_const1}, we will consider two sub-cases:
\begin{enumerate}
    \item[2a)] $\|B^\top\vec{g}(\vec{x}_m)\| < d_4$, and
    \item[2b)] $\|B^\top\vec{g}(\vec{x}_m)\| \geq d_4$,
\end{enumerate}
where $d_4 > 0$ is the constant from Lemma \ref{lem:hbar_CBF_for_reference_model}. \\

\noindent\underline{Case 2a) $\|B^\top\vec{g}(\vec{x}_m)\| < d_4$}

In this sub-case, it follows immediately from Lemma \ref{lem:hbar_CBF_for_reference_model} and \eqref{eqn:z_m_ebsb} that $\mu_m = \mu_p = \nu_m = 0$ satisfies \eqref{eqn:safe_lagrange_EBSB_const1}. Thus, $\mu_m = \mu_p = \nu_m = 0$ and $\nu_p$ as in \eqref{eqn:feasible_nup_case_1} is a feasible solution to \eqref{eqn:safe_lagrange_EBSB}. \\

\noindent\underline{Case 2b) $\|B^\top\vec{g}(\vec{x}_m)\| \geq d_4$}

In this sub-case, it follows from \eqref{eqn:safe_lagrange_EBSB}-\eqref{eqn:rho_m} that $\mu_p = \nu_m = 0$ and $\mu_m$ satisfying
\begin{equation} \label{eqn:feasible_mum_case_2}
    \|B^\top\vec{g}(\vec{x}_m)\|^2\mu_m = \alpha_r\delta - z_m(\vec{x}_m, t)
\end{equation}
satisfy \eqref{eqn:safe_lagrange_EBSB_const1}. Thus, $\mu_p = \nu_m = 0$, $\mu_m$ as in \eqref{eqn:feasible_mum_case_2}, and $\nu_p$ as in \eqref{eqn:feasible_nup_case_1} is a feasible solution to \eqref{eqn:safe_lagrange_EBSB}. \\

\noindent\underline{Case 3) $\bar{h}(\vec{x}) < d_2$ and $\bar{h}(\vec{x}_m) \geq d_2$}

In this case, $\nu_m$ as in \eqref{eqn:feasible_num_case_1} satisfies \eqref{eqn:safe_lagrange_EBSB_const1} for any $\mu_m$ and $\mu_p$. Additionally, this case implies $e_h < 0$, and thus, using \eqref{eqn:z_m_ebsb}-\eqref{eqn:z_p_ebsb}, the right-hand side of \eqref{eqn:safe_lagrange_EBSB_const2} can be rewritten as
\if \numsides 2
\begin{equation} \label{eqn:constraint_2_rhs_rewritten}
    \begin{aligned}
        &-\bar{h}(\vec{x})z_m(\vec{x}_m, t) + \max_{\theta \in \Theta} e_hz_p(\vec{x}, t, \theta) \\
        &\indenti{sp}= -\bar{h}(\vec{x})z_m(\vec{x}_m, t) + e_h\min_{\theta \in \Theta} z_p(\vec{x}, t, \theta) \\
        &\indenti{sp}= -\bar{h}(\vec{x})z_m(\vec{x}_m, t) + e_hz_m(\vec{x}, t) \\
        &\indenti{sp=} + e_h\min_{\theta \in \Theta} \vec{g}(\vec{x})^\top BF(\vec{x})(\hat{\theta} - \theta).
    \end{aligned}
\end{equation}
\else
\begin{equation} \label{eqn:constraint_2_rhs_rewritten}
    \begin{aligned}
        &-\bar{h}(\vec{x})z_m(\vec{x}_m, t) + \max_{\theta \in \Theta} e_hz_p(\vec{x}, t, \theta) \\
        &\indenti{sp}= -\bar{h}(\vec{x})z_m(\vec{x}_m, t) + e_h\min_{\theta \in \Theta} z_p(\vec{x}, t, \theta) \\
        &\indenti{sp}= -\bar{h}(\vec{x})z_m(\vec{x}_m, t) + e_hz_m(\vec{x}, t) + e_h\min_{\theta \in \Theta} \vec{g}(\vec{x})^\top BF(\vec{x})(\hat{\theta} - \theta).
    \end{aligned}
\end{equation}
\fi
Furthermore, note that, as $\vec{r}_*(t)$ is bounded, $\bar{h}$ is smooth, and $\bar{\calC}$ is bounded, the quantities $\sup_{\vec{x} \in \bar{\calC}} \|B^\top\vec{g}(\vec{x})\|$ and $\sup_{\vec{x} \in \bar{\calC}, t \geq 0} |z_m(\vec{x}, t)|$ exist and are finite.
Then, to show that $\mu_m$, $\mu_p$, and $\nu_p$ exist satisfying \eqref{eqn:safe_lagrange_EBSB_const2}, we will consider three sub-cases:
\begin{enumerate}
    \item[3a)] $\bar{h}(\vec{x}) \geq d_6$,
    \item[3b)] $\bar{h}(\vec{x}) < d_6$ and $\|B^\top\vec{g}(\vec{x})\| < d_5$, and
    \item[3c)] $\bar{h}(\vec{x}) < d_6$ and $\|B^\top\vec{g}(\vec{x})\| \geq d_5$,
\end{enumerate}
where $d_5 > 0$ is the constant from Lemma \ref{lem:hbar_CBF_for_plant}, and $d_6 := \min\{\frac{d_2\alpha_r\delta}{\alpha_r\delta + \sup_{\vec{x} \in \bar{\calC}, t \geq 0} |z_m(\vec{x}, t)|}, \frac{d_2d_5}{2(d_5 + \sup_{\vec{x} \in \bar{\calC}} \|B^\top\vec{g}(\vec{x})\|)}\}$. \\

\noindent\underline{Case 3a) $\bar{h}(\vec{x}) \geq d_6$}

In this sub-case, $\nu_p$ as in \eqref{eqn:feasible_nup_case_1} satisfies \eqref{eqn:safe_lagrange_EBSB_const2} for any $\mu_m$ and $\mu_p$. Thus, $\mu_p = \mu_m = 0$ and $\nu_m, \nu_p$ as in \eqref{eqn:feasible_num_case_1}-\eqref{eqn:feasible_nup_case_1} is a feasible solution to \eqref{eqn:safe_lagrange_EBSB}. \\

\noindent\underline{Case 3b) $\bar{h}(\vec{x}) < d_6$ and $\|B^\top\vec{g}(\vec{x})\| < d_5$}

In this sub-case, we have $e_h \leq d_6 - d_2 < 0$. Then, Lemma \ref{lem:hbar_CBF_for_plant} with \eqref{eqn:constraint_2_rhs_rewritten} implies that
\if \numsides 2
\begin{equation*}
    \begin{aligned}
        &-\bar{h}(\vec{x})z_m(\vec{x}_m, t) + e_hz_m(\vec{x}, t) \\
        &\indenti{sp} + e_h\min_{\theta \in \Theta} \vec{g}(\vec{x})^\top BF(\vec{x})(\hat{\theta} - \theta) \\
        %
        %
        &\indenti{spsp} \leq (d_6 - d_2)\alpha_r\delta + d_6\sup_{\vec{x} \in \bar{\calC}, t \geq 0} |z_m(\vec{x}_m, t)| \\
        &\indenti{spsp} \leq 0.
    \end{aligned}
\end{equation*}
\else
\begin{equation*}
    \begin{aligned}
        &-\bar{h}(\vec{x})z_m(\vec{x}_m, t) + e_hz_m(\vec{x}, t) + e_h\min_{\theta \in \Theta} \vec{g}(\vec{x})^\top BF(\vec{x})(\hat{\theta} - \theta) \\
        %
        %
        &\indenti{spspspsp} \leq (d_6 - d_2)\alpha_r\delta + d_6\sup_{\vec{x} \in \bar{\calC}, t \geq 0} |z_m(\vec{x}_m, t)| \\
        &\indenti{spspspsp} \leq 0.
    \end{aligned}
\end{equation*}
\fi
Therefore, $\mu_m = \mu_p = \nu_p = 0$ satisfies \eqref{eqn:safe_lagrange_EBSB_const2}, and $\mu_m = \mu_p =$ $\nu_p = 0$ and $\nu_m$ as in \eqref{eqn:feasible_num_case_1} is a feasible solution to \eqref{eqn:safe_lagrange_EBSB}. \\

\noindent\underline{Case 3c) $\bar{h}(\vec{x}) < d_6$ and $\|B^\top\vec{g}(\vec{x})\| \geq d_5$}

In this sub-case, it follows from \eqref{eqn:safe_lagrange_EBSB}-\eqref{eqn:rho_p} that $\mu_m = \nu_p = 0$ and $\mu_p$ satisfying
\if \numsides 2
\begin{equation} \label{eqn:feasible_mup_case_3}
    \begin{aligned}
        &(\bar{h}(\vec{x})\vec{g}(\vec{x}_m)^\top BB^\top\vec{g}(\vec{x}) - e_h\|B^\top\vec{g}(\vec{x})\|^2)\mu_p \\
        &\indenti{sp} = -\bar{h}(\vec{x})z_m(\vec{x}_m, t) + \max_{\theta \in \Theta} e_hz_p(\vec{x}, t, \theta)
    \end{aligned}
\end{equation}
\else
\begin{equation} \label{eqn:feasible_mup_case_3}
    (\bar{h}(\vec{x})\vec{g}(\vec{x}_m)^\top BB^\top\vec{g}(\vec{x}) - e_h\|B^\top\vec{g}(\vec{x})\|^2)\mu_p = -\bar{h}(\vec{x})z_m(\vec{x}_m, t) + \max_{\theta \in \Theta} e_hz_p(\vec{x}, t, \theta)
\end{equation}
\fi
satisfy \eqref{eqn:safe_lagrange_EBSB_const2}. Such a $\mu_p$ exists, because
\begin{align*}
    &\bar{h}(\vec{x})\vec{g}(\vec{x}_m)^\top BB^\top\vec{g}(\vec{x}) - e_h\|B^\top\vec{g}(\vec{x})\|^2 \\
    &\indenti{sp}\geq -\bar{h}(\vec{x})\|B^\top\vec{g}(\vec{x}_m)\|\|B^\top\vec{g}(\vec{x})\| - e_h\|B^\top\vec{g}(\vec{x})\|^2 \\
    &\indenti{sp}= \|B^\top\vec{g}(\vec{x})\|(-\bar{h}(\vec{x})\|B^\top\vec{g}(\vec{x}_m)\| - e_h\|B^\top\vec{g}(\vec{x})\|) \\
    &\indenti{sp}\geq \|B^\top\vec{g}(\vec{x})\|\left((d_2 - d_6)d_5 - d_6\sup_{\vec{x} \in \bar{\calC}}\|B^\top\vec{g}(\vec{x})\|\right) \\
    &\indenti{sp}\geq \|B^\top\vec{g}(\vec{x})\|d_2d_5/2 \\
    &\indenti{sp}\geq d_2d_5^2/2 > 0.
    %
\end{align*}
Thus, $\mu_m = \nu_p = 0$, $\mu_p$ as in \eqref{eqn:feasible_mup_case_3}, and $\nu_m$ as in \eqref{eqn:feasible_num_case_1} is a feasible solution to \eqref{eqn:safe_lagrange_EBSB}. \\

\noindent\underline{Case 4) $\bar{h}(\vec{x}) < d_2$ and $\bar{h}(\vec{x}_m) < d_2$}

In this case, we will immediately consider the following two sub-cases:
\begin{enumerate}
    \item[4a)] $e_h \geq 0$, and
    \item[4b)] $e_h < 0$.
\end{enumerate}
\vspace{\baselineskip}

\noindent\underline{Case 4a) $e_h \geq 0$}

In this sub-case, as $\bar{h}(\vec{x}_m) \geq \delta$ by assumption, we also have $\bar{h}(\vec{x}) \geq \delta$. Thus, $\nu_p$ as in \eqref{eqn:feasible_nup_case_1} satisfies \eqref{eqn:safe_lagrange_EBSB_const2} for any $\mu_m$ and $\mu_p$. The remainder of this sub-case proceeds identically to Case 2 above, resulting in the feasible solution $\mu_p = \nu_m = 0$, $\nu_p$ as in \eqref{eqn:feasible_nup_case_1}, and either $\mu_m = 0$ or $\mu_m$ as in \eqref{eqn:feasible_mum_case_2} depending on $\|B^\top\vec{g}(\vec{x}_m)\|$. \\

\noindent\underline{Case 4b) $e_h < 0$}

In this sub-case, as $e_h < 0$, \eqref{eqn:constraint_2_rhs_rewritten} holds. Additionally, we note that the coordinate transformation from $(\mu_m, \mu_p)$ to $(\rho_m, \rho_p)$ is given by
\begin{subequations}
    \begin{gather}
        \begin{bmatrix} \rho_m \\ \rho_p \end{bmatrix} = T(\vec{x}, \vec{x}_m)\begin{bmatrix} \mu_m \\ \mu_p \end{bmatrix}, \label{eqn:mu_to_rho} \\
        T(\vec{x}, \vec{x}_m) = \begin{bmatrix} \|B^\top\vec{g}(\vec{x}_m)\|^2 & \vec{g}(\vec{x}_m)^\top BB^\top\vec{g}(\vec{x}) \\ \vec{g}(\vec{x}_m)^\top BB^\top\vec{g}(\vec{x}) & \|B^\top\vec{g}(\vec{x})\|^2 \end{bmatrix}. \label{eqn:mu_to_rho_matrix}
    \end{gather}
\end{subequations}
Then, with $T(\vec{x}, \vec{x}_m)$ defined, we will consider the following four sub-sub-cases:
\begin{enumerate}
    \item[4b.i)] $\|B^\top\vec{g}(\vec{x}_m)\| < d_4$ and $\|B^\top\vec{g}(\vec{x})\| < d_5$,
    \item[4b.ii)] $\|B^\top\vec{g}(\vec{x}_m)\| \geq d_4$ and $\|B^\top\vec{g}(\vec{x})\| < d_5$,
    \item[4b.iii)] $\|B^\top\vec{g}(\vec{x}_m)\| < d_4$ and $\|B^\top\vec{g}(\vec{x})\| \geq d_5$,
    \item[4b.iv)] $\|B^\top\vec{g}(\vec{x}_m)\| \geq d_4$ and $\|B^\top\vec{g}(\vec{x})\| \geq d_5$.
\end{enumerate}
\vspace{\baselineskip}

\noindent\underline{Case 4b.i) $\|B^\top\vec{g}(\vec{x}_m)\| < d_4$ and $\|B^\top\vec{g}(\vec{x})\| < d_5$}

In this sub-sub-case, it follows immediately from Lemma \ref{lem:hbar_CBF_for_reference_model} and Lemma \ref{lem:hbar_CBF_for_plant} with \eqref{eqn:constraint_2_rhs_rewritten} that $\mu_m = \mu_p = \nu_m = \nu_p = 0$ is a feasible solution to \eqref{eqn:safe_lagrange_EBSB}. \\

\noindent\underline{Case 4b.ii) $\|B^\top\vec{g}(\vec{x}_m)\| \geq d_4$ and $\|B^\top\vec{g}(\vec{x})\| < d_5$}

In this sub-sub-case, it is immediate from \eqref{eqn:safe_lagrange_EBSB} and \eqref{eqn:constraint_2_rhs_rewritten} with Lemma \ref{lem:hbar_CBF_for_plant} that $\nu_m = \nu_p = 0$ and $\mu_m$ and $\mu_p$ satisfying
\begin{subequations}
    \begin{align}
        \rho_m(\mu_m, \mu_p, \vec{x}, \vec{x}_m) &\geq \alpha_r\delta - \min\{z_m(\vec{x}_m, t), 0\}, \label{eqn:feasible_rhom_case_4bii} \\
        \rho_p(\mu_m, \mu_p, \vec{x}, \vec{x}_m) &\geq -\alpha_r\delta \label{eqn:feasible_rhop_case_4bii}
    \end{align}
\end{subequations}
is a feasible solution for \eqref{eqn:safe_lagrange_EBSB}. What remains is to prove that $\mu_m$ and $\mu_p$ satisfying \eqref{eqn:feasible_rhom_case_4bii}-\eqref{eqn:feasible_rhop_case_4bii} exist, which will be done as follows.

First, suppose that either $\|B^\top\vec{g}(\vec{x})\| \leq \frac{d_4\alpha_r\delta}{\rho_m}$ for $\rho_m$ in \eqref{eqn:feasible_rhom_case_4bii}, or $\vec{g}(\vec{x}_m)^\top BB^\top\vec{g}(\vec{x}) \geq 0$. Then, one may choose $\mu_p = 0$ and $\mu_m = \rho_m/\|B^\top\vec{g}(\vec{x}_m)\|^2$ for $\rho_m$ in \eqref{eqn:feasible_rhom_case_4bii}, and this choice will also result in $\rho_p \geq -\alpha_r\delta$. Second, suppose that $\|B^\top\vec{g}(\vec{x})\| \geq \frac{d_4\alpha_r\delta}{\rho_m}$ for $\rho_m$ in \eqref{eqn:feasible_rhom_case_4bii}, and $\vec{g}(\vec{x}_m)^\top BB^\top\vec{g}(\vec{x}) \leq 0$. Then, as $\rho_m \in [\alpha_r\delta, \alpha_r\delta + \sup_{\vec{x} \in \bar{\calC}, t \geq 0} |z_m(\vec{x}, t)|]$, we have $\|B^\top\vec{g}(\vec{x})\| > 0$, and Lemma \ref{lem:gradients_unopposed_at_boundary} implies that $T(\vec{x}, \vec{x}_m)$ is invertible with bounded condition number. Therefore, one may choose $[\mu_m, \mu_p]^\top = T(\vec{x}, \vec{x}_m)^{-1}[\rho_m, \rho_p]^\top$ for any $\rho_m$, $\rho_p$ satisfying \eqref{eqn:feasible_rhom_case_4bii}-\eqref{eqn:feasible_rhop_case_4bii}. Finally, as these two possibilities cover the full range of possible scenarios under this sub-sub-case, we have proven that feasible solutions always exist under this sub-sub-case. \\

\noindent\underline{Case 4b.iii) $\|B^\top\vec{g}(\vec{x}_m)\| < d_4$ and $\|B^\top\vec{g}(\vec{x})\| \geq d_5$}

In this sub-sub-case, it is immediate from \eqref{eqn:safe_lagrange_EBSB}, Lemma \ref{lem:hbar_CBF_for_reference_model}, and \eqref{eqn:constraint_2_rhs_rewritten} that $\nu_m = \nu_p = 0$ and $\mu_m$ and $\mu_p$ satisfying
\begin{subequations}
    \begin{align}
        \rho_m(\mu_m, \mu_p, \vec{x}, \vec{x}_m) &\geq -\alpha_rd_1/2, \label{eqn:feasible_rhom_case_4biii} \\
        \rho_p(\mu_m, \mu_p, \vec{x}, \vec{x}_m) &\geq \max\{-\min_{\theta \in \Theta} z_p(\vec{x}, t, \theta), c\} \label{eqn:feasible_rhop_case_4biii}
    \end{align}
\end{subequations}
is a feasible solution for \eqref{eqn:safe_lagrange_EBSB} for any $c > 0$. What remains is to prove that $\mu_m$ and $\mu_p$ satisfying \eqref{eqn:feasible_rhom_case_4biii}-\eqref{eqn:feasible_rhop_case_4biii} exist, which we will do as follows.

First, suppose that either $\|B^\top\vec{g}(\vec{x}_m)\| \leq \frac{d_5\alpha_rd_1/2}{\rho_p}$ for $\rho_p$ in \eqref{eqn:feasible_rhop_case_4biii}, or $\vec{g}(\vec{x}_m)^\top BB^\top\vec{g}(\vec{x}) \geq 0$. Then, one may choose $\mu_m = 0$ and $\mu_p = \rho_p/\|B^\top\vec{g}(\vec{x})\|^2$ for $\rho_p$ in \eqref{eqn:feasible_rhop_case_4biii}, and this choice will also result in $\rho_m \geq -\alpha_rd_1/2$. Second, suppose that $\|B^\top\vec{g}(\vec{x}_m)\| \geq \frac{d_5\alpha_rd_1/2}{\rho_p}$ for $\rho_p$ in \eqref{eqn:feasible_rhop_case_4biii}, and $\vec{g}(\vec{x}_m)^\top BB^\top\vec{g}(\vec{x}) \leq 0$. Then, as $\rho_p \in [c, \sup_{\vec{x} \in \bar{\calC}, t \geq 0, \theta \in \Theta} |z_p(\vec{x}, t, \theta)|]$, we have $\|B^\top\vec{g}(\vec{x}_m)\| > 0$, and Lemma \ref{lem:gradients_unopposed_at_boundary} implies that $T(\vec{x}, \vec{x}_m)$ is invertible with bounded condition number. Therefore, one may choose $[\mu_m, \mu_p]^\top = T(\vec{x}, \vec{x}_m)^{-1}[\rho_m, \rho_p]^\top$ for any $\rho_m$, $\rho_p$ satisfying \eqref{eqn:feasible_rhom_case_4biii}-\eqref{eqn:feasible_rhop_case_4biii}. Finally, as these two possibilities cover the full range of possible scenarios under this sub-sub-case, we have proven that feasible solutions always exist under this sub-sub-case. \\

\noindent\underline{Case 4b.iv) $\|B^\top\vec{g}(\vec{x}_m)\| \geq d_4$ and $\|B^\top\vec{g}(\vec{x})\| \geq d_5$}

In this sub-sub-case, it is immediate from \eqref{eqn:safe_lagrange_EBSB} and \eqref{eqn:constraint_2_rhs_rewritten} that $\nu_m = \nu_p = 0$ and $\mu_m$ and $\mu_p$ satisfying
\begin{subequations}
    \begin{align}
        \rho_m(\mu_m, \mu_p, \vec{x}, \vec{x}_m) &\geq \alpha_r\delta - z_m(\vec{x}_m, t), \label{eqn:feasible_rhom_case_4biv} \\
        \rho_p(\mu_m, \mu_p, \vec{x}, \vec{x}_m) &\geq -\min_{\theta \in \Theta} z_p(\vec{x}, t, \theta) \label{eqn:feasible_rhop_case_4biv}
    \end{align}
\end{subequations}
is a feasible solution for \eqref{eqn:safe_lagrange_EBSB}. What remains is to prove that $\mu_m$ and $\mu_p$ satisfying \eqref{eqn:feasible_rhom_case_4biv}-\eqref{eqn:feasible_rhop_case_4biv} exist, which we will do as follows.

First, suppose that $\vec{g}(\vec{x}_m)^\top BB^\top\vec{g}(\vec{x}) \geq 0$. Then, one may choose $\mu_m = \rho_m/\|B^\top\vec{g}(\vec{x}_m)\|^2$ and $\mu_p = \rho_p/\|B^\top\vec{g}(\vec{x})\|^2$ for $\rho_m$, $\rho_p$ in \eqref{eqn:feasible_rhom_case_4biv}-\eqref{eqn:feasible_rhop_case_4biv}, and this choice will result in satisfaction of the inequalities in \eqref{eqn:feasible_rhom_case_4biv}-\eqref{eqn:feasible_rhop_case_4biv}. Second, suppose that $\vec{g}(\vec{x}_m)^\top BB^\top\vec{g}(\vec{x}) \geq 0$. Then, Lemma \ref{lem:gradients_unopposed_at_boundary} implies that $T(\vec{x}, \vec{x}_m)$ is invertible with bounded condition number. Therefore, one may choose $[\mu_m, \mu_p]^\top = T(\vec{x}, \vec{x}_m)^{-1}[\rho_m, \rho_p]^\top$ for any $\rho_m$, $\rho_p$ satisfying \eqref{eqn:feasible_rhom_case_4biv}-\eqref{eqn:feasible_rhop_case_4biv}. Finally, as these two possibilities cover the full range of possible scenarios under this sub-sub-case, we have proven that feasible solutions always exist under this sub-sub-case. \\

\noindent\underline{Combining All Cases}

As cases 1, 2a, 2b, 3a, 3b, 3c, 4a, 4b.i, 4b.ii, 4b.iii, and 4b.iv cover all possible scenarios under the conditions of Claim (a), we have shown that there always exists a feasible solution to \eqref{eqn:safe_lagrange_EBSB}. Claim (a) is thus proven. \\

\noindent\textbf{Proof of Claim (b)}

Define the composite CBF candidates
\begin{align}
    H_m(t) &= \bar{h}(\vec{x}_m(t)) - \delta, \label{eqn:H_m} \\
    H_e(t) &= \bar{h}(\vec{x}_m(t))^2 - e_h(t)^2. \label{eqn:H_e}
\end{align}
Then, from \eqref{eqn:reference_model} and \eqref{eqn:safe_reference_EBSB}-\eqref{eqn:z_m_ebsb},
\if \numsides 2
\begin{align}
    \dot{H}_m &= \vec{g}(\vec{x}_m)^\top(A_m\vec{x}_m + B\vec{r}_*) + \mu_m\|B^\top\vec{g}(\vec{x}_m)\|^2 \nonumber \\
    &\indenti{=} + \mu_p\vec{g}(\vec{x}_m)^\top BB^\top\vec{g}(\vec{x}) \nonumber \\
    &\geq \alpha_r\delta - \alpha_r\bar{h}(\vec{x}_m) - \max\{\bar{h}(\vec{x}_m) - \delta, 0\}\nu_m \nonumber \\
    &= -\alpha_rH_m - \max\{H_m, 0\}\nu_m. \label{eqn:Hdot_m}
\end{align}
\else
\begin{align}
    \dot{H}_m &= \vec{g}(\vec{x}_m)^\top(A_m\vec{x}_m + B\vec{r}_*) + \mu_m\|B^\top\vec{g}(\vec{x}_m)\|^2 + \mu_p\vec{g}(\vec{x}_m)^\top BB^\top\vec{g}(\vec{x}) \nonumber \\
    &\geq \alpha_r\delta - \alpha_r\bar{h}(\vec{x}_m) - \max\{\bar{h}(\vec{x}_m) - \delta, 0\}\nu_m \nonumber \\
    &= -\alpha_rH_m - \max\{H_m, 0\}\nu_m. \label{eqn:Hdot_m}
\end{align}
\fi
Additionally, from \eqref{eqn:plant_no_input_uncertainty}-\eqref{eqn:input} and \eqref{eqn:safe_reference_EBSB}-\eqref{eqn:z_p_ebsb} and Assumption \ref{asn:theta_in_set},
\if \numsides 2
\begin{align}
    \dot{H}_e &= 2\bar{h}(\vec{x}_m)[\vec{g}(\vec{x}_m)^\top(A_m\vec{x}_m + B\vec{r}_*) + \mu_m\|B^\top\vec{g}(\vec{x}_m)\|^2 \nonumber \\
    &\indenti{= 2\bar{h}(\vec{x}_m)[} + \mu_p\vec{g}(\vec{x}_m)^\top BB^\top\vec{g}(\vec{x})] \nonumber \\
    &\indenti{=} - 2e_h[\vec{g}(\vec{x})^\top(A_m\vec{x} + B(\vec{r}_* + F(\vec{x})\tilde{\theta})) \nonumber \\
    &\indenti{= - 2e_h[} + \mu_m\vec{g}(\vec{x})^\top BB^\top\vec{g}(\vec{x}_m) + \mu_p\|B^\top\vec{g}(\vec{x})\|^2 \nonumber \\
    &\indenti{= - 2e_h[} - \vec{g}(\vec{x}_m)^\top(A_m\vec{x}_m + B\vec{r}_*) - \mu_m\|B^\top\vec{g}(\vec{x}_m)\|^2 \nonumber \\
    &\indenti{= - 2e_h[} - \mu_p\vec{g}(\vec{x}_m)^\top BB^\top\vec{g}(\vec{x})] \nonumber \\
    %
    &= 2\bar{h}(\vec{x})(\mu_m\|B^\top\vec{g}(\vec{x}_m)\|^2 + \mu_p\vec{g}(\vec{x}_m)^\top BB^\top\vec{g}(\vec{x})) \nonumber \\
    &\indenti{=} - 2e_h(\mu_m\vec{g}(\vec{x})^\top BB^\top\vec{g}(\vec{x}_m) + \mu_p\|B^\top\vec{g}(\vec{x})\|^2) \nonumber \\
    &\indenti{=} + 2\bar{h}(\vec{x})\vec{g}(\vec{x}_m)^\top(A_m\vec{x}_m + B\vec{r}_*) \nonumber \\
    &\indenti{=} - 2e_h\vec{g}(\vec{x})^\top(A_m\vec{x} + B(\vec{r}_* + F(\vec{x})\tilde{\theta})) \nonumber \\
    &\geq 2\bar{h}(\vec{x})(\mu_m\|B^\top\vec{g}(\vec{x}_m)\|^2 + \mu_p\vec{g}(\vec{x}_m)^\top BB^\top\vec{g}(\vec{x})) \nonumber \\
    &\indenti{=} - 2e_h(\mu_m\vec{g}(\vec{x})^\top BB^\top\vec{g}(\vec{x}_m) + \mu_p\|B^\top\vec{g}(\vec{x})\|^2) \nonumber \\
    &\indenti{=} + 2\bar{h}(\vec{x})\vec{g}(\vec{x}_m)^\top(A_m\vec{x}_m + B\vec{r}_*) \nonumber \\
    &\indenti{=} - 2\max_{\theta \in \Theta}e_h\vec{g}(\vec{x})^\top(A_m\vec{x} + B(\vec{r}_* + F(\vec{x})(\hat{\theta} - \theta))) \nonumber \\
    &\geq 2\bar{h}(\vec{x})(-\alpha_r\bar{h}(\vec{x}_m) + \alpha_re_h - \max\{\bar{h}(\vec{x}), 0\}\nu_p) \nonumber \\
    &= -2\alpha_rH_e - 2\bar{h}(\vec{x})\max\{\bar{h}(\vec{x}), 0\}\nu_p. \label{eqn:Hdot_e}
    %
\end{align}
\else
\begin{align}
    \dot{H}_e &= 2\bar{h}(\vec{x}_m)[\vec{g}(\vec{x}_m)^\top(A_m\vec{x}_m + B\vec{r}_*) + \mu_m\|B^\top\vec{g}(\vec{x}_m)\|^2 + \mu_p\vec{g}(\vec{x}_m)^\top BB^\top\vec{g}(\vec{x})] \nonumber \\
    &\indenti{=} - 2e_h[\vec{g}(\vec{x})^\top(A_m\vec{x} + B(\vec{r}_* + F(\vec{x})\tilde{\theta})) + \mu_m\vec{g}(\vec{x})^\top BB^\top\vec{g}(\vec{x}_m) + \mu_p\|B^\top\vec{g}(\vec{x})\|^2 \nonumber \\
    &\indenti{= - 2e_h[} - \vec{g}(\vec{x}_m)^\top(A_m\vec{x}_m + B\vec{r}_*) - \mu_m\|B^\top\vec{g}(\vec{x}_m)\|^2 - \mu_p\vec{g}(\vec{x}_m)^\top BB^\top\vec{g}(\vec{x})] \nonumber \\
    %
    &= 2\bar{h}(\vec{x})(\mu_m\|B^\top\vec{g}(\vec{x}_m)\|^2 + \mu_p\vec{g}(\vec{x}_m)^\top BB^\top\vec{g}(\vec{x})) - 2e_h(\mu_m\vec{g}(\vec{x})^\top BB^\top\vec{g}(\vec{x}_m) + \mu_p\|B^\top\vec{g}(\vec{x})\|^2) \nonumber \\
    &\indenti{=} + 2\bar{h}(\vec{x})\vec{g}(\vec{x}_m)^\top(A_m\vec{x}_m + B\vec{r}_*) - 2e_h\vec{g}(\vec{x})^\top(A_m\vec{x} + B(\vec{r}_* + F(\vec{x})\tilde{\theta})) \nonumber \\
    &\geq 2\bar{h}(\vec{x})(\mu_m\|B^\top\vec{g}(\vec{x}_m)\|^2 + \mu_p\vec{g}(\vec{x}_m)^\top BB^\top\vec{g}(\vec{x})) - 2e_h(\mu_m\vec{g}(\vec{x})^\top BB^\top\vec{g}(\vec{x}_m) + \mu_p\|B^\top\vec{g}(\vec{x})\|^2) \nonumber \\
    &\indenti{=} + 2\bar{h}(\vec{x})\vec{g}(\vec{x}_m)^\top(A_m\vec{x}_m + B\vec{r}_*) - 2\max_{\theta \in \Theta}e_h\vec{g}(\vec{x})^\top(A_m\vec{x} + B(\vec{r}_* + F(\vec{x})(\hat{\theta} - \theta))) \nonumber \\
    &\geq 2\bar{h}(\vec{x})(-\alpha_r\bar{h}(\vec{x}_m) + \alpha_re_h - \max\{\bar{h}(\vec{x}), 0\}\nu_p) \nonumber \\
    &= -2\alpha_rH_e - 2\bar{h}(\vec{x})\max\{\bar{h}(\vec{x}), 0\}\nu_p. \label{eqn:Hdot_e}
    %
\end{align}
\fi
Equations \eqref{eqn:Hdot_m} and \eqref{eqn:Hdot_e} hold whenever there exists a feasible solution to \eqref{eqn:safe_lagrange_EBSB}, which, from Claim (a), occurs whenever $\vec{x}, \vec{x}_m \in \calC_0$ and $\bar{h}(\vec{x}_m) \geq \delta$.

Furthermore, recall that $S_1$, $\dots$, $S_{r-1}$ are defined by $h_i, i \in [1, r)$ as given in \eqref{eqn:hocbf_base_case}-\eqref{eqn:hocbf_recursion}. Definition \ref{def:relative_degree} implies that $\frac{\partial h_i}{\partial\vec{x}}B = 0\ \forall i < r$. Then, using \eqref{eqn:plant_no_input_uncertainty}, \eqref{eqn:reference_model}, and \eqref{eqn:hocbf_recursion}, we have
\if \numsides 2
\begin{align}
    \frac{d}{dt}h_i(\vec{x}) &= \frac{\partial h_i}{\partial\vec{x}}\Big|_{\vec{x}}(A\vec{x} + B(\vec{u} - F(\vec{x})\theta_*)) \nonumber \\
    &= h_{i+1}(\vec{x}) - \alpha_ih_i(\vec{x})\ \forall i \in [1, r), \label{eqn:hdot_pi} \\
    \frac{d}{dt}h_i(\vec{x}_m) &= \frac{\partial h_i}{\partial\vec{x}}\Big|_{\vec{x}_m}((A + BK)\vec{x}_m + B\vec{r}_s) \nonumber \\
    &= h_{i+1}(\vec{x}_m) - \alpha_ih_i(\vec{x}_m)\ \forall i \in [1, r). \label{eqn:hdot_mi}
\end{align}
\else
\begin{align}
    \frac{d}{dt}h_i(\vec{x}) &= \frac{\partial h_i}{\partial\vec{x}}\Big|_{\vec{x}}(A\vec{x} + B(\vec{u} - F(\vec{x})\theta_*)) = h_{i+1}(\vec{x}) - \alpha_ih_i(\vec{x})\ \forall i \in [1, r), \label{eqn:hdot_pi} \\
    \frac{d}{dt}h_i(\vec{x}_m) &= \frac{\partial h_i}{\partial\vec{x}}\Big|_{\vec{x}_m}((A + BK)\vec{x}_m + B\vec{r}_s) = h_{i+1}(\vec{x}_m) - \alpha_ih_i(\vec{x}_m)\ \forall i \in [1, r). \label{eqn:hdot_mi}
\end{align}
\fi
Similarly, for each function $b_j, j \in [1, N_\calC]$ in Property \ref{ppty:hbar}(c) with relative degree $r_j > 1$ w.r.t \eqref{eqn:general_plant}, we have
\begin{align}
    \frac{d}{dt}b_j(\vec{x}) &= \frac{\partial b_j}{\partial\vec{x}}\Big|_{\vec{x}}(A\vec{x} + B(\vec{u} - F(\vec{x})\theta_*)) = \frac{\partial b_j}{\partial\vec{x}}\Big|_{\vec{x}}A\vec{x}, \label{eqn:bdot_pj} \\
    \frac{d}{dt}b_j(\vec{x}_m) &= \frac{\partial b_j}{\partial\vec{x}}\Big|_{\vec{x}_m}((A + BK)\vec{x}_m + B\vec{r}_s) = \frac{\partial b_j}{\partial\vec{x}}\Big|_{\vec{x}_m}A\vec{x}_m. \label{eqn:bdot_mj}
\end{align}
Equations \eqref{eqn:hdot_pi}- \eqref{eqn:bdot_mj} hold for all $\vec{x}, \vec{x}_m \in \bbR^n$.

Now, define the vector $\bar{\vec{x}} := [\vec{x}^\top, \vec{x}_m^\top]^\top$ and the set
\begin{equation} \label{eqn:C_1}
    \calC_1 := \{\bar{\vec{x}} \in \bbR^{2n} : \vec{x}, \vec{x}_m \in \calC_0 \land \bar{h}(\vec{x}_m) \geq \delta\}.
\end{equation}
Whenever $\bar{\vec{x}} \in \partial\calC_1$, one or more of the following is zero: $h_i(\vec{x})$ for some $i \in [1, r)$, $h_i(\vec{x}_m)$ for some $i \in [1, r)$, $\bar{h}(\vec{x})$, $\bar{h}(\vec{x}_m) - \delta$, $b_j(\vec{x})$ for some $j \in [1, N_\calC]$ with $r_j > 1$, or $b_j(\vec{x}_m)$ for some $j \in [1, N_\calC]$ with $r_j > 1$.

One may observe straightforwardly from \eqref{eqn:hdot_pi} that, whenever $\bar{\vec{x}} \in \calC_1$ and $h_i(\vec{x}) = 0$ for any $i < r$, we have $\frac{d}{dt}h_i(\vec{x}) \geq 0$. The same holds for $\vec{x}_m$, using \eqref{eqn:hdot_mi}. Similarly, whenever $\bar{\vec{x}} \in \calC_1$ and $b_j(\vec{x}) = 0$ for any $j$ with $r_j > 1$, from Property \ref{ppty:hbar}(c) and \eqref{eqn:bdot_pj}, we have $\frac{d}{dt}b_j(\vec{x}) \geq 0$, and the same holds for $\vec{x}_m$ using \eqref{eqn:bdot_mj}.

Next, let $\bar{\vec{x}} \in \calC_1$, and $\bar{h}(\vec{x}_m) - \delta = 0$. Then, $H_m = 0$, and it follows from \eqref{eqn:Hdot_m} that $\dot{H}_m \geq 0 \implies \frac{d}{dt}(\bar{h}(\vec{x}_m) - \delta) \geq 0$.

Finally, let $\bar{\vec{x}} \in \calC_1$ and $\bar{h}(\vec{x}) = 0$. Then, we have $H_e = 0$, and it follows from \eqref{eqn:Hdot_e} that $\dot{H}_e \geq 0$. Furthermore,
\begin{align*}
    \dot{H}_e &= \frac{d}{dt}(\bar{h}(\vec{x}_m)^2 - e_h^2) = \frac{d}{dt}\bar{h}(\vec{x})(2\bar{h}(\vec{x}_m) - \bar{h}(\vec{x})) \\
    &= (2\bar{h}(\vec{x}_m) - \bar{h}(\vec{x}))\frac{d}{dt}\bar{h}(\vec{x}) + \bar{h}(\vec{x})\frac{d}{dt}(\bar{h}(\vec{x}_m) - \bar{h}(\vec{x})) \\
    &= 2\bar{h}(\vec{x}_m)\frac{d}{dt}\bar{h}(\vec{x}) \geq 0.
\end{align*}
As $\bar{\vec{x}} \in \calC_1 \implies \bar{h}(\vec{x}_m) \geq \delta$, we have $\frac{d}{dt}\bar{h}(\vec{x}) \geq 0$.

Putting everything together, we have shown in the preceding paragraphs that $\calC_1$ is the intersection of the zero superlevel sets of some functions $\psi_1, \dots, \psi_N : \bbR^{2n} \to \bbR$ who all satisfy $\psi_k(\bar{\vec{x}}) = 0 \implies \frac{d}{dt}\psi_k(\bar{\vec{x}}) \geq 0$, and the initial conditions imply that $\bar{\vec{x}}(0) \in \calC_1$. Therefore, a straightforward application of Nagumo's theorem \cite{Nagumo_1942} shows that $\bar{\vec{x}}(t) \in \calC_1\ \forall t \geq 0$. Claim (b) follows immediately. \\

\noindent\textbf{Proof of Claim (c)}

In the proof of Claim (b), we showed that $\bar{\vec{x}}(t) \in \calC_1\ \forall t \geq 0$, and from \eqref{eqn:initial_condition_set} and \eqref{eqn:C_1}, we have $S_1 \subset \calC_1$. Claim (c) follows. \\

\noindent\textbf{Proof of Claim (d)}

From Theorem \ref{thm:lyapunov_no_input_uncertainty}, in order to prove Claim (d), we need only to prove that $\vec{r}_s(t)$ is bounded. First, we note that
we have proven that $\bar{\vec{x}}(t) \in \calC_1\ \forall t \geq 0 \implies \vec{x}_m(t), \vec{x}(t) \in \bar{\calC}\ \forall t \geq 0 \implies \vec{x}_m(t)$ and $\vec{x}(t)$ are bounded. This conclusion, together with \eqref{eqn:safe_reference_EBSB} and smoothness of $\bar{h}$, implies that $\vec{r}_s(t)$ is bounded if $\mu_m(\vec{x}(t), \vec{x}_m(t), t)$ and $\mu_p(\vec{x}(t), \vec{x}_m(t), t)$ are bounded.

From Claims (a) and (b), we know that $\mu_m(\vec{x}(t), \vec{x}_m(t), t)$ and $\mu_p(\vec{x}(t), \vec{x}_m(t), t)$ exist and have unique values for all $t \geq 0$. Furthermore, by strong convexity of the cost function in \eqref{eqn:J}, in order to show that $\mu_m(\vec{x}(t), \vec{x}_m(t), t)$ and $\mu_p(\vec{x}(t), \vec{x}_m(t), t)$ are bounded, it is sufficient to show that there exist bounded feasible solutions to \eqref{eqn:safe_lagrange_EBSB}. Feasible solutions have been found in the proof of Claim (a), and it follows straightforwardly from boundedness of $\vec{x}(t)$ and $\vec{x}_m(t)$ and smoothness of $\bar{h}$ that in each of the sub-cases of Claim (a), the feasible solution is bounded. Claim (d) follows. \\

\noindent\textbf{Proof of Claim (e)} 

Theorem \ref{thm:lyapunov_no_input_uncertainty} and Claims (b) and (d) imply that $\vec{x}(t)$, $\vec{x}_m(t)$, and $\vec{r}_s(t)$ are all bounded and that $\lim_{t \to \infty} \|\vec{e}_x(t)\| = 0$. Smoothness of $\bar{h}$ also implies that $\lim_{t \to \infty} |e_h(t)| = 0$. It follows that the following two finite times must exist:
\begin{enumerate}
    \item[(a)] $\exists T_1 \geq 0$ such that $|e_h(t)\rho_p(t) + \displaystyle\max_{\theta \in \Theta}\textstyle e_h(t)z_p(\vec{x}, t, \theta)| \leq \frac{\delta^2}{2}\ \forall t \geq T_1$, and
    \item[(b)] $\exists T_2 \geq 0$ such that $|e_h(t)| \leq \frac{\delta}{2}\ \forall t \geq T_2$.
\end{enumerate}
It follows from Claim (b) that $\bar{h}(\vec{x}(t)) \geq \frac{\delta}{2}\ \forall t \geq T_1$, and thus, defining $T := \max\{T_1, T_2\}$, from \eqref{eqn:Delta_ebsb}, we have $\Delta_{ebsb}(t) \leq \delta\ \forall t \geq T$.

Let $t \geq T$, and consider any $(\mu_m, \mu_p, \nu_m)$ satisfying \eqref{eqn:safe_lagrange_EBSB_const1}, and set $\nu_p = \nu_m$. Then, from \eqref{eqn:safe_lagrange_EBSB_const1_delta}-\eqref{eqn:safe_lagrange_EBSB_const2_Delta_ebsb} and $\Delta_{ebsb} \leq \delta$, we see that $(\mu_m, \mu_p, \nu_p)$ also satisfies \eqref{eqn:safe_lagrange_EBSB_const2}. Claim (e) follows. \\

\noindent\textbf{Proof of Claim (f)}

Proof of the safety objective is immediate from Claim (c). For the control objective, $\limsup_{t \to \infty} \|\vec{x}(t) - \vec{x}_*(t)\| \leq M$ for a constant $M > 0$ follows immediately from boundedness of $\vec{x}(t)$, and boundedness of $\vec{x}_*(t)$ due to boundedness of $\vec{r}_*(t)$ and Hurwitzness of $A_m$. Independence of $M$ from the parametric uncertainty follows from Claim (e) and the fact that \eqref{eqn:safe_reference_EBSB}-\eqref{eqn:safe_lagrange_EBSB_const1} are independent of the parametric uncertainty.
\hfill\qed

%% file: Appendix/Input_Uncertainty.tex
\section{}
\label{app:input_matrix}

\subsection{Proof of Theorem \ref{thm:lyapunov_input_uncertainty}}
\label{app:lyapunov_input_uncertainty}



Claim (a) is immediate from \eqref{eqn:vartheta_adaptive_law_1}-\eqref{eqn:vartheta_adaptive_law_2} and a simple application of Nagumo's theorem \cite{Nagumo_1942}. Then, we propose the following Lyapunov function candidate, which is standard in adaptive control \cite{Narendra05}:
\begin{equation} \label{eqn:lyapunov_function_input_uncertainty}
    V = \vec{e}_x^\top P\vec{e}_x + \frac{\|\tilde{\theta}\|^2}{\gamma_\theta} + \frac{\tilde{\vartheta}^\top\Lambda\tilde{\vartheta}}{\gamma_\vartheta}.
\end{equation}
Since $\theta_*$ and $\vartheta_*$ are assumed constant, we have $\dot{\tilde{\theta}} = \dot{\hat{\theta}}$ and $\dot{\tilde{\vartheta}} = \dot{\hat{\vartheta}}$. Additionally, note that $\diag(\tilde{\lambda})\vec{u} = \diag(\vec{u})\tilde{\lambda}$. Furthermore,
applying \eqref{eqn:vartheta_adaptive_law_2} and Assumption \ref{asn:lambda_in_set}, one may show that
\begin{equation} \label{eqn:vartheta_adaptive_law_property}
    \tilde{\vartheta}^\top\Lambda\dot{\hat{\vartheta}} = \sum_i \lambda_{*i}\tilde{\vartheta}_i\dot{\hat{\vartheta}}_i \leq \sum_i \lambda_{*i}\tilde{\vartheta}_i\bar{\vartheta}_i = \tilde{\vartheta}^\top\Lambda\bar{\vartheta}.
\end{equation}
Then, using \eqref{eqn:theta_adaptive_law}, \eqref{eqn:lyapunov_eqn}, \eqref{eqn:error_model_input_uncertainty}, \eqref{eqn:vartheta_adaptive_law_1}, and \eqref{eqn:vartheta_adaptive_law_property}, one can show through straightforward algebra that
\begin{equation} \label{eqn:lyap_decay_rate_EBSF}
    \dot{V} \leq -\vec{e}_x^\top Q\vec{e}_x \leq 0
\end{equation}
and thus that $V(t) \leq V(0)\ \forall t \geq 0$. Result (c) is immediate.

Now, if $\vec{r}_s(t)$ is bounded, then from \eqref{eqn:reference_model} and Hurwitzness of $A_m$, we know that $\vec{x}_m(t)$ is bounded. It follows from boundedness of $\vec{e}_x(t)$ that $\vec{x}(t)$ is bounded, thus giving Claim (d). Finally, \eqref{eqn:lyap_decay_rate_EBSF} shows that $\vec{e}_x(t) \in L^2$, and \eqref{eqn:error_model_input_uncertainty} with boundedness of $\vec{e}_x(t)$, $\tilde{\theta}(t)$, and $\vec{x}(t)$ shows that $\dot{\vec{e}}_x(t)$ is bounded. Claim (d) follows from Barbalat's lemma \cite{Narendra05}.
\hfill\qed

\subsection{Results of Assumption \ref{asn:hbar_CBF} for Problem \ref{prob:input_matrix}} \label{app:hbar_CBF_prob_2}

Here, before proving Theorem \ref{thm:safety_objective_input_uncertainty}, we provide useful lemmas that result from Assumption \ref{asn:hbar_CBF}.

\begin{lemma} \label{lem:hbar_CBF_for_plant_prob_2}
    Under Assumptions \ref{asn:theta_in_set}, \ref{asn:hbar_CBF}, and \ref{asn:lambda_in_set}, there exists a constant $d_7 > 0$ such that $(1 - \beta)z_m(\vec{x}, t) + \beta\min_{\theta \in \Xi, \lambda \in L} z_p(\vec{x}, t, \theta, \lambda) \geq \alpha_r\delta$ for every $t \geq 0$ and every $\vec{x} \in \bar{S} \cap \bar{\calC}$ where $\|B^\top\vec{g}(\vec{x})\| < d_7$.
\end{lemma}
\proof{
Follows identically to the proof of Lemma \ref{lem:hbar_CBF_for_reference_model}, but additionally using \eqref{eqn:z_p_ebsf}, the boundedness of $F(\vec{x})$ for bounded $\vec{x}$, boundedness of $\beta$ from \eqref{eqn:ebsf_interpolation}, boundedness of $\hat{\vartheta}(t)$ and $\hat{\theta}(t)$ from Theorem \ref{thm:lyapunov_input_uncertainty}(a)-(b), and boundedness of $\Theta$ and $L$ from Assumptions \ref{asn:theta_in_set} and \ref{asn:lambda_in_set}.
\hfill\qed
}

\subsection{Proof of Theorem \ref{thm:safety_objective_input_uncertainty}}
\label{app:safety_objective_input_uncertainty}

\noindent\textbf{Proof of Claim (a)}

First, applying Assumption \ref{asn:lambda_in_set}, Theorem \ref{thm:lyapunov_input_uncertainty}(a), and the fact that $\beta \in [0, 1]$, we have
\begin{equation} \label{eqn:mu_coefficient_lower_bound}
    \begin{aligned}
        &\min_{\lambda \in L} \vec{g}(\vec{x})^\top B((1 - \beta)I_m + \beta\diag(\lambda)\diag(\hat{\vartheta}))B^\top\vec{g}(\vec{x}) \\
        &\indenti{sp} \geq \|B^\top\vec{g}(\vec{x})\|^2\min_i(1 - \beta + \beta\ubar{\lambda}_i/\bar{\lambda}_i) \\
        &\indenti{sp} \geq \|B^\top\vec{g}(\vec{x})\|^2\min_i \ubar{\lambda}_i/\bar{\lambda}_i.
    \end{aligned}
\end{equation}
Then, in order to prove Claim (a), due to the strong convexity of the cost function, it is sufficient to show that a feasible solution exists whenever $\bar{h}(\vec{x}) \geq 0$ and $\vec{x} \in \bar{\calC}$. To show feasibility, we will consider two cases:
\begin{enumerate}
    \item $\|B^\top\vec{g}(\vec{x})\| < d_7$, and
    \item $\|B^\top\vec{g}(\vec{x})\| \geq d_7$,
\end{enumerate}
where $d_7 > 0$ is the constant from Lemma \ref{lem:hbar_CBF_for_plant_prob_2}. \\

\noindent\underline{Case 1) $\|B^\top\vec{g}(\vec{x})\| < d_7$}

In this case, as $\beta \in [0, 1]$, it follows immediately from Lemma \ref{lem:hbar_CBF_for_plant_prob_2} that $\mu = 0$ is a feasible solution to \eqref{eqn:safe_lagrange_EBSF}. \\

\noindent\underline{Case 2) $\|B^\top\vec{g}(\vec{x})\| \geq d_7$}

In this case, $\mu$ satisfying
\if \numsides 2
\begin{equation} \label{eqn:feasible_mu_case_2}
    \begin{aligned}
        &\mu\min_{\lambda \in L} \vec{g}(\vec{x})^\top B((1 - \beta)I_m + \beta\diag(\lambda)\diag(\hat{\vartheta}))B^\top\vec{g}(\vec{x}) \\
        &\indenti{sp} \geq \alpha_r\delta - (1 - \beta)z_m(\vec{x}, t) - \beta\min_{\theta \in \Xi, \lambda \in L} z_p(\vec{x}, t, \theta, \lambda)
    \end{aligned}
\end{equation}
\else
\begin{equation} \label{eqn:feasible_mu_case_2}
    \mu\min_{\lambda \in L} \vec{g}(\vec{x})^\top B((1 - \beta)I_m + \beta\diag(\lambda)\diag(\hat{\vartheta}))B^\top\vec{g}(\vec{x}) \geq \alpha_r\delta - (1 - \beta)z_m(\vec{x}, t) - \beta\min_{\theta \in \Xi, \lambda \in L} z_p(\vec{x}, t, \theta, \lambda)
\end{equation}
\fi
exists by \eqref{eqn:mu_coefficient_lower_bound} and is a feasible solution to \eqref{eqn:safe_lagrange_EBSF}. \\

\noindent\underline{Combining Cases 1 and 2}

As cases 1 and 2 cover all scenarios, Claim (a) is proven. \\

\noindent\textbf{Proof of Claim (b)}

Using \eqref{eqn:general_plant}, \eqref{eqn:input_input_uncertainty_1}-\eqref{eqn:input_input_uncertainty_2}, \eqref{eqn:safe_reference_EBSF}, \eqref{eqn:z_p_ebsf}, and Assumption \ref{asn:lambda_in_set}, we obtain
\if \numsides 2
\begin{equation} \label{eqn:h_x_derivative_1}
    \begin{aligned}
        \frac{d}{dt}\bar{h}(\vec{x}) &= \mu(\vec{x}, \vec{x}_m, t)\vec{g}(\vec{x})^\top B\diag(\lambda_*)\diag(\hat{\vartheta})B^\top\vec{g}(\vec{x}) \\
        &\indenti{=} + z_p(\vec{x}, t, \theta_*, \lambda_*) - \alpha_r\bar{h}(\vec{x}).
    \end{aligned}
\end{equation}
\else
\begin{equation} \label{eqn:h_x_derivative_1}
    \frac{d}{dt}\bar{h}(\vec{x}) = \mu(\vec{x}, \vec{x}_m, t)\vec{g}(\vec{x})^\top B\diag(\lambda_*)\diag(\hat{\vartheta})B^\top\vec{g}(\vec{x}) + z_p(\vec{x}, t, \theta_*, \lambda_*) - \alpha_r\bar{h}(\vec{x}).
\end{equation}
\fi
Whenever $\bar{h}(\vec{x}) \leq \frac{\delta}{3}$, we know from \eqref{eqn:ebsf_interpolation} that $\beta = 1$ and thus that \eqref{eqn:safe_lagrange_EBSF} reduces to
\if \numsides 2
\begin{equation}
    \begin{gathered} \label{eqn:safe_lagrange_EBSF_beta=1}
        \mu(\vec{x}, \vec{x}_m, t) = \argminbelow_{\mu \in \bbR} \mu^2 \\
        \mathrm{s.t.} \\
        \begin{gathered}
            \mu\min_{\lambda \in L} \vec{g}(\vec{x})^\top B\diag(\lambda)\diag(\hat{\vartheta})B^\top\vec{g}(\vec{x}) \\
            \begin{aligned}
                &\geq \alpha_r\delta - \min_{\theta \in \Xi, \lambda \in L} z_p(\vec{x}, t, \theta, \lambda).
            \end{aligned}
        \end{gathered}
    \end{gathered}
\end{equation}
\else
\begin{equation}
    \begin{gathered} \label{eqn:safe_lagrange_EBSF_beta=1}
        \mu(\vec{x}, \vec{x}_m, t) = \argminbelow_{\mu \in \bbR} \mu^2 \\
        \mathrm{s.t.} \\
        \mu\min_{\lambda \in L} \vec{g}(\vec{x})^\top B\diag(\lambda)\diag(\hat{\vartheta})B^\top\vec{g}(\vec{x}) \geq \alpha_r\delta - \min_{\theta \in \Xi, \lambda \in L} z_p(\vec{x}, t, \theta, \lambda).
    \end{gathered}
\end{equation}
\fi
One may verify by solving the KKT conditions of \eqref{eqn:safe_lagrange_EBSF_beta=1} that $\mu(\vec{x}, \vec{x}_m, t) \geq 0$.
Therefore, using \eqref{eqn:h_x_derivative_1}-\eqref{eqn:safe_lagrange_EBSF_beta=1} and Assumption \ref{asn:lambda_in_set}, whenever $\bar{h}(\vec{x}) \leq \frac{\delta}{3}$, we have
\begin{equation} \label{eqn:h_x_derivative_2}
    \frac{d}{dt}\bar{h}(\vec{x}) \geq -\alpha_r(\bar{h}(\vec{x}) - \delta) > -\alpha_r\bar{h}(\vec{x}).
\end{equation}
In particular, we have $\vec{x} \in \bar{S} \cap \bar{\calC}, \bar{h}(\vec{x}) = 0 \implies \frac{d}{dt}\bar{h}(\vec{x}) \geq 0$.

Now, consider the set $\calC_0$ in \eqref{eqn:initial_condition_set}.
Identically to the proof of Theorem \ref{thm:safety_objective_no_input_uncertainty}(b) in Appendix \ref{app:safety_objective_no_input_uncertainty}, for every function $h_i, i \in [1, r)$, we have $\vec{x} \in \calC_0, h_i(\vec{x}) = 0 \implies \frac{d}{dt}h_i(\vec{x}) \geq 0$, and for every function $b_j, j \in [1, N_\calC]$ in Property \ref{ppty:hbar}(c) with relative degree $r_j > 1$ w.r.t \eqref{eqn:general_plant}, we have $\vec{x} \in \calC_0, b_j(\vec{x}) = 0 \implies \frac{d}{dt}b_j(\vec{x}) \geq 0$. Additionally, the initial conditions state that $\vec{x}(0) \in \calC_0$. Therefore, applying the same reasoning as in the proof of Theorem \ref{thm:safety_objective_no_input_uncertainty}(b) in Appendix \ref{app:safety_objective_no_input_uncertainty}, Nagumo's theorem \cite{Nagumo_1942} shows that $\vec{x}(t) \in \calC_0\ \forall t \geq 0$. \\

\noindent\textbf{Proof of Claim (c)}

In the proof of Claim (b), we showed that $\vec{x}(t) \in \calC_0\ \forall t \geq 0$, and from \eqref{eqn:initial_condition_set}, we have $S_1 \subset \calC_0$. Claim (c) follows. \\

\noindent\textbf{Proof of Claim (d)}

From Theorem \ref{thm:lyapunov_input_uncertainty}, in order to prove claim (d), we need only to prove that $\vec{r}_s(t)$ is bounded. First, we have just shown that
$\vec{x}(t) \in \calC_0\ \forall t \geq 0 \implies \vec{x}(t) \in \bar{\calC}\ \forall t \geq 0$, and thus that $\vec{x}(t)$ is bounded. This conclusion, together with \eqref{eqn:safe_reference_EBSF} and smoothness of $\bar{h}$, implies that $\vec{r}_s(t)$ is bounded if $\mu(\vec{x}(t), \vec{x}_m(t), t)$ is bounded.

From Claims (a) and (b), we know that $\mu(\vec{x}(t), \vec{x}_m(t), t)$ exists and has a unique value for all $t \geq 0$. Furthermore, by strong convexity of the cost function in \eqref{eqn:safe_lagrange_EBSF}, in order to show that $\mu(\vec{x}(t), \vec{x}_m(t), t)$ is bounded, it is sufficient to show that there exist bounded feasible solutions to \eqref{eqn:safe_lagrange_EBSF}. Feasible solutions have already been found in the proof of Claim (a), and it is straightforward to show from boundedness of $\vec{x}(t)$ and $\vec{x}_m(t)$ and smoothness of $\bar{h}$ that in each of the cases of Claim (a), the feasible solution is bounded. Claim (d) follows. \\

\noindent\textbf{Proof of Claim (e)}

Define the gradient error signal $\vec{e}_g(t) := \vec{g}(\vec{x}(t)) - \vec{g}(\vec{x}_m(t))$.
Then, consider again the signal $H_m(t)$ defined in \eqref{eqn:H_m} and first used in the proof of Theorem \ref{thm:safety_objective_no_input_uncertainty}(b).
The minimizations over $\theta \in \Xi$ and $\lambda \in L$ in \eqref{eqn:safe_lagrange_EBSF} together with Theorem \ref{thm:lyapunov_input_uncertainty}(a) imply that the constraint in \eqref{eqn:safe_lagrange_EBSF} holds in particular for $\theta = \hat{\theta}$ and $\lambda = 1/\hat{\vartheta}$ element-wise. Thus, applying \eqref{eqn:reference_model}, \eqref{eqn:z_m_ebsb}, and \eqref{eqn:safe_reference_EBSF}-\eqref{eqn:z_p_ebsf}, we obtain
\if \numsides 2
\begin{align}
    \dot{H}_m &= \vec{g}(\vec{x}_m)^\top(A_m\vec{x}_m + B\vec{r}_s) \nonumber \\
    &= \vec{g}(\vec{x})^\top(A_m\vec{x} + B\vec{r}_s) - \vec{e}_g^\top(A_m\vec{x}_m + B\vec{r}_s) \nonumber \\
    &\indenti{=}\ - \vec{g}(\vec{x})^\top A_m\vec{e}_x \nonumber \\
    &= \vec{g}(\vec{x})^\top(A_m\vec{x} + B\vec{r}_*) + \|B^\top\vec{g}(\vec{x})\|^2\mu(\vec{x}, \vec{x}_m, t) \nonumber \\
    &\indenti{=}\ - \vec{e}_g^\top(A_m\vec{x}_m + B\vec{r}_s) - \vec{g}(\vec{x})^\top A_m\vec{e}_x \nonumber \\
    &\geq \vec{g}(\vec{x})^\top(A_m\vec{x} + B\vec{r}_*) + \alpha_r\delta - (1 - \beta)z_m(\vec{x}, t)\nonumber \\
    &\indenti{=}\ - \beta z_p(\vec{x}, t, \hat{\theta}, \hat{\lambda})  - \vec{e}_g^\top(A_m\vec{x}_m + B\vec{r}_s) - \vec{g}(\vec{x})^\top A_m\vec{e}_x \nonumber \\
    &= \alpha_r\delta - \alpha_r\bar{h}(\vec{x}) - \vec{e}_g^\top(A_m\vec{x}_m + B\vec{r}_s) - \vec{g}(\vec{x})^\top A_m\vec{e}_x \nonumber \\
    &= -\alpha_rH_m - \alpha_re_h - \vec{e}_g^\top(A_m\vec{x}_m + B\vec{r}_s) \nonumber \\
    &\indenti{=}\ - \vec{g}(\vec{x})^\top A_m\vec{e}_x.
\end{align}
\else
\begin{align}
    \dot{H}_m &= \vec{g}(\vec{x}_m)^\top(A_m\vec{x}_m + B\vec{r}_s) \nonumber \\
    &= \vec{g}(\vec{x})^\top(A_m\vec{x} + B\vec{r}_s) - \vec{e}_g^\top(A_m\vec{x}_m + B\vec{r}_s) - \vec{g}(\vec{x})^\top A_m\vec{e}_x \nonumber \\
    &= \vec{g}(\vec{x})^\top(A_m\vec{x} + B\vec{r}_*) + \|B^\top\vec{g}(\vec{x})\|^2\mu(\vec{x}, \vec{x}_m, t) - \vec{e}_g^\top(A_m\vec{x}_m + B\vec{r}_s) - \vec{g}(\vec{x})^\top A_m\vec{e}_x \nonumber \\
    &\geq \vec{g}(\vec{x})^\top(A_m\vec{x} + B\vec{r}_*) + \alpha_r\delta - (1 - \beta)z_m(\vec{x}, t) - \beta z_p(\vec{x}, t, \hat{\theta}, \hat{\lambda}) - \vec{e}_g^\top(A_m\vec{x}_m + B\vec{r}_s) - \vec{g}(\vec{x})^\top A_m\vec{e}_x \nonumber \\
    &= \alpha_r\delta - \alpha_r\bar{h}(\vec{x}) - \vec{e}_g^\top(A_m\vec{x}_m + B\vec{r}_s) - \vec{g}(\vec{x})^\top A_m\vec{e}_x \nonumber \\
    &= -\alpha_rH_m - \alpha_re_h - \vec{e}_g^\top(A_m\vec{x}_m + B\vec{r}_s) - \vec{g}(\vec{x})^\top A_m\vec{e}_x.
\end{align}
\fi
Finally, from Theorem \ref{thm:lyapunov_input_uncertainty} and Claims (b) and (d), we know that $\lim_{t \to \infty} \|\vec{e}_x(t)\| = 0$ and that $\vec{x}(t)$, $\vec{x}_m(t)$, and $\vec{r}_s(t)$ are all bounded. By smoothness of $\bar{h}$, we also know that $\lim_{t \to \infty} |e_h(t)| = \lim_{t \to \infty} \|\vec{e}_g(t)\| = 0$. Additionally, by the construction of $\omega(t)$, we have $\omega(t) \to 0$ as $\|\vec{e}_x(t)\| \to 0$. It follows that the following three finite times must exist:
\begin{enumerate}
    \item[(a)] $\exists T_1 \geq 0$ such that $|\alpha_re_h(t) + \vec{e}_g(t)^\top(A_m\vec{x}_m(t) + B\vec{r}_s(t)) + \vec{g}(\vec{x}(t))^\top A_m\vec{e}_x(t)| \leq \frac{\delta}{12}\ \forall t \geq T_1$,
    \item[(b)] $\exists T_2 \geq 0$ such that $|e_h(t)| \leq \frac{\delta}{6}\ \forall t \geq T_2$, and
    \item[(c)] $\exists T_3 \geq 0$ such that $\omega(t) \leq \frac{2\delta}{3}\ \forall t \geq T_3$.
\end{enumerate}
Then, for all $t \geq T_1$ such that $H_m(t) \leq -\frac{\delta}{6}$, we have $\dot{H}_m \geq \frac{\delta}{12}$, implying that there exists a finite time $T_4 \geq T_1$ such that $H_m(t) \geq -\frac{\delta}{6}\ \forall t \geq T_4 \implies \bar{h}(\vec{x}_m(t)) \geq \frac{5\delta}{6}\ \forall t \geq T_4 \implies \bar{h}(\vec{x}(t)) \geq \frac{2\delta}{3}\ \forall t \geq \max\{T_2, T_4\}$. It follows from the definition of $\beta$ in \eqref{eqn:ebsf_interpolation} that $\beta(\vec{x}(t), \omega(t)) = 0\ \forall t \geq \max\{T_2, T_3, T_4\} := T$. \\

\noindent\textbf{Proof of Claim (f)}

Proof of the safety objective is immediate from Claim (a). For the control objective, $\limsup_{t \to \infty} \|\vec{x}(t) - \vec{x}_*(t)\| \leq M$ for a constant $M > 0$ follows immediately from the boundedness of $\vec{x}(t)$ as previously shown, and the boundedness of  $\vec{x}_*(t)$ due to boundedness of $\vec{r}_*(t)$ and Hurwitzness of $A_m$. Furthermore, whenever $\beta = 0$, \eqref{eqn:safe_lagrange_EBSF} is independent of the parametric uncertainty. Independence of $M$ from the parametric uncertainty follows from Claim (e).
\hfill\qed

%% file: Appendix/SMID.tex
\section{} \label{app:SMID}


In \cite{lopez2021RaCBFs}, Set-Membership Identification (SMID) was proposed as a way to reduce conservatism of the RaCBF approach, allowing $\vec{x}_m$ to approach $\vec{x}_*$ more closely in response to excitation. As shown in Theorems \ref{thm:safety_objective_no_input_uncertainty} and \ref{thm:safety_objective_input_uncertainty}, EBSB and EBSF have conservatism which vanishes regardless of excitation. Nevertheless, in the event that excitation is present, we show in this section how to leverage SMID in EBSB and EBSF to further reduce conservatism in response to excitation. Without making assumptions on the amount of excitation available to learn from, the addition of SMID can neither strengthen nor weaken the theoretical performance guarantees. Thus, we prove here that all results of Theorems \ref{thm:lyapunov_no_input_uncertainty}-\ref{thm:safety_objective_input_uncertainty} hold under augmentation with SMID, and we explore the practical benefits of SMID in simulation in Section \ref{sec:simulations}. It should be noted that SMID cannot easily be extended to the setting where parameters are time-varying, and thus is only applicable when $\theta_*$ and $\lambda_*$ are constant for all time.

\subsection{Set Membership Identification (SMID)}

First, we will suppose that Assumptions \ref{asn:theta_in_set} and \ref{asn:lambda_in_set} are satisfied in the following manner:
\begin{assumptionp}{\ref{asn:theta_in_set}'} \label{asn:smid_theta_in_set}
    $\theta_* \in \Theta_0 := \{\theta \in \bbR^p : \theta^{(l)}_{i,0} \leq \theta_i \leq \theta^{(u)}_{i,0}\ \forall i\}$ for known constants $\theta^{(u)}_{i,0} > \theta^{(l)}_{i,0}$.
\end{assumptionp}
\begin{assumptionp}{\ref{asn:lambda_in_set}'} \label{asn:smid_lambda_in_set}
    $\Lambda = \diag(\lambda_*)$, where $\lambda_* \in L_0 := \{\lambda \in \bbR^m : \lambda^{(l)}_{i,0} \leq \lambda_i \leq \lambda^{(u)}_{i,0}\ \forall i\}$ for known constants $\lambda^{(u)}_{i,0} > \lambda^{(l)}_{i,0} > 0$.
\end{assumptionp}
\begin{remark}
    Clearly, Assumption \ref{asn:smid_lambda_in_set} is equivalent to Assumption \ref{asn:lambda_in_set} with $L = L_0$. Assumption \ref{asn:smid_theta_in_set} is also equivalent to Assumption \ref{asn:theta_in_set}, as Assumption \ref{asn:theta_in_set} implies that $\Theta$ is bounded, and thus may be enclosed inside a compact rectangular envelope. In what follows, we replace Assumptions \ref{asn:theta_in_set} and \ref{asn:lambda_in_set} with \ref{asn:smid_theta_in_set} and \ref{asn:smid_lambda_in_set} purely for ease of exposition.
\end{remark}
\noindent Then, for the purposes of this work, we define the following:
\begin{definition}[Element-Wise SMID Algorithm] \label{def:smid}
    A method of updating parameter sets $\Theta(t) \subset \bbR^p$ and $L(t) \subset \bbR^m$ is an \textit{element-wise SMID algorithm} if it satisfies the following for times $0 = t_0 < t_1 < \cdots < t_k\ \forall k \geq 0$:
    \begin{enumerate}
        \item[(a)] $\Theta(t) = \{\theta \in \bbR^p : \theta^{(l)}_i(t) \leq \theta_i \leq \theta^{(u)}_i(t)\ \forall i\}\ \forall t \geq 0$,
        \item[(b)] $L(t) := \{\lambda \in \bbR^m : \lambda^{(l)}_i(t) \leq \lambda_i \leq \lambda^{(u)}_i(t)\ \forall i\}\ \forall t \geq 0$,
        \item[(c)] $\Theta(0) = \Theta_0$ and $L(0) = L_0$,
        \item[(d)] $\Theta(t) = \Theta(t_k)$ and $L(t) = L(t_k)\ \forall t \in [t_k, t_{k+1})\ \forall k \geq 0$,
        \item[(e)] $\Theta(t_{k+1}) \subseteq \Theta(t_k)$ and $L(t_{k+1}) \subseteq L(t_k)\ \forall k \geq 0$, and
        \item[(f)] $\theta_* \in \Theta(t_k)$ and $\lambda_* \in L(t_k)\ \forall k \geq 0$.
    \end{enumerate}
\end{definition}

\subsection{Adding SMID to EBSB} \label{subsec:ebsb_smid}

Consider the following modifications to the control design in Section \ref{sec:no_input_uncertainty}. First, replace the adaptive law in \eqref{eqn:theta_adaptive_law} with
\begin{subequations}
    \begin{align}
        \bar{\theta} &= -\gamma_\theta F(\vec{x})^\top B^\top P\vec{e}_x, \label{eqn:theta_adaptive_law_3} \\
        \dot{\hat{\theta}}_i &= \begin{cases} 0, & (\hat{\theta}_i \geq \theta^{(u)}_i(t) \land \bar{\theta}_i \geq 0) \lor (\hat{\theta}_i \leq \theta^{(l)}_i(t) \land \theta'_i \leq 0) \\ \bar{\theta}_i, & \mathrm{otherwise} \end{cases}\ \forall i \label{eqn:theta_adaptive_law_4}
    \end{align}
\end{subequations}
with an additional modification at each SMID update time $t_k$ given by
\begin{equation} \label{eqn:theta_adaptive_law_smid}
    \hat{\theta}_i(t_k) \gets \min\left\{\max\left\{\hat{\theta}_i(t_k), \theta^{(l)}_i(t_k)\right\}, \theta^{(u)}_i(t_k)\right\}\ \forall i,
\end{equation}
where $\theta^{(u)}_i(t)$ and $\theta^{(l)}_i(t)$ are the element-wise bounds in $\Theta(t)$ given by an element-wise SMID algorithm.
Second, simply replace $\Theta$ in \eqref{eqn:safe_lagrange_EBSB_const2} with $\Theta(t)$.
Then, given these two modifications, it is straightforward to show that the main results of Section \ref{sec:no_input_uncertainty} hold as follows:
\begin{theorem} \label{thm:ebsb_smid}
    For $\calC_0$ defined in \eqref{eqn:initial_condition_set}, let $\vec{x}_m(0), \vec{x}(0) \in \calC_0$, $\bar{h}(\vec{x}_m(0)) \geq \delta$, and $\hat{\theta}(0) \in \Theta_0$. Then, for the closed-loop adaptive system consisting of \eqref{eqn:plant_no_input_uncertainty}-\eqref{eqn:z_p_ebsb} under Assumptions \ref{asn:smid_theta_in_set} and \ref{asn:hbar_CBF}-\ref{asn:gradients_unopposed_at_boundary} with the modifications above and $\Theta(t)$ given by an element-wise SMID algorithm, all results in Theorems \ref{thm:lyapunov_no_input_uncertainty} and \ref{thm:safety_objective_no_input_uncertainty} hold, and we additionally have $\hat{\theta}(t) \in \Theta(t)\ \forall t \geq 0$.
\end{theorem}
\begin{proof}
    The extensions of the proofs of Theorems \ref{thm:lyapunov_no_input_uncertainty}- \ref{thm:safety_objective_no_input_uncertainty} are straightforward. The claim that $\hat{\theta}(t) \in \Theta(t)\ \forall t \geq 0$ follows immediately from \eqref{eqn:theta_adaptive_law_3}-\eqref{eqn:theta_adaptive_law_4} and Nagumo's theorem \cite{Nagumo_1942} for all $t \notin \{t_1, t_2, \dots\}$, and from \eqref{eqn:theta_adaptive_law_smid} for $t \in \{t_1, t_2, \dots\}$. Then, for Theorem \ref{thm:lyapunov_no_input_uncertainty}, as $\theta_* \in \Theta(t) \ \forall t \geq 0$ by Definition \ref{def:smid}, the relation $\tilde{\theta}^\top\dot{\hat{\theta}} \leq \tilde{\theta}^\top\bar{\theta}$ holds similarly to the reasoning in \eqref{eqn:vartheta_adaptive_law_property}. Thus, for $V$ in \eqref{eqn:lyapunov_function}, one may show that $\dot{V} \leq -\vec{e}_x^\top Q\vec{e}_x$ for all $t \notin \{t_1, t_2, \dots\}$, and it is straightforward to verify that \eqref{eqn:theta_adaptive_law_smid} along with $\Theta_{k+1} \subseteq \Theta_k\ \forall k$ ensures that $V(t_k^+) \leq V(t_k^-)\ \forall k$. Therefore, the results of Theorem \ref{thm:lyapunov_no_input_uncertainty} hold.
    Finally, the fact that $\hat{\theta}(t), \theta_* \in \Theta(t)\ \forall t \geq 0$ means that the proof of Theorem \ref{thm:safety_objective_no_input_uncertainty} remains unchanged.
\end{proof}

\subsection{Adding SMID to EBSF} \label{subsec:ebcg_smid}

Consider the following modifications to the control design in Section \ref{sec:input_uncertainty}. First, replace the adaptive law in \eqref{eqn:theta_adaptive_law} with \eqref{eqn:theta_adaptive_law_3}-\eqref{eqn:theta_adaptive_law_smid}, and replace the adaptive law in \eqref{eqn:vartheta_adaptive_law_1}-\eqref{eqn:vartheta_adaptive_law_2} with
\begin{subequations}
    \begin{align}
        \bar{\vartheta} &= -\gamma_\vartheta\diag(\phi)B^\top P\vec{e}_x, \label{eqn:vartheta_adaptive_law_3} \\
        \dot{\hat{\vartheta}}_i &= \begin{cases} 0, & (\hat{\vartheta}_i \geq \frac{1}{\lambda^{(l)}_i(t)} \land \bar{\vartheta}_i \geq 0) \lor (\hat{\vartheta}_i \leq \frac{1}{\lambda^{(u)}}_i(t) \land \bar{\vartheta}_i \leq 0) \\ \bar{\vartheta}_i, & \mathrm{otherwise} \end{cases}\ \forall i \label{eqn:vartheta_adaptive_law_4}
    \end{align}
\end{subequations}
with an additional modification at each SMID update time $t_k$ given by
\begin{equation} \label{eqn:vartheta_adaptive_law_smid}
    \hat{\vartheta}_i(t_k) \gets \min\left\{\max\left\{\hat{\vartheta}_i(t_k), \frac{1}{\lambda^{(u)}_i(t_k)}\right\}, \frac{1}{\lambda^{(l)}_i(t_k)}\right\}\ \forall i,
\end{equation}
where $\lambda^{(u)}_i(t)$ and $\lambda^{(l)}_i(t)$ are the element-wise bounds in $L(t)$ given by an element-wise SMID algorithm.
Second, simply replace $\Xi$ and $L$ in \eqref{eqn:safe_lagrange_EBSF_const} with $\Theta(t)$ and $L(t)$ respectively.
Finally, define $L^{-1}(t) := \{\vartheta \in \bbR^m : \frac{1}{\lambda^{(u)}_i(t)} \leq \vartheta_i \leq \frac{1}{\lambda^{(l)}_i(t)}\ \forall i\}$ with $L^{-1}_0 := L^{-1}(0)$, and it is straightforward to see that Definition \ref{def:smid}(f) implies that $\vartheta_* \in L^{-1}(t_k)\ \forall k \geq 0$.
Then, given these modifications, it is straightforward to verify that the main results of Section \ref{sec:input_uncertainty} hold as follows:
\begin{theorem} \label{thm:ebsf_smid}
    For $\calC_0$ defined in \eqref{eqn:initial_condition_set}, let $\vec{x}(0) \in \calC_0$, $\hat{\theta}(0) \in \Theta_0$, and $\hat{\vartheta}(0) \in L^{-1}_0$. Then, for the closed-loop adaptive system consisting of \eqref{eqn:general_plant}, \eqref{eqn:reference_model}, \eqref{eqn:hocbf_base_case}-\eqref{eqn:hocbf_recursion}, \eqref{eqn:z_m_ebsb}, and \eqref{eqn:input_input_uncertainty_1}-\eqref{eqn:sat}, under Assumptions \ref{asn:smid_theta_in_set}, \ref{asn:hbar_CBF}, and \ref{asn:smid_lambda_in_set} with the modifications above and $\Theta(t)$, $L(t)$ given by an element-wise SMID algorithm, all results in Theorems \ref{thm:lyapunov_input_uncertainty} and \ref{thm:safety_objective_input_uncertainty} hold, and we additionally have $\hat{\theta}(t) \in \Theta(t), \hat{\vartheta}(t) \in L^{-1}(t)\ \forall t \geq 0$.
\end{theorem}
\begin{proof}
    The extensions of the proofs of Theorems \ref{thm:lyapunov_input_uncertainty} and \ref{thm:safety_objective_input_uncertainty} are straightforward. First, as in the proof of Theorem \ref{thm:ebsb_smid}, the claim that $\hat{\theta}(t) \in \Theta(t), \hat{\vartheta}(t) \in L^{-1}(t)\ \forall t \geq 0$ follows from \eqref{eqn:theta_adaptive_law_3}-\eqref{eqn:vartheta_adaptive_law_smid} and Nagumo's Theorem \cite{Nagumo_1942}. Then, we have $\tilde{\theta}^\top\dot{\hat{\theta}} \leq \tilde{\theta}^\top\bar{\theta}$, and \eqref{eqn:vartheta_adaptive_law_property} holds. For Theorem \ref{thm:lyapunov_input_uncertainty}, for $V$ in \eqref{eqn:lyapunov_function_input_uncertainty}, one may show that $\dot{V} \leq -\vec{e}_x^\top Q\vec{e}_x$ for all $t \notin \{t_1, t_2, \dots\}$, and it is straightforward to verify that \eqref{eqn:theta_adaptive_law_smid} and \eqref{eqn:vartheta_adaptive_law_smid} along with $\Theta_{k+1} \subseteq \Theta_k$ and $L_{k+1} \subseteq L_k\ \forall k$ ensure that $V(t_k^+) \leq V(t_k^-)\ \forall k$. Therefore, the results of Theorem \ref{thm:lyapunov_input_uncertainty} hold.
    Finally, the facts that $\hat{\theta}(t), \theta_* \in \Theta(t)$, $\hat{\vartheta}(t) \in L^{-1}(t)$, and $\lambda_* \in L(t)\ \forall t \geq 0$ means that the proof of Theorem \ref{thm:safety_objective_input_uncertainty} remains unchanged.
\end{proof}

%% file: Appendix/Simulation_Details.tex
\section{} \label{app:simulation}


\subsection{Ideal Control Design if All Parameters Are Known} \label{app:simulation_nominal_control}

For the simulation setting in Section \ref{sec:simulations}, if all parameters were known, one could satisfy the control and safety objectives as follows. Choose $\vec{u}$ as
\begin{equation}
    \vec{u} = \frac{1}{\lambda_*}(K\vec{x} + \vec{r}_s + F(\vec{x})\theta_*)
\end{equation}
where $K$ is chosen by LQR on the dynamics $(A, B)$. Define $A_m = A + BK$. Then, choose the nominal tracking reference input as
\begin{equation}
    \vec{r}_* = T(s)^{-1}\vec{p}_*, \quad T(s) = \begin{bmatrix} I_2 & 0_{2 \times 2} \end{bmatrix}(sI_4 - A_m)^{-1}B,
\end{equation}
where $T(s)$ is the transfer function from $\vec{r}_s$ to $\vec{p}$. Finally,
construct the augmented CBF $\bar{h}(\vec{x})$ as in Section \ref{subsec:simulations_composite_cbf},
and choose $\vec{r}_s$ according to
\begin{equation}
    \vec{r}_s = \vec{r}_* + \mu B^\top\vec{g}(\vec{x}), \label{eqn:safe_reference_ideal}
\end{equation}
\begin{subequations} \label{eqn:safe_lagrange_ideal}
    \begin{gather}
        \mu(\vec{x}, t) = \argminbelow_{\mu \in \bbR} \mu^2 \label{eqn:safe_lagrange_ideal_cost} \\
        \mathrm{s.t.} \nonumber \\
        \mu\|B^\top\vec{g}(\vec{x})\|^2 \geq \alpha_r\delta - z_m(\vec{x}, t) \label{eqn:safe_lagrange_ideal_const}
    \end{gather}
\end{subequations}
with $z_m$ given by \eqref{eqn:z_m_ebsb}.

\subsection{aCBF Implementation} \label{app:aCBF}

In our simulations, for simplicity and for ease of comparison, we integrate adaptive control with the aCBF approach in a way which mirrors the structure of our EBSB and EBSF control designs, differing from the aCLF-aCBF-QP approach proposed in \cite{taylor2020aCBFs}. In this section, we consider Problem \ref{prob:input_matrix} with the plant given in \eqref{eqn:general_plant}, and treat Problem \ref{prob:unforced_dynamics} as merely a special case. Thus, we also extend the approach in \cite{taylor2020aCBFs} to handle uncertainties in the input matrix. We first require Assumptions \ref{asn:smid_theta_in_set}, \ref{asn:hbar_CBF}, and \ref{asn:smid_lambda_in_set}, and design a standard model-reference adaptive controller using the reference model in \eqref{eqn:reference_model} and the control design in \eqref{eqn:input_input_uncertainty_1}-\eqref{eqn:vartheta_adaptive_law_2}. Then, substituting \eqref{eqn:input_input_uncertainty_1}-\eqref{eqn:input_input_uncertainty_2} into \eqref{eqn:general_plant}, we obtain
\begin{equation} \label{eqn:aCBF_cl_plant}
    \dot{\vec{x}} = A\vec{x} + B(\diag(\lambda_*)\diag(\hat{\vartheta})(K\vec{x} + \vec{r}_s + F(\vec{x})\hat{\theta}) - F(\vec{x})\theta_*)
\end{equation}

Now, to ensure safety of the plant, we construct $\bar{h}(\vec{x})$ as in Section \ref{subsec:simulations_composite_cbf}, and following \cite{taylor2020aCBFs}, we design the adaptive control barrier function as
\begin{equation} \label{eqn:aCBF}
    h_a(\vec{x}) = \begin{cases} \Delta_{acbf}^2, & \bar{h}(\vec{x}) \geq \Delta_{acbf}, \\ \Delta_{acbf}^2 - (\bar{h}(\vec{x}) - \Delta_{acbf})^2, & \bar{h}(\vec{x}) < \Delta_{acbf} \end{cases}
\end{equation}
for a constant $\Delta_{acbf} > 0$ to be specified later. Then, we choose $\vec{r}_s$ according to
\if \numsides 2
\begin{equation} \label{eqn:aCBF_governor}
    \begin{gathered}
        \vec{r}_s = \argminbelow_{\vec{r} \in \bbR^m} \|\vec{r} - \vec{r}_*\|^2\ \mathrm{s.t.} \\
        \begin{aligned}
            &\frac{\partial h_a}{\partial\vec{x}}\Big|_{\vec{x}}\Big(A_m\vec{x} \\
            &\indenti{sp} + B\left(\vec{r} + F(\vec{x})(\hat{\theta} - \hat{\theta}_s) - \diag(\vec{u})(\hat{\lambda} - \hat{\lambda}_s)\right)\Big) \geq 0
        \end{aligned}
    \end{gathered}
\end{equation}
\else
\begin{subequations} \label{eqn:aCBF_governor}
    \begin{gather}
        \vec{r}_s = \argminbelow_{\vec{r} \in \bbR^m} \|\vec{r} - \vec{r}_*\|^2 \label{eqn:aCBF_governor_cost} \\
        \mathrm{s.t.} \nonumber \\
        \frac{\partial h_a}{\partial\vec{x}}\Big|_{\vec{x}}\left(A\vec{x} + B\left(\diag(\hat{\lambda}_s)\diag(\hat{\vartheta})(K\vec{x} + \vec{r} + F(\vec{x})\hat{\theta}) - F(\vec{x})\hat{\theta}_s\right)\right) \geq 0 \label{eqn:aCBF_governor_const}
    \end{gather}
\end{subequations}
\fi
which is a quadratic programming problem in $\vec{r}$. $\hat{\theta}_s$ and $\hat{\lambda}_s$ are auxiliary parameter estimates with adaptive laws given by
\begin{align}
    \bar{\theta}_s &= \gamma_{\theta,s}\left(\frac{\partial h_a}{\partial\vec{x}}\Big|_{\vec{x}}BF(\vec{x})\right)^\top, \label{eqn:aCBF_theta_adaptive_law_1} \\
    \dot{\hat{\theta}}_{s,i} &= \begin{cases} 0, & (\hat{\theta}_{s,i} \geq \theta^{(u)}_{i,0} \land \bar{\theta}_{s,i} \geq 0) \lor (\hat{\theta}_{s,i} \leq \theta^{(l)}_{i,0} \land \bar{\theta}_{s,i} \leq 0) \\ \bar{\theta}_{s,i}, & \mathrm{otherwise} \end{cases}\ \forall i \label{eqn:aCBF_theta_adaptive_law_2} \\
    \bar{\lambda}_s &= -\gamma_{\lambda,s}\left(\frac{\partial h_a}{\partial\vec{x}}\Big|_{\vec{x}}B\diag(\hat{\vartheta})\diag(K\vec{x} + \vec{r}_s + F(\vec{x})\hat{\theta})\right)^\top, \label{eqn:aCBF_lambda_adaptive_law_1} \\
    \dot{\hat{\lambda}}_{s,i} &= \begin{cases} 0, & (\hat{\lambda}_{s,i} \geq \lambda^{(u)}_{i,0} \land \bar{\lambda}_{s,i} \geq 0) \lor (\hat{\lambda}_{s,i} \leq \lambda^{(l)}_{i,0} \land \bar{\lambda}_{s,i} \leq 0) \\ \bar{\lambda}_{s,i}, & \mathrm{otherwise} \end{cases}\ \forall i \label{eqn:aCBF_lambda_adaptive_law_2}
\end{align}
for any adaptive gains $\gamma_{\theta,s}, \gamma_{\lambda,s} > 0$. Finally, as in \cite{taylor2020aCBFs}, defining $\tilde{\theta}_s := \hat{\theta}_s - \theta_*$ and $\tilde{\lambda}_s := \hat{\lambda}_s - \lambda_*$, we require $\Delta_{acbf} \geq \sqrt{\frac{1}{2\gamma_{\theta,s}}\|\tilde{\theta}_s(0)\|^2 + \frac{1}{2\gamma_{\lambda,s}}\|\tilde{\lambda}_s(0)\|^2}$. The proof of safety follows the proof of Theorem 3 in \cite{taylor2020aCBFs}.

\subsection{RaCBF Implementation} \label{app:RaCBF}

In our simulations, for simplicity and for ease of comparison, we integrate adaptive control with the RaCBF approach in a way which mirrors the structure of our EBSB and EBSF control designs, differing from the approach proposed in \cite{lopez2021RaCBFs}. In this section, we consider Problem \ref{prob:input_matrix} with the plant given in \eqref{eqn:general_plant}, and treat Problem \ref{prob:unforced_dynamics} as merely a special case. Thus, we also extend the approach in \cite{lopez2021RaCBFs} to handle uncertainties in the input matrix. We first require Assumptions \ref{asn:smid_theta_in_set}, \ref{asn:hbar_CBF}, and \ref{asn:smid_lambda_in_set}, and design a standard model-reference adaptive controller using the reference model in \eqref{eqn:reference_model} and the control design in \eqref{eqn:input_input_uncertainty_1}-\eqref{eqn:vartheta_adaptive_law_2}. Then, substituting \eqref{eqn:input_input_uncertainty_1}-\eqref{eqn:input_input_uncertainty_2} into \eqref{eqn:general_plant}, we obtain the closed-loop plant in \eqref{eqn:aCBF_cl_plant}.

Now, to ensure safety of the plant, we construct $\bar{h}(\vec{x})$ as in Section \ref{subsec:simulations_composite_cbf}. Then, we choose $\vec{r}_s$ according to
\if \numsides 2
\begin{subequations} \label{eqn:RaCBF_governor}
    \begin{gather}
        \vec{r}_s = \argminbelow_{\vec{r} \in \bbR^m} \|\vec{r} - \vec{r}_*\|^2 \label{eqn:RaCBF_governor_cost} \\
        \mathrm{s.t.} \nonumber \\
        \frac{\partial h_r}{\partial\vec{x}}\Big|_{\vec{x}}\left(A_m\vec{x} + B\left(\vec{r} + F(\vec{x})(\hat{\theta} - \hat{\theta}_s) - \diag(\vec{u})(\hat{\lambda} - \hat{\lambda}_s)\right)\right) \\
        \geq -\alpha_rh_r(\vec{x}) + \Delta_{racbf}, \label{eqn:RaCBF_governor_const}
    \end{gather}
\end{subequations}
\else
\begin{subequations} \label{eqn:RaCBF_governor}
    \begin{gather}
        \vec{r}_s = \argminbelow_{\vec{r} \in \bbR^m} \|\vec{r} - \vec{r}_*\|^2 \label{eqn:RaCBF_governor_cost} \\
        \mathrm{s.t.} \nonumber \\
        \frac{\partial\bar{h}}{\partial\vec{x}}\Big|_{\vec{x}}\left(A\vec{x} + B\left(\diag(\hat{\lambda}_s)\diag(\hat{\vartheta})(K\vec{x} + \vec{r} + F(\vec{x})\hat{\theta}) - F(\vec{x})\hat{\theta}_s\right)\right)) \geq \alpha_r(\Delta_{racbf} - \bar{h}(\vec{x})), \label{eqn:RaCBF_governor_const}
    \end{gather}
\end{subequations}
\fi
which is a quadratic programming problem in $\vec{r}$, for a constant $\Delta_{racbf} > 0$ to be specified later. $\hat{\theta}_s$ and $\hat{\lambda}_s$ are auxiliary parameter estimates with adaptive laws given by \eqref{eqn:aCBF_theta_adaptive_law_1}-\eqref{eqn:aCBF_lambda_adaptive_law_2}, with $h_a$ replaced by $\bar{h}$, for any adaptive gains $\gamma_{\theta,s}, \gamma_{\lambda,s} > 0$. Finally, as in \cite{lopez2021RaCBFs}, defining $\tilde{\theta}_s := \hat{\theta}_s - \theta_*$ and $\tilde{\lambda}_s := \hat{\lambda}_s - \lambda_*$, we require $\Delta_{racbf} \geq \sup_{t \geq 0} (\frac{1}{2\gamma_{\theta,s}}\|\tilde{\theta}_s(t)\|^2 + \frac{1}{2\gamma_{\lambda,s}}\|\tilde{\lambda}_s(t)\|^2)$. The proof of safety follows the proof of Theorem 2 in \cite{lopez2021RaCBFs}.

\subsection{Set Membership Identification Implementation} \label{app:SMID_implementation}

Our implementation of Set Membership Identification maintains high-probability upper and lower bounds on each uncertain parameter, and updating the bounds in response to input-output data. We will describe our SMID approach for Problem \ref{prob:input_matrix} here, and the approach for Problem \ref{prob:unforced_dynamics} is simply a special case. To begin, we Euler-discretize the dynamics in \eqref{eqn:general_plant} with a time step $\Delta t > 0$, resulting in the discrete-time dynamics
\begin{equation} \label{eqn:discretized_dynamics}
    \vec{x}_{k+1} = (I_n + A\Delta t)\vec{x}_k + B\Delta t(\diag(\lambda_*)\vec{u}_k - F(\vec{x}_k)\theta_*) + \vec{w}_{k+1}
\end{equation}
where $\vec{w}_{k+1}$ represents the error in the discretization. Define the quantities $\vec{y}_{k+1} = (B\Delta t)^\dagger(\vec{x}_{k+1} - (I_n + A\Delta t)\vec{x}_k)$, $\Phi_k = [-F(\vec{x}_k), \diag(\vec{u}_k)]$, $\xi_* = [\theta_*^\top, \lambda_*^\top]^\top$, and $\eta_{k+1} = (B\Delta t)^\dagger\vec{w}_{k+1}$, where $\dagger$ represents the Moore-Penrose pseudo-inverse. We then make the following modeling assumptions:
\begin{assumption} \label{asn:B_pinv}
    $B$ has linearly independent columns, so that $B^\dagger = (B^\top B)^{-1}B^\top$.
\end{assumption}
\begin{assumption}
    At each time step $k \geq 0$, the scaled discretization error $\eta_{k+1}$ is drawn from a zero-mean Gaussian distribution with symmetric positive-definite covariance $\Sigma_{k+1}$.
\end{assumption}
\noindent Then, \eqref{eqn:discretized_dynamics} is equivalent to
\begin{equation} \label{eqn:linear_regression}
    \vec{y}_{k+1} = \Phi_k\zeta_* + \eta_{k+1}
\end{equation}
with $\Phi_k \in \bbR^{m \times (p+m)}$ and $\vec{y}_{k+1} \in \bbR^m$ measured, $\zeta_* \in \bbR^{p+m}$ the unknown parameter vector, and $\eta_{k+1} \sim \calN(0, \Sigma_{k+1})\ \forall k \geq 0$.

One can show using Bayes' rule that, if the probability distribution of $\zeta_*$ at time step $k$ is characterized by the mean and covariance $\hat{\zeta}_k$ and $P_k$, then the updated mean and covariance in response to the data $(\Phi_k, \vec{y}_{k+1})$ are given by
\begin{subequations}
    \begin{align}
        \hat{\zeta}_{k+1} &= \hat{\zeta}_k + P_k\Phi_k^\top(\Sigma_{k+1} + \Phi_kP_k\Phi_k^\top)^{-1}(\vec{y}_{k+1} - \Phi_k\hat{\zeta}_k), \label{eqn:RLS_mean_update} \\
        P_{k+1} &= P_k - P_k\Phi_k^\top(\Sigma_{k+1} + \Phi_kP_k\Phi_k^\top)^{-1}\Phi_kP_k. \label{eqn:RLS_covariance_update}
    \end{align}
\end{subequations}
Now, given a mean and covariance for the probability distribution of $\zeta_*$ at each time step, we can calculate high-probability upper and lower bounds for each element of $\zeta_*$ as follows.

At time step $k$, we have $\zeta_* \sim \calN(\hat{\zeta}_k, P_k)$. Then, given any $\epsilon > 0$, we want to find $\ubar{\zeta}_{ki}$ and $\bar{\zeta}_{ki}$ such that $\bbP[\zeta_{*i} \in [\ubar{\zeta}_{ki}, \bar{\zeta}_{ki}]\ \forall i \in [1, p+m]] \geq 1 - \epsilon$, where $\bbP[X]$ denotes the probability of event $X$. First, we diagonalize the covariance matrix using SVD to obtain
\begin{equation} \label{eqn:cov_SVD}
    P_k = V_kD_kV_k^\top,
\end{equation}
where $V_k$ is orthogonal and $D_k = \diag(\vec{d}_k)$. Then, define the vectors
\begin{equation} \label{eqn:rotated_zeta}
    \vec{v}_* = V_k^\top\zeta_*, \quad \hat{\vec{v}}_k = V_k^\top\hat{\zeta}_k.
\end{equation}
It is straightforward to show that $\vec{v}_* \sim \calN(\hat{\vec{v}}_k, D_k)$, or in other words, that $v_{*i} \sim \calN(\hat{v}_{ki}, d_{ki})$ for each $i \in [1, p+m]$. Given that the elements $v_{*i}$ are independent for all $i$, we can straightforwardly produce high-probability bounds using their cumulative distribution functions: defining
\begin{equation} \label{eqn:high_probability_offsets}
    \tilde{v}_{ki} = \sqrt{2d_{ki}}\erf^{-1}\left(1 - \frac{\epsilon}{p+m}\right)
\end{equation}
where $\erf^{-1}$ denotes the inverse error function, it is easy to show that $\bbP[v_{*i} \in [\hat{v}_{ki} - \tilde{v}_{ki}, \hat{v}_{ki} + \tilde{v}_{ki}]] \geq 1 - \frac{\epsilon}{p+m}$ for each $i$. Finally, denoting the $(i,j)$th element of $V_k$ as $v_{kij}$, for each $j \in [1, p+m]$, with probability at least $1 - \frac{\epsilon}{p+m}$, we have
\if \numsides 2
\begin{subequations}
    \begin{align}
        \zeta_{*j} &= \sum_{j=1}^{p+m} v_{kij}v_{*j} \geq \sum_{j=1}^{p+m} v_{kij}\begin{Bmatrix} \hat{v}_{kj} - \tilde{v}_{kj}, & v_{kij} \geq 0 \\ \hat{v}_{kj} + \tilde{v}_{kj}, & v_{kij} < 0 \end{Bmatrix} \nonumber \\
        &= \sum_{j=1}^{p+m} (v_{kij}\hat{v}_{kj} - |v_{kij}|\tilde{v}_{kj}), \label{eqn:zeta_lower_bound} \\
        \zeta_{*j} &= \sum_{j=1}^{p+m} v_{kij}v_{*j} \leq \sum_{j=1}^{p+m} v_{kij}\begin{Bmatrix} \hat{v}_{kj} + \tilde{v}_{kj}, & v_{kij} \geq 0 \\ \hat{v}_{kj} - \tilde{v}_{kj}, & v_{kij} < 0 \end{Bmatrix} \nonumber \\
        &= \sum_{j=1}^{p+m} (v_{kij}\hat{v}_{kj} + |v_{kij}|\tilde{v}_{kj}). \label{eqn:zeta_upper_bound}
    \end{align}
\end{subequations}
\else
\begin{subequations}
    \begin{align}
        \zeta_{*j} &= \sum_{j=1}^{p+m} v_{kij}v_{*j} \geq \sum_{j=1}^{p+m} v_{kij}\begin{Bmatrix} \hat{v}_{kj} - \tilde{v}_{kj}, & v_{kij} \geq 0 \\ \hat{v}_{kj} + \tilde{v}_{kj}, & v_{kij} < 0 \end{Bmatrix} = \sum_{j=1}^{p+m} (v_{kij}\hat{v}_{kj} - |v_{kij}|\tilde{v}_{kj}), \label{eqn:zeta_lower_bound} \\
        \zeta_{*j} &= \sum_{j=1}^{p+m} v_{kij}v_{*j} \leq \sum_{j=1}^{p+m} v_{kij}\begin{Bmatrix} \hat{v}_{kj} + \tilde{v}_{kj}, & v_{kij} \geq 0 \\ \hat{v}_{kj} - \tilde{v}_{kj}, & v_{kij} < 0 \end{Bmatrix} = \sum_{j=1}^{p+m} (v_{kij}\hat{v}_{kj} + |v_{kij}|\tilde{v}_{kj}). \label{eqn:zeta_upper_bound}
    \end{align}
\end{subequations}
\fi
Then, defining $\zeta^{(l)}_{ki} = \sum_{j=1}^{p+m} (v_{kij}\hat{v}_{kj} - |v_{kij}|\tilde{v}_{kj})$ and $\zeta^{(u)}_{ki} = \sum_{j=1}^{p+m} (v_{kij}\hat{v}_{kj} + |v_{kij}|\tilde{v}_{kj})$, a simple union bound shows that $\bbP[\zeta_{*i} \in [\zeta^{(l)}_{ki}, \zeta^{(u)}_{ki}]\ \forall i \in [1, p+m]] \geq 1 - \epsilon$.

As a final step before implementation, we must initialize $\hat{\zeta}_0$ and $P_0$. We can simply choose $\hat{\zeta}_0 = [\hat{\theta}(0), \hat{\vartheta}(0)^{-1}]^\top$ where the inverse is taken element-wise, using the initializations from the adaptive laws in \eqref{eqn:theta_adaptive_law_3}-\eqref{eqn:vartheta_adaptive_law_4}. Then, to initialize $P_0$, we require Assumptions \ref{asn:smid_theta_in_set} and \ref{asn:smid_lambda_in_set}, and choose $P_0$ such that the element-wise upper and lower bounds on $\zeta_*$ resulting from \eqref{eqn:cov_SVD}-\eqref{eqn:zeta_upper_bound} fully enclose $\Theta_0$ and $L_0$.